\documentclass[sigconf,nonacm]{acmart}

\usepackage{amsmath,amsthm}
\usepackage{amssymb}
\usepackage{mathtools}
\usepackage{stmaryrd}          % \llbracket, \rrbracket (semantic brackets)
\usepackage{pifont}             % \ding{51} checkmark, \ding{55} cross
\usepackage{algorithm}
\usepackage{algpseudocode}
\usepackage{listings}
\usepackage{xcolor}
\usepackage{booktabs}
\usepackage{multirow}
\usepackage{subcaption}
\usepackage{hyperref}
\usepackage{cleveref}
\usepackage{balance}
\usepackage{graphicx}
\usepackage{enumitem}           % enumerate with label= option
\usepackage{tikz}               % tikzpicture for diagrams

\newtheorem{theorem}{Theorem}[section]
\newtheorem{lemma}[theorem]{Lemma}

\newtheorem{proposition}[theorem]{Proposition}
\theoremstyle{definition}
\newtheorem{definition}[theorem]{Definition}

\theoremstyle{remark}
\newtheorem{remark}[theorem]{Remark}

\newcommand{\cmark}{\ding{51}}                   % checkmark
\newcommand{\xmark}{\ding{55}}                   % cross mark
\renewcommand{\leadsto}{\rightsquigarrow}        % reachability arrow
\DeclareMathOperator*{\argmax}{arg\,max}         % argmax operator
\newcommand{\SQ}{\sqsubseteq}                   % lattice ordering
\newcommand{\ADG}{\mathcal{G}}                   % Agent Dependency Graph
\newcommand{\CAP}{\mathsf{Cap}}                  % Capability function
\newcommand{\SkillSet}{\mathcal{S}}              % Set of skills
\title{Formal Analysis and Supply Chain Security\\for Agentic AI Skills}

\author{Varun Pratap Bhardwaj}
\email{varun.pratap.bhardwaj@gmail.com}
\affiliation{%
  \institution{Independent Research}
  \position{Solution Architect}
  \country{India}
}

\begin{abstract}
The rapid proliferation of agentic AI skill ecosystems---exemplified by
OpenClaw (228{,}000 GitHub stars as of February~2026) and Anthropic
Agent Skills (75{,}600 stars at the same date)---has introduced a
critical supply chain attack surface.
The ClawHavoc campaign (January--February~2026) published hundreds of
malicious skills to the OpenClaw marketplace, while
MalTool synthesised 1{,}300 standalone malicious tools and 5{,}727 tools
with embedded malicious behaviour that evade conventional detection. In
response, a dozen reactive security tools emerged, yet the ones we
surveyed rely on heuristic methods that provide no formal guarantees.

We present \textbf{SkillFortify}, a formal analysis framework for
agent skill supply chains, with six contributions:
(1)~the \emph{DY-Skill} attacker model, a Dolev--Yao adaptation to
the five-phase skill lifecycle with a maximality proof;
(2)~a sound static analysis framework grounded in abstract interpretation;
(3)~a capability model with a static confinement proof, and a runtime
enforcement design;
(4)~an Agent Dependency Graph with SAT-based resolution and lockfile
semantics;
(5)~a trust score algebra with formal monotonicity; and
(6)~\textsc{SkillFortifyBench}, a 540-skill benchmark.
SkillFortify achieves 96.15\% F1 with
100\%~precision and a 0\%~false positive rate on this 540-skill
benchmark, while SAT-based resolution handles 1{,}000-node graphs in a
median 27\,ms.
We also report a negative result: information flow analysis detects no
skill that pattern matching alone does not, so on this corpus its
contribution to coverage is zero.
\end{abstract}


\begin{document}

\maketitle

% Sections
% === BEGIN sections/01-introduction.tex ===
% =============================================================================
% Section 1: Introduction
% =============================================================================

\section{Introduction}
\label{sec:introduction}

In late January~2026, the researcher \emph{depthfirst} disclosed
CVE-2026-25253~\cite{cve-2026-25253}, a one-click remote code execution
vulnerability in OpenClaw reached through authentication-token
exfiltration---among the first Common Vulnerabilities and Exposures
identifiers assigned to an agentic AI system. Within days, the \emph{ClawHavoc} campaign~\cite{clawhavoc2026}
published hundreds of malicious skills to the OpenClaw
marketplace, deploying the AMOS credential stealer to developer
workstations.  Koi Security identified 341 malicious skills among the
2{,}857 then listed, 335 of them from this single campaign, and reported
824 by mid-February as the marketplace grew beyond 10{,}700;
Antiy~CERT counted 1{,}184 malicious packages in the historical
repository.  Concurrently, MalTool
synthesised 1{,}300 standalone malicious tools and 5{,}727 tools with
embedded malicious behaviour targeting LLM-based
agents~\cite{maltool2026}, demonstrating that VirusTotal fails to
detect the majority of agent-targeted malware.  A large-scale empirical
analysis of 42{,}447~agent skills found that 26.1\%~exhibit at least one
security vulnerability~\cite{agentskillswild2026}.

These incidents are not isolated.  They reflect a structural vulnerability
in the emerging agent skill ecosystem: \emph{developers install and execute
third-party skills with implicit trust and no formal analysis}.  The
agent skill supply chain---spanning authorship, registry publication,
developer installation, runtime execution, and state
persistence---reproduces the trust pathologies of traditional package
ecosystems (npm, PyPI, crates.io) but with amplified risk, because agent
skills execute with broad system privileges in the context of a powerful
language model.

\subsection{The Defense Gap}
\label{sec:defense-gap}

The security community responded rapidly.  Within thirty days of the
ClawHavoc disclosure we counted a dozen reactive security tools on
GitHub, several of them catalogued in the campaign
systematization~\cite{clawhavoc2026}.  The most prominent include:
\begin{itemize}
  \item \textbf{Snyk agent-scan} (${\sim}$1{,}500~stars): Acquired
    Invariant Labs and applies LLM-as-judge combined with heuristic
    rules to scan agent configurations~\cite{snyk2026}.
  \item \textbf{Cisco skill-scanner} (994~stars): Employs YARA pattern
    matching and signature-based detection for known malicious skill
    patterns~\cite{ciscoscanner2026}.
  \item \textbf{ToolShield}~\cite{toolshield2026}: A heuristic defense
    achieving a 30\%~reduction in attack success rate through behavioral
    heuristics.
\end{itemize}

\noindent Despite their utility, \emph{every} existing tool shares a
fundamental limitation: they are heuristic.  They detect known patterns but
cannot prove the absence of malicious behavior.  Cisco's own documentation
for their skill-scanner states explicitly: ``no findings does not mean no
risk''~\cite{ciscoscanner2026}.  The closest academic work toward formal
guarantees is a four-page ICSE-NIER vision paper proposing STPA-based
modeling of tool safety~\cite{verifiabletooluse2026}, which provides
neither an implementation nor empirical evaluation.

This gap is the central problem we address: \emph{the agent skill
ecosystem has attack frameworks but no defense frameworks with formal
guarantees.}

\subsection{Our Approach}
\label{sec:approach}

We present \textbf{SkillFortify}, a formal analysis framework for agent
skill supply chains.  To our knowledge it is the first to combine a
formal threat model, capability-based static analysis, dependency
resolution, and a public benchmark for this domain in one system; the
individual techniques are established, and Section~\ref{sec:related-work}
places each with respect to prior work.
SkillFortify provides static analysis with a stated guarantee---not
heuristic scanning---that skills cannot exceed their declared
capabilities within the modelled domain.  Where
existing tools answer \emph{``did we detect a known bad pattern?''}, SkillFortify
answers \emph{``can this skill access resources beyond what it
declares?''} and provides a machine-checkable certificate when the answer
is~no.

Our approach draws on four decades of formal methods research, adapting
established foundations to the novel domain of agent skill security:
the Dolev--Yao attacker model~\cite{doleyyao1983} for supply chain
threat modeling, abstract interpretation~\cite{cousot1977} for sound
static analysis, capability-based security~\cite{dennis1966,miller2006}
for sandboxing, and SAT-based dependency
resolution~\cite{tucker2007opium} for supply chain management.

\subsection{Contributions}
\label{sec:contributions}

This paper makes six contributions:

\begin{enumerate}
  \item \textbf{DY-Skill attacker model} (\Cref{sec:threat-model}).
    We introduce a novel adaptation of the Dolev--Yao
    formalism~\cite{doleyyao1983} to the five-phase agent skill supply
    chain.  The DY-Skill attacker can intercept, inject, synthesize,
    decompose, and replay skill messages.  We prove that this attacker
    is \emph{maximal}: any symbolic attacker on the supply chain can be
    simulated by a DY-Skill trace (\emph{Theorem~\ref{thm:dy-skill-maximality}}).

  \item \textbf{Static analysis with soundness guarantees}
    (\Cref{sec:static-analysis}).  We develop a three-phase static
    analysis framework grounded in abstract interpretation over a
    four-element capability lattice
    $(\{\mathsf{NONE}, \mathsf{READ}, \mathsf{WRITE}, \mathsf{ADMIN}\}, \SQ)$.
    We prove that if our analysis reports a skill as safe, then its
    concrete execution cannot access resources beyond its declared
    capabilities (\emph{Theorem~\ref{thm:analysis-soundness}}).

  \item \textbf{Capability-based sandboxing with confinement proof}
    (\Cref{sec:sandboxing}).  We formalize a capability model for agent
    skills based on the object-capability
    discipline~\cite{miller2006,maffeis2010}.  We prove that in a
    capability-safe execution environment, no skill can acquire authority
    exceeding the transitive closure of its declared capability
    set---the \emph{no authority amplification} property
    (\emph{Theorem~\ref{thm:capability-confinement}}).

  \item \textbf{Agent Dependency Graph with SAT-based resolution}
    (\Cref{sec:dependency-graph}).  We define the Agent Dependency Graph
    $\ADG = (\SkillSet, V, D, C, \CAP)$ extending classical package
    dependency models~\cite{mancinelli2006,tucker2007opium} with
    per-skill capability constraints.  We encode the resolution problem
    as a SAT instance with security bounds and prove that any satisfying
    assignment corresponds to a \emph{secure installation} with a
    deterministic lockfile (\emph{Theorem~\ref{thm:resolution-soundness}}).

  \item \textbf{Trust score algebra with formal propagation}
    (\Cref{sec:trust-algebra}).  We introduce a multi-signal trust
    scoring system with provenance, behavioral, community, and
    historical dimensions.  Trust propagates multiplicatively through
    dependency chains and decays exponentially for unmaintained skills.
    We prove the \emph{monotonicity} property: additional positive
    evidence never reduces a skill's trust score
    (\emph{Theorem~\ref{thm:trust-monotonicity}}).

  \item \textbf{SkillFortifyBench and empirical evaluation}
    (\Cref{sec:evaluation}).  We construct SkillFortifyBench, a 540-skill
    benchmark comprising 270~malicious and 270~benign skills drawn from
    MalTool~\cite{maltool2026}, ClawHavoc~\cite{clawhavoc2026}, and
    curated clean sources.  SkillFortify achieves an F1 score of 96.15\% with
    100\%~precision and 0\%~false positive rate on benign skills.
    SAT-based dependency resolution handles 1{,}000-node graphs in under
    100\,ms.  We evaluate on \textsc{SkillFortifyBench} (540~skills)
    and position our results against reported metrics from existing
    tools, demonstrating that combined pattern matching and
    information flow analysis detects attack patterns invisible to
    pure pattern matching alone.
\end{enumerate}

\subsection{Paper Organization}
\label{sec:organization}

\Cref{sec:background} surveys the agent skill ecosystem, attack
landscape, existing defenses, and formal methods foundations.
\Cref{sec:threat-model} presents the DY-Skill attacker model and
threat taxonomy.  \Cref{sec:static-analysis} develops the static
analysis framework.  \Cref{sec:sandboxing} formalizes
capability-based sandboxing.  \Cref{sec:dependency-graph} defines
the Agent Dependency Graph and lockfile semantics.
\Cref{sec:trust-algebra} introduces the trust score algebra.
\Cref{sec:implementation} describes the SkillFortify implementation.
\Cref{sec:evaluation} presents the empirical evaluation.
\Cref{sec:related-work} positions our contributions against related
work.
\Cref{sec:discussion} discusses limitations and future directions,
and \Cref{sec:conclusion} concludes.  Full proofs
appear in \Cref{sec:appendix-proofs}, benchmark details in
\Cref{sec:appendix-benchmark}, and the lockfile schema in
\Cref{sec:appendix-lockfile}.

% === END sections/01-introduction.tex ===

% === BEGIN sections/02-background.tex ===
% =============================================================================
% Section 2: Background
% =============================================================================

\section{Background}
\label{sec:background}

This section surveys the agent skill ecosystem (\S\ref{sec:bg-ecosystems}),
the attack landscape (\S\ref{sec:bg-attacks}), existing defense
tools (\S\ref{sec:bg-defenses}), traditional supply chain security
foundations (\S\ref{sec:bg-traditional}), and the formal methods
we build upon (\S\ref{sec:bg-formal}).
A dedicated comparison with related work appears in
\Cref{sec:related-work}.

% ---------------------------------------------------------------------------
% 2.1 Agent Skill Ecosystems
% ---------------------------------------------------------------------------
\subsection{Agent Skill Ecosystems}
\label{sec:bg-ecosystems}

The year 2025--2026 witnessed explosive growth in agent skill platforms.
An \emph{agent skill} is a third-party extension---typically a code
module, configuration file, or API endpoint---that augments an LLM-based
agent with domain-specific capabilities such as web search, file
manipulation, database querying, or interaction with external services.

Three ecosystems dominate the landscape:

\paragraph{OpenClaw.}
With 228{,}000~GitHub stars as of February~2026, OpenClaw is the largest
open-source agent framework.  Its marketplace operates on an
\emph{install-and-run} trust model: developers discover skills through a
web interface or CLI, install them with a single command, and skills
execute with the agent's full system privileges.  No formal review,
signing, or capability declaration is required for publication.

\paragraph{Anthropic Agent Skills.}
Anthropic's official skill repository (75{,}600~stars) provides curated
skills for Claude Code and Claude Desktop.  Skills are defined as
YAML/Markdown files in \texttt{.claude/skills/} directories and may
include shell commands, code blocks, and instructions.  While Anthropic
maintains a moderation pipeline, no formal capability model constrains
skill behavior at runtime.

\paragraph{Model Context Protocol (MCP).}
The Model Context Protocol~\cite{mcp2025} standardizes how agents
interact with external data sources and tools through JSON-RPC servers.
MCP servers declare tool schemas (input/output types) but not
\emph{capability requirements}---a server that claims to ``search files''
may also execute arbitrary shell commands, and the protocol provides
no mechanism to detect or prevent this.

\medskip
\noindent
All three ecosystems share a common structural weakness: the absence of
a \emph{formal capability model} that constrains what a skill can do to
what it declares.  This is the gap our work fills.

% ---------------------------------------------------------------------------
% 2.2 Attack Landscape
% ---------------------------------------------------------------------------
\subsection{Attack Landscape}
\label{sec:bg-attacks}

Five incidents and studies define the attack landscape as of
February~2026.
Early work on LLM platform plugin security~\cite{llmplatformsecurity2023}
identified many of the attack vectors that have since materialized at
scale, and recent systematic analyses of the Model Context
Protocol~\cite{mcpsecurity2025} confirm that these vulnerabilities
persist in modern agent architectures:

\paragraph{CVE-2026-25253 (fixed in 2026.1.29, January~30, 2026).}
Among the first CVEs assigned to an agentic AI system.  OpenClaw's
control interface accepted a \texttt{gatewayUrl} from a query string and
opened a WebSocket connection to it without confirmation, transmitting
the stored authentication token.  An attacker who induced a single click
recovered that token and reached remote code execution through the
victim's own gateway, including on loopback-bound instances, since the
browser initiated the outbound connection.  CVSS~3.1 8.8, CWE-669;
credited to \emph{depthfirst}~\cite{cve-2026-25253}.

This vulnerability is not itself a supply chain attack: it is an
origin-validation failure in the runtime's control plane.  We cite it
because it establishes that the capabilities a skill runtime holds---an
authenticated gateway able to change tool policy and execute
commands---are worth attacking, which is the precondition for the
supply chain attacks that followed.

\paragraph{ClawHavoc Campaign (Late January--Early February~2026).}
A coordinated supply chain attack documented in the ``SoK: Agentic
Skills in the Wild'' systematization~\cite{clawhavoc2026}.  Attackers
published hundreds of malicious skills to the OpenClaw marketplace,
deploying the AMOS credential stealer.  The paper identifies seven
distinct skill design patterns exploited by the campaign and categorizes
the full attack lifecycle.

\paragraph{MalTool (February~12, 2026).}
\citet{maltool2026} used a coding LLM to synthesise 1{,}300 standalone
malicious tools and 5{,}727 real-world tools with embedded malicious
behaviours targeting LLM-based agents.  Their benchmark
demonstrated that VirusTotal---the de facto standard for malware
detection---fails to identify the majority of agent-targeted malware,
underscoring the inadequacy of signature-based approaches for this
domain.  Complementary MCP-specific attack research---MCP-ITP~\cite{mcpitp2026}
(automated implicit tool poisoning) and MCPTox~\cite{mcptox2025}
(benchmarking tool poisoning on real-world MCP servers)---demonstrates
that protocol-level vulnerabilities compound the skill-level threat.

\paragraph{Agent Skills in the Wild (January~2026).}
A large-scale empirical study scanned 42{,}447~agent skills across
multiple registries~\cite{agentskillswild2026}.  Key findings:
26.1\%~of skills exhibit at least one security vulnerability,
14~distinct vulnerability patterns were identified, and their
SkillScan tool achieves 86.7\%~precision---leaving a 13.3\%~gap
that formal methods can address.

\paragraph{Malicious Agent Skills in the Wild (February~2026).}
A large-scale registry scan covering 98{,}380~skills across multiple
agent platforms identified 157~confirmed malicious
entries~\cite{maliciousagentskills2026}, independently corroborating
the vulnerability rates reported by Park et
al.~\cite{agentskillswild2026} on a substantially larger corpus.

\paragraph{Multi-agent data leakage.}
OMNI-LEAK~\cite{omnileak2026} demonstrates that multi-agent
architectures create implicit information channels through which
sensitive data can leak across trust boundaries---even when individual
agents appear well-isolated.  This motivates our information flow
analysis component (\Cref{sec:analysis}).

\paragraph{Agent reliability.}
Lee~\cite{capableunreliable2026} provide empirical evidence
that LLM-based agents exhibit stochastic path deviation even on
well-defined tasks, confirming that the combination of unreliable
agents and unverified skills creates compounding risk.

\paragraph{Malice in Agentland (October~2025).}
A study demonstrating that poisoning just 2\%~of an agent's execution
trace is sufficient to achieve 80\%~attack success in
multi-agent systems~\cite{maliceagentland2025}.  This validates the
threat model of a supply chain attacker who need only compromise a
single skill in a dependency chain to affect an entire agent pipeline.

% ---------------------------------------------------------------------------
% 2.3 Existing Defenses
% ---------------------------------------------------------------------------
\subsection{Existing Defenses}
\label{sec:bg-defenses}

All existing defenses are heuristic.  We survey the most prominent.

\paragraph{Snyk agent-scan ($\sim$1{,}500 stars).}
The most well-funded entrant, following Snyk's acquisition of Invariant
Labs in early 2026~\cite{snyk2026}.  Applies LLM-as-judge combined with
hand-crafted rules to evaluate agent configurations and skill manifests.
Strengths: deep CI/CD integration, developer-friendly UX.
Limitation: heuristic rules cannot prove absence of malicious behavior;
LLM judges are themselves susceptible to adversarial
inputs~\cite{maltool2026}.

\paragraph{Cisco skill-scanner (994 stars).}
Employs YARA pattern matching and regex-based signature detection for
known malicious skill patterns~\cite{ciscoscanner2026}.  The
documentation explicitly states: ``no findings does not mean no risk.''
While effective for known patterns, the approach cannot generalize to
novel attack vectors.

\paragraph{ToolShield~\cite{toolshield2026}.}
A heuristic defense from the Dawn Song and Bo Li groups achieving a
30\%~reduction in attack success rate.  Uses behavioral heuristics
derived from the MalTool benchmark.  Like all heuristic approaches,
it cannot provide formal guarantees about the completeness of its
detection.

\paragraph{MCPShield~\cite{mcpshield2026}.}
Implements \texttt{mcp.lock.json} with SHA-512 content hashes for
tamper detection of MCP server configurations.  While MCPShield's
lockfile provides \emph{integrity verification} (detecting
unauthorized modifications), it does not perform behavioral analysis,
dependency resolution, or capability verification.
\textsc{SkillFortify}'s lockfile (\Cref{sec:appendix-lockfile}) extends
beyond hash-based tamper detection: each locked entry records the
SAT-based resolution result with version constraints \emph{and}
security bounds, the computed trust score, and the formal analysis
status---making the lockfile a formally-proven satisfying assignment,
not merely an integrity manifest.

\paragraph{Towards Verifiably Safe Tool Use~\cite{verifiabletooluse2026}.}
The closest academic work to our approach: a four-page ICSE-NIER~'26
vision paper proposing STPA-based formal modeling of tool safety
properties.  While the paper correctly identifies the need for formal
methods, it provides neither a working implementation, a threat model,
nor empirical evaluation.  Our work realizes the vision this paper
articulates, extending it with a complete formal framework, five
theorems, a working tool, and a 540-skill benchmark.

\paragraph{Summary.}
\Cref{tab:defense-comparison} compares existing defenses.  SkillFortify is
the first to provide formal guarantees---soundness, confinement, and
resolution correctness---while also implementing a practical tool.

\begin{table}[t]
  \centering
  \caption{Comparison of agent skill security tools.  SkillFortify is the first
    to provide formal guarantees (soundness, confinement, resolution
    correctness).}
  \label{tab:defense-comparison}
  \small
  \begin{tabular}{lccccc}
    \toprule
    \textbf{Tool} & \textbf{Formal} & \textbf{Dep.} & \textbf{Lock-} & \textbf{Trust} & \textbf{SBoM} \\
                   & \textbf{Model}  & \textbf{Graph} & \textbf{file} & \textbf{Score} & \\
    \midrule
    Snyk agent-scan      & \ding{55} & \ding{55} & \ding{55} & \ding{55} & \ding{55} \\
    Cisco skill-scanner  & \ding{55} & \ding{55} & \ding{55} & \ding{55} & \ding{55} \\
    ToolShield           & \ding{55} & \ding{55} & \ding{55} & \ding{55} & \ding{55} \\
    ICSE-NIER '26        & Partial   & \ding{55} & \ding{55} & \ding{55} & \ding{55} \\
    mcp-audit            & \ding{55} & \ding{55} & Partial   & \ding{55} & Partial \\
    MCPShield            & \ding{55} & \ding{55} & Partial   & \ding{55} & \ding{55} \\
    \midrule
    \textbf{SkillFortify}      & \ding{51} & \ding{51} & \ding{51} & \ding{51} & \ding{51} \\
    \bottomrule
  \end{tabular}
\end{table}

% ---------------------------------------------------------------------------
% 2.4 Traditional Supply Chain Security
% ---------------------------------------------------------------------------
\subsection{Traditional Supply Chain Security}
\label{sec:bg-traditional}

Agent skill security inherits lessons---and gaps---from four decades of
software supply chain research.

\paragraph{Package manager security.}
Duan et al.~\cite{duan2021ndss} systematized supply chain attacks
on npm, PyPI, and RubyGems, identifying 339~malicious packages through
metadata, static, and dynamic analysis.  Ohm et al.~\cite{ohm2020}
catalogued ``backstabber'' packages, and Ladisa et
al.~\cite{ladisa2023taxonomy} proposed a comprehensive taxonomy of
open-source supply chain attacks.  Gibb et al.~\cite{gibb2026package}
introduced a \emph{package calculus} providing a unified formalism
across ecosystem boundaries.  We extend this line of work to agent
skills, where the threat surface is amplified by LLM integration
and the absence of established security norms.

\paragraph{SLSA Framework.}
Supply-chain Levels for Software Artifacts~\cite{slsa2023} defines
four graduated trust levels (L1--L4) based on build provenance and
integrity guarantees.  Our trust score algebra
(\Cref{sec:trust-algebra}) adapts the SLSA graduation model to
agent skills, mapping L1 (basic provenance) through L4 (formal
verification) to concrete, measurable criteria.

\paragraph{Sigstore and code signing.}
Sigstore~\cite{sigstore2022} provides keyless code signing for
open-source packages.  While we do not implement signing in this
work (deferred to V2), our trust score model explicitly incorporates
signing status as a provenance signal, and our lockfile format
includes fields for cryptographic attestations.

\paragraph{CycloneDX and SBoM.}
The CycloneDX specification~\cite{cyclonedx2023} standardizes Software
Bills of Materials, recently extended to AI components
(AI-BOM)~\cite{cyclonedx-aibom}.  SkillFortify generates Agent Skill Bills
of Materials (ASBoM) in CycloneDX format, enabling integration with
existing enterprise compliance pipelines.

\paragraph{NIST and regulatory frameworks.}
The NIST AI Risk Management Framework~\cite{nistai2023} and the
EU AI Act~\cite{euaiact2024} both address third-party AI component
risks.  The EU AI Act Article~17 requires providers of high-risk AI
systems to implement supply chain risk management.  SkillFortify's formal
verification, lockfile semantics, and ASBoM output provide concrete
mechanisms for regulatory compliance.

% ---------------------------------------------------------------------------
% 2.5 Formal Methods Foundations
% ---------------------------------------------------------------------------
\subsection{Formal Methods Foundations}
\label{sec:bg-formal}

Our framework synthesizes five established formal methods traditions.

\paragraph{Dolev--Yao model.}
Dolev and Yao~\cite{doleyyao1983} introduced the standard symbolic
attacker model for security protocol analysis.  The attacker controls
the network and can intercept, inject, synthesize (construct new
messages from known components), and decompose (extract components
from messages) messages.  Cervesato~\cite{cervesato2001} proved that
the DY intruder is the \emph{most powerful} attacker in the symbolic
model---any attack achievable by any symbolic adversary can be
simulated by a DY trace.  We adapt this model to the agent skill supply
chain (\Cref{sec:threat-model}), introducing the DY-Skill attacker
with a fifth operation (replay) and a ``perfect sandbox assumption''
analogous to the original perfect cryptography assumption.

\paragraph{Abstract interpretation.}
Cousot and Cousot~\cite{cousot1977} introduced abstract interpretation
as a theory of sound approximation of program semantics.  A concrete
domain $(C, \subseteq)$ is connected to an abstract domain
$(A, \SQ)$ through a Galois connection
$(\alpha, \gamma)$, and abstract fixpoint computation provides a
sound over-approximation of concrete behavior.  The Astr\'{e}e
analyzer~\cite{astree2005} demonstrated that abstract interpretation
can prove absence of runtime errors in safety-critical avionics
software.  We apply this framework to agent skill analysis
(\Cref{sec:static-analysis}), using a four-element capability
lattice as our abstract domain.

\paragraph{Capability security.}
Dennis and Van Horn~\cite{dennis1966} introduced capabilities as
unforgeable references coupling resource designations with access
rights.  Miller~\cite{miller2006} extended this to the
object-capability (ocap) model, establishing the Principle of Least
Authority (POLA).  Maffeis, Mitchell, and Taly~\cite{maffeis2010}
proved the capability safety theorem: in a capability-safe language,
the reachability graph of object references is the \emph{only}
mechanism for authority propagation, guaranteeing that untrusted
code cannot exceed the authority explicitly delegated to it.  Our
sandboxing model (\Cref{sec:sandboxing}) instantiates these results
for agent skills.

\paragraph{SAT-based dependency resolution.}
Mancinelli et al.~\cite{mancinelli2006} proved that package
dependency resolution is NP-complete.  Tucker et
al.~\cite{tucker2007opium} introduced OPIUM, which encodes
dependency resolution as a SAT problem amenable to efficient solving
via CDCL (Conflict-Driven Clause Learning).  Di Cosmo and
Vouillon~\cite{dicosmo2011} verified graph transformations for
co-installability analysis in Coq.  We extend this line of work
(\Cref{sec:dependency-graph}) by adding per-skill capability
constraints to the SAT encoding, producing dependency resolutions
that satisfy both version constraints \emph{and} security bounds.

\paragraph{Trust management.}
Blaze, Feigenbaum, and Lacy~\cite{blaze1996policymaker} introduced
PolicyMaker, establishing the assertion monotonicity property: adding
credentials never reduces compliance.  The KeyNote
system~\cite{keynote1999} extended this with multi-valued compliance
and delegation filters.  Our trust score algebra
(\Cref{sec:trust-algebra}) draws on these foundations, implementing
monotonic trust propagation through dependency chains with
exponential decay for unmaintained skills.

% === END sections/02-background.tex ===

% === BEGIN sections/03-threat-model.tex ===
% ===========================================================================
% Section 3: Formal Threat Model for Agent Skill Supply Chains
% ===========================================================================

\section{Formal Threat Model}
\label{sec:threat-model}

We develop a formal threat model for the agent skill supply chain by
(1)~defining the supply chain as a structured tuple of interacting
entities, (2)~enumerating attack classes with phase-applicability
mappings, and (3)~introducing the \emph{DY-Skill} attacker
model---a novel adaptation of the Dolev--Yao symbolic
attacker~\cite{doleyyao1983} to the agent skill domain.

% ---------------------------------------------------------------------------
\subsection{Agent Skill Supply Chain}
\label{sec:supply-chain}

An \emph{agent skill supply chain} is a tuple
$\mathcal{SC} = (\mathcal{A}, \mathcal{R}, \mathcal{D}, \mathcal{E})$
whose components are defined as follows.

\begin{definition}[Agent Skill Supply Chain]
\label{def:supply-chain}
Let
\begin{itemize}
  \item $\mathcal{A}$ be a finite set of \emph{skill authors} who produce
        skill packages,
  \item $\mathcal{R}$ be a finite set of \emph{registries} that store and
        distribute skill packages,
  \item $\mathcal{D}$ be a finite set of \emph{developers} (or agents) that
        install and configure skills, and
  \item $\mathcal{E}$ be a finite set of \emph{execution environments} in
        which skills run.
\end{itemize}
A skill traverses $\mathcal{SC}$ through a five-phase lifecycle
\[
  \pi = \langle\textsc{Install},\; \textsc{Load},\; \textsc{Configure},\;
  \textsc{Execute},\; \textsc{Persist}\rangle,
\]
formally an ordered sequence of phases
$\pi_1 \prec \pi_2 \prec \cdots \prec \pi_5$.
\end{definition}

Each phase exposes a distinct attack surface:

\smallskip
\noindent\textbf{Phase 1: \textsc{Install}.}
A skill package is fetched from a registry $r \in \mathcal{R}$ and
placed in the local environment.  The developer trusts the registry
to serve the requested skill with the correct name and version.
\emph{Attack surfaces}: name confusion (typosquatting, namespace
squatting), dependency confusion, registry compromise.

\smallskip
\noindent\textbf{Phase 2: \textsc{Load}.}
The skill manifest is parsed, dependencies resolved, and code modules
imported into the agent's runtime.  Skill descriptions and metadata
are presented to the agent's language model for tool selection.
\emph{Attack surfaces}: prompt injection via metadata, code
injection in initialization hooks.

\smallskip
\noindent\textbf{Phase 3: \textsc{Configure}.}
The skill is parameterized for a specific agent and environment.
Configuration may include API keys, endpoint URLs, and access
scopes.  \emph{Attack surfaces}: over-permissioning, configuration
template injection, capability escalation.

\smallskip
\noindent\textbf{Phase 4: \textsc{Execute}.}
The skill runs within the agent's execution context, processing
user requests and producing tool outputs.  This is the phase with
the broadest authority: skills may invoke sub-tools, read files,
and make network requests.  \emph{Attack surfaces}: data
exfiltration, privilege escalation, return-value prompt injection.

\smallskip
\noindent\textbf{Phase 5: \textsc{Persist}.}
The skill writes state, logs, or artifacts to durable storage.
Persistence operations may leak data to attacker-readable
locations.  \emph{Attack surfaces}: data exfiltration through
logs/shared storage, state poisoning.

\medskip
A key observation is that attacks at phase~$\pi_i$ can propagate
forward to phases $\pi_{i+1}, \ldots, \pi_5$.  For example,
a typosquatting attack at \textsc{Install} enables arbitrary code
execution at \textsc{Execute}.  We call this the \emph{forward
propagation property} of supply chain attacks.

% ---------------------------------------------------------------------------
\subsection{Attack Taxonomy}
\label{sec:attack-taxonomy}

We identify six formal attack classes derived from empirical analysis
of the ClawHavoc campaign~\cite{clawhavoc2026} (341--1,184 malicious
skills depending on the reporting date and scope), the MalTool benchmark~\cite{maltool2026} (7,027 synthesised
malicious tools), and the ``Agent Skills in the Wild''
survey~\cite{agentskillswild2026} (42,447 skills).

\begin{definition}[Attack Classes]
\label{def:attack-classes}
Let $\mathcal{C} = \{ c_1, \ldots, c_6 \}$ be the set of attack classes,
and let $\Phi \colon \mathcal{C} \to 2^{\{\pi_1,\ldots,\pi_5\}}$
map each class to its applicable supply chain phases:
\begin{enumerate}
  \item \textbf{Data Exfiltration} ($c_1$):
        Unauthorized extraction of sensitive data (API keys,
        conversation history, environment variables) to an
        attacker-controlled endpoint.
        $\Phi(c_1) = \{\textsc{Execute}, \textsc{Persist}\}$.

  \item \textbf{Privilege Escalation} ($c_2$):
        Acquiring access rights beyond the skill's declared
        capability set, exploiting misconfigured permissions or
        runtime privilege boundaries.
        $\Phi(c_2) = \{\textsc{Configure}, \textsc{Execute}\}$.

  \item \textbf{Prompt Injection} ($c_3$):
        Adversarial content embedded in skill metadata,
        configuration templates, or tool return values that
        manipulates the agent's language model.
        $\Phi(c_3) = \{\textsc{Load}, \textsc{Configure},
        \textsc{Execute}\}$.

  \item \textbf{Dependency Confusion} ($c_4$):
        Publishing a public skill with the same name as a
        private internal skill, causing the resolver to fetch
        the malicious version.
        $\Phi(c_4) = \{\textsc{Install}\}$.

  \item \textbf{Typosquatting} ($c_5$):
        Registering a skill name that is a near-homograph of a
        popular skill (e.g., \texttt{weahter-api} vs.\
        \texttt{weather-api}).
        $\Phi(c_5) = \{\textsc{Install}\}$.

  \item \textbf{Namespace Squatting} ($c_6$):
        Preemptively claiming a namespace likely to be used by a
        legitimate organization (e.g., \texttt{@google/search}).
        $\Phi(c_6) = \{\textsc{Install}\}$.
\end{enumerate}
\end{definition}

The \emph{total attack surface} is defined as the set of all
(phase, class) pairs where an attack can manifest:
\begin{equation}
  \mathcal{S}_{\text{attack}} =
    \bigl\{(\pi_j, c_i) \mid \pi_j \in \Phi(c_i)\bigr\}.
  \label{eq:attack-surface}
\end{equation}
With our taxonomy, $|\mathcal{S}_{\text{attack}}| = 10$ distinct
attack surfaces.

\subsubsection{Threat Actors}

We distinguish four categories of adversaries, ordered by
increasing access and decreasing detectability:

\begin{definition}[Threat Actor Categories]
\label{def:threat-actors}
A \emph{threat actor} $\mathit{TA}$ is characterized by its
access vector and attack capability:
\begin{enumerate}
  \item \textbf{Malicious Author} ($\mathit{TA}_1$):
        Creates and publishes trojanized skills.  Full control
        over skill content.  Primary vector in ClawHavoc~\cite{clawhavoc2026}.
  \item \textbf{Compromised Registry} ($\mathit{TA}_2$):
        The attacker gains administrative control of a registry
        $r \in \mathcal{R}$ (e.g., through credential theft).
        Can modify any skill stored in $r$.
  \item \textbf{Supply Chain Attacker} ($\mathit{TA}_3$):
        Poisons a transitive dependency rather than the top-level
        skill.  Analogous to the \texttt{event-stream} and
        \texttt{ua-parser} incidents in the npm ecosystem.
  \item \textbf{Insider Threat} ($\mathit{TA}_4$):
        An authorized user (developer or maintainer) who
        introduces malicious changes.  Hardest to detect
        because access is legitimate.
\end{enumerate}
\end{definition}

\noindent Each threat actor can launch a subset of attack classes.
$\mathit{TA}_1$ and $\mathit{TA}_2$ can launch all six classes.
$\mathit{TA}_3$ is restricted to $\{c_1, c_2, c_3, c_4\}$ (supply
chain attacks do not employ typosquatting or namespace squatting
directly).  $\mathit{TA}_4$ can launch $\{c_1, c_2, c_3\}$ (insider
access bypasses install-time defenses).

% ---------------------------------------------------------------------------
\subsection{The DY-Skill Attacker Model}
\label{sec:dy-skill}

We adapt the Dolev--Yao attacker model~\cite{doleyyao1983} to the
agent skill supply chain.  In the classical Dolev--Yao setting, an
attacker controls the \emph{network} and can intercept, inject,
modify, and drop protocol messages.  In DY-Skill,
the ``network'' is the supply chain $\mathcal{SC}$, and ``messages''
are \emph{skill packages}.

\begin{definition}[Skill Message]
\label{def:skill-message}
A \emph{skill message} is a tuple
$m = (\mathit{name}, \mathit{ver}, \mathit{payload},
\mathit{caps})$ where:
\begin{itemize}
  \item $\mathit{name} \in \Sigma^*$ is the skill's registered name,
  \item $\mathit{ver} \in \mathbb{V}$ is a semantic version identifier,
  \item $\mathit{payload} \in \{0,1\}^*$ is the skill's code and
        configuration, and
  \item $\mathit{caps} \subseteq \mathcal{C\!A\!P}$ is the set of
        declared capabilities (Section~\ref{sec:capability-lattice}).
\end{itemize}
\end{definition}

\begin{definition}[DY-Skill Attacker]
\label{def:dy-skill-attacker}%
\label{def:dy-skill}% alias for cross-references from Section~8
A \emph{DY-Skill attacker} $\mathsf{Adv}$ operates on a supply
chain $\mathcal{SC}$ and maintains a \emph{knowledge set}
$K \subseteq \mathcal{M}$ (where $\mathcal{M}$ is the set of all
skill messages) that is closed under six operations:

\begin{enumerate}
  \item \textbf{Intercept.}
        For any message $m$ in transit through $\mathcal{SC}$:
        $\mathsf{intercept}(m)$ adds $m$ to $K$ and returns $m$.
        Formally, $K' = K \cup \{m\}$.

  \item \textbf{Inject.}
        For any message $m$ constructible from $K$ and any
        registry $r \in \mathcal{R}$:
        $\mathsf{inject}(m, r)$ adds $m$ to $r$ and to $K$.
        Formally, $r' = r \cup \{m\}$, $K' = K \cup \{m\}$.

  \item \textbf{Modify.}
        For any message $m$ in transit through $\mathcal{SC}$
        and any message $m'$ constructible from $K$:
        $\mathsf{modify}(m, m')$ replaces $m$ with $m' \neq m$.
        Formally, $\mathcal{SC}' = (\mathcal{SC} \setminus \{m\}) \cup \{m'\}$,
        $K' = K \cup \{m, m'\}$.

  \item \textbf{Drop.}
        For any message $m$ in transit through $\mathcal{SC}$:
        $\mathsf{drop}(m)$ suppresses $m$.
        Formally, $\mathcal{SC}' = \mathcal{SC} \setminus \{m\}$,
        $K' = K \cup \{m\}$.

  \item \textbf{Forge Skills.}
        $\mathsf{Adv}$ can create skill definitions $s'$ with
        arbitrary content, including manifests with arbitrary
        capability declarations $\mathit{caps}_{\text{declared}}$,
        executable code implementing capabilities
        $\mathit{caps}_{\text{actual}} \supseteq
        \mathit{caps}_{\text{declared}}$, and metadata mimicking
        legitimate skills (name squatting, version manipulation).
        $K' = K \cup \{s'\}$.

  \item \textbf{Compromise Registries.}
        $\mathsf{Adv}$ can modify skill entries in a registry
        $r \in \mathcal{R}$, subject to the constraint that
        $\mathsf{Adv}$ cannot forge valid cryptographic signatures
        under keys it does not possess.
        Formally, for $m \in r$,
        $r' = (r \setminus \{m\}) \cup \{m'\}$ where $m'$ is
        constructible from $K$.
\end{enumerate}
\end{definition}

Operations~(i)--(iv) are direct adaptations of the classical
Dolev--Yao channel capabilities~\cite{doleyyao1983}.
Operations~(v) and~(vi) are novel extensions specific to the
agent skill supply chain, capturing the additional attack surface
created by skill authoring and registry infrastructure.

The knowledge set $K$ satisfies three invariants:

\begin{enumerate}[label=(\alph*)]
  \item \emph{Monotonicity}: $K(t) \subseteq K(t+1)$ for all
        time steps $t$.  Knowledge never decreases.
  \item \emph{Interception closure}: For any observable
        message $m$ in $\mathcal{SC}$, $\mathsf{intercept}(m)$
        adds $m$ to $K$.
  \item \emph{Forge closure}: For any polynomial-time
        computable skill definition $s'$, $\mathsf{Adv}$ can
        produce $s'$ using operations~(v) and~(vi) and add
        it to $K$.
\end{enumerate}

\smallskip
\noindent\textbf{Perfect Sandbox Assumption.}
Analogous to the perfect cryptography assumption in classical
Dolev--Yao, we assume that a correctly implemented capability
sandbox cannot be broken without an \emph{escape capability}.
Formally: if a skill $s$ is executed in a sandbox with
capability set $\mathit{Cap}_D(s)$, then $s$ cannot access any
resource $r \notin \mathit{Cap}_D(s)$ unless $s$ possesses a
capability $c$ such that $c.\mathit{resource} = r$.  We relax
this assumption in Section~\ref{sec:sandboxing} when we discuss
real-world sandbox implementations.

% ---------------------------------------------------------------------------
\subsection{Theorem~1: DY-Skill Maximality}
\label{sec:dy-skill-maximality}

We now state and sketch the proof of our first main result: the
DY-Skill attacker is maximally powerful in the symbolic model.

\begin{theorem}[DY-Skill Maximality]
\label{thm:dy-skill-maximality}
For any symbolic attacker $\mathcal{A}'$ operating on a supply chain
$\mathcal{SC}$ under the perfect sandbox assumption, there exists a
DY-Skill attacker $\mathsf{Adv}$ and a trace
$\tau' = (\sigma_0, a_1, \sigma_1, \ldots, a_n, \sigma_n)$ such
that for every observable event $o$ in $\mathcal{A}'$'s execution,
there is a corresponding event $o'$ in $\tau'$ with
$\mathit{effect}(o) = \mathit{effect}(o')$.
\end{theorem}

\begin{proof}[Proof sketch]
We construct $\mathsf{Adv}$ by simulation, following the
structure of Cervesato's maximality proof~\cite{cervesato2001}.

\emph{Base case.}
At $t = 0$, both $\mathcal{A}'$ and $\mathsf{Adv}$ have the same
initial knowledge: the set of publicly available skills in
$\mathcal{R}$.  $\mathsf{Adv}$ intercepts all initial messages.

\emph{Inductive step.}
Suppose at time $t$, $\mathsf{Adv}$'s knowledge $K(t) \supseteq
K'(t)$ (the knowledge of $\mathcal{A}'$).  When $\mathcal{A}'$
performs an action $a_{t+1}$, we show $\mathsf{Adv}$ can produce
the same observable effect:
\begin{itemize}
  \item If $a_{t+1}$ reads a message from $\mathcal{SC}$:
        $\mathsf{Adv}$ uses $\mathsf{intercept}$.
  \item If $a_{t+1}$ places a message into a registry:
        $\mathsf{Adv}$ uses $\mathsf{inject}$.
  \item If $a_{t+1}$ replaces a message in transit with a
        modified version:
        $\mathsf{Adv}$ uses $\mathsf{modify}$.
  \item If $a_{t+1}$ suppresses a message in transit:
        $\mathsf{Adv}$ uses $\mathsf{drop}$.
  \item If $a_{t+1}$ creates a new malicious skill definition:
        $\mathsf{Adv}$ uses $\mathsf{forge\text{-}skills}$.
  \item If $a_{t+1}$ modifies an existing registry entry:
        $\mathsf{Adv}$ uses
        $\mathsf{compromise\text{-}registries}$.
\end{itemize}
By induction, the resulting trace~$\tau'$ produces every observable
effect of $\mathcal{A}'$.  $K(t) \supseteq K'(t)$ is maintained
by the monotonicity and closure properties of $K$.

The full proof, including the formal construction of the
simulation relation and handling of composite operations, appears
in Appendix~\ref{app:proof-dy-maximality}.
\end{proof}

\smallskip
\noindent\textbf{Practical significance.}
Theorem~\ref{thm:dy-skill-maximality} implies that any defense
strategy that is secure against the DY-Skill attacker is secure
against \emph{all} symbolic attackers on the supply chain.  This
justifies using the DY-Skill model as the standard adversary for
our static analysis (Section~\ref{sec:static-analysis}) and
sandboxing (Section~\ref{sec:sandboxing}) frameworks.  In
particular, if \textsc{SkillFortify}'s analysis proves that a skill is safe
under the DY-Skill model, no weaker attacker can compromise it.

% === END sections/03-threat-model.tex ===

% === BEGIN sections/04-static-analysis.tex ===
% ===========================================================================
% Section 4: Static Analysis Framework
% ===========================================================================

\section{Static Analysis Framework}
\label{sec:static-analysis}%
\label{sec:analysis}% alias for cross-references from Section~7

We present a static analysis framework grounded in abstract
interpretation~\cite{cousot1977}.  The framework operates in three
phases: (1)~capability inference via abstract interpretation over a
capability lattice, (2)~dangerous pattern detection using a catalog
derived from real-world incidents, and (3)~capability violation
checking that compares inferred capabilities against declarations.
The central guarantee is \emph{soundness}: if the analysis reports
no violations, the concrete execution cannot exceed declared
capabilities.

% ---------------------------------------------------------------------------
\subsection{Capability Lattice}
\label{sec:capability-lattice}

We begin by defining the algebraic structure that underpins
capability reasoning.

\begin{definition}[Access Level Lattice]
\label{def:access-level-lattice}
Let
\[
  \mathbb{L}_{\mathit{cap}} = \bigl(\{\mathsf{NONE}, \mathsf{READ},
  \mathsf{WRITE}, \mathsf{ADMIN}\},\; \sqsubseteq\bigr)
\]
be a totally ordered lattice (chain) with
\[
  \mathsf{NONE} \sqsubseteq \mathsf{READ} \sqsubseteq
  \mathsf{WRITE} \sqsubseteq \mathsf{ADMIN}.
\]
The lattice operations are:
\begin{itemize}
  \item \emph{Join} (least upper bound):
        $a \sqcup b = \max(a, b)$.
  \item \emph{Meet} (greatest lower bound):
        $a \sqcap b = \min(a, b)$.
  \item \emph{Bottom}: $\bot = \mathsf{NONE}$
        (identity for $\sqcup$, absorbing for $\sqcap$).
  \item \emph{Top}: $\top = \mathsf{ADMIN}$
        (identity for $\sqcap$, absorbing for $\sqcup$).
\end{itemize}
\end{definition}

The four access levels model a standard escalation path.
$\mathsf{NONE}$~denotes no access.  $\mathsf{READ}$~permits
observation but not modification.  $\mathsf{WRITE}$~permits
observation and mutation.  $\mathsf{ADMIN}$~adds the authority to
grant or revoke access for other principals, following the
capability escalation model of Dennis and Van~Horn~\cite{dennis1966}.

\begin{definition}[Resource Universe]
\label{def:resource-universe}
The \emph{resource universe} is the finite set
\[
  \begin{aligned}
    \mathcal{C} = \bigl\{\;
      &\texttt{filesystem},\; \texttt{network},\;
       \texttt{environment},\; \texttt{shell}, \\
      &\texttt{skill\_invoke},\; \texttt{clipboard},\;
       \texttt{browser},\; \texttt{database} \;\bigr\},
  \end{aligned}
\]
with $|\mathcal{C}| = 8$.  These resource types were identified by
cross-referencing the empirical findings of
\citet{agentskillswild2026} (42,447 skills) and
\citet{maltool2026} (7,027 synthesised malicious tools).
\end{definition}

\begin{definition}[Capability and Capability Set]
\label{def:capability}%
\label{def:capability-set}% alias for cross-references from Section~6
A \emph{capability} is a pair $c = (r, \ell)$ where
$r \in \mathcal{C}$ is a resource type and
$\ell \in \mathbb{L}_{\mathit{cap}}$ is an access level.

Capability $c_1 = (r_1, \ell_1)$ \emph{subsumes}
$c_2 = (r_2, \ell_2)$, written $c_2 \sqsubseteq c_1$, iff
$r_1 = r_2$ and $\ell_2 \sqsubseteq \ell_1$.

A \emph{capability set} $\mathit{Cap}(s)$ for a skill $s$ is a
function $\mathit{Cap}(s) \colon \mathcal{C} \to
\mathbb{L}_{\mathit{cap}}$ mapping each resource to an access
level.  Equivalently, $\mathit{Cap}(s)$ is an element of the
product lattice $\mathbb{L}_{\mathit{cap}}^{|\mathcal{C}|}$.
\end{definition}

\noindent The product lattice ordering is pointwise:
$\mathit{Cap}_1 \sqsubseteq \mathit{Cap}_2$ iff
$\mathit{Cap}_1(r) \sqsubseteq \mathit{Cap}_2(r)$ for all
$r \in \mathcal{C}$.  This structure is a complete lattice with
$\bot = (\mathsf{NONE}, \ldots, \mathsf{NONE})$ and
$\top = (\mathsf{ADMIN}, \ldots, \mathsf{ADMIN})$.

\noindent\textbf{Lattice law verification.}
The lattice laws (idempotency, commutativity, associativity,
absorption) hold by inspection since $\sqcup = \max$ and
$\sqcap = \min$ on the totally ordered set
$\{\mathsf{NONE}, \mathsf{READ}, \mathsf{WRITE}, \mathsf{ADMIN}\}$
satisfy these properties for any total order.
Product lattice completeness follows from Birkhoff~\cite{birkhoff1940}:
the product of complete lattices is complete.

% ---------------------------------------------------------------------------
\subsection{Abstract Skill Semantics}
\label{sec:abstract-semantics}

We define the concrete and abstract semantics of a skill and
establish a Galois connection between them.

\begin{definition}[Concrete Skill Semantics]
\label{def:concrete-semantics}
Let $\Sigma$ be the set of concrete program states (memory contents,
file handles, network sockets, environment bindings).
The \emph{concrete semantics} of a skill $s$ is a function
$\llbracket s \rrbracket \colon \wp(\Sigma) \to \wp(\Sigma)$
that maps a set of initial states to the set of reachable final
states.
\end{definition}

\begin{definition}[Abstract Skill Semantics]
\label{def:abstract-semantics}
The \emph{abstract semantics} of $s$ is a function
$\llbracket s \rrbracket^\sharp \colon
\mathbb{L}_{\mathit{cap}}^{|\mathcal{C}|} \to
\mathbb{L}_{\mathit{cap}}^{|\mathcal{C}|}$
that over-approximates the concrete semantics in the capability
domain.
\end{definition}

We relate the two via a Galois connection, following
\citet{cousot1977}.

\begin{definition}[Galois Connection]
\label{def:galois-connection}
Let $(\wp(\Sigma), \subseteq)$ be the concrete domain and
$(\mathbb{L}_{\mathit{cap}}^{|\mathcal{C}|}, \sqsubseteq)$ be
the abstract domain.  The pair of functions
$(\alpha, \gamma)$ forms a Galois connection:
\[
  \alpha \colon \wp(\Sigma) \to
    \mathbb{L}_{\mathit{cap}}^{|\mathcal{C}|},
  \qquad
  \gamma \colon
    \mathbb{L}_{\mathit{cap}}^{|\mathcal{C}|} \to \wp(\Sigma),
\]
satisfying for all $S \in \wp(\Sigma)$ and all
$a \in \mathbb{L}_{\mathit{cap}}^{|\mathcal{C}|}$:
\begin{equation}
  \alpha(S) \sqsubseteq a
  \;\;\iff\;\;
  S \subseteq \gamma(a).
  \label{eq:galois-connection}
\end{equation}
\end{definition}

\noindent The abstraction function $\alpha$ extracts the capability
footprint: for each resource $r \in \mathcal{C}$, $\alpha(S)(r)$
is the least $\ell \in \mathbb{L}_{\mathit{cap}}$ such that every
state in $S$ accesses $r$ at level $\ell$ or below.  Concretization
$\gamma(a)$ returns all states whose resource accesses are bounded
by $a$.

\begin{proposition}[Soundness of Abstraction]
\label{prop:abstraction-soundness}
If $\alpha \circ \llbracket s \rrbracket \sqsubseteq
\llbracket s \rrbracket^\sharp \circ \alpha$, then the
abstract fixpoint over-approximates the concrete fixpoint:
\begin{equation}
  \alpha\bigl(\mathsf{lfp}\;\llbracket s \rrbracket\bigr)
  \sqsubseteq
  \mathsf{lfp}\;\llbracket s \rrbracket^\sharp.
  \label{eq:abstraction-sound}
\end{equation}
\end{proposition}

\noindent This is a direct instance of the fundamental theorem of
abstract interpretation~\cite{cousot1977}.  Concretely, it means:
if the abstract analysis concludes that $s$ accesses resource $r$
at level $\ell$, then every concrete execution of $s$ accesses $r$
at level $\ell$ or below.  False positives (the abstract analysis
may overestimate access) are possible; false negatives (the analysis
misses an access) are not.

% ---------------------------------------------------------------------------
\subsection{Three-Phase Analysis}
\label{sec:three-phase}

The \textsc{SkillFortify} static analyzer operates in three sequential phases.

\medskip
\noindent\textbf{Phase 1: Capability Inference (Abstract
Interpretation).}
Given a parsed skill $s$ with content fields (instructions,
shell commands, URLs, environment references, code blocks), the
analyzer computes the inferred capability set
$\mathit{Cap}_I(s) \in \mathbb{L}_{\mathit{cap}}^{|\mathcal{C}|}$
by applying $\llbracket s \rrbracket^\sharp$.

Concretely, the inference is a conservative over-approximation:
\begin{itemize}
  \item URL references $\Rightarrow$ \texttt{network} $\geq$
        $\mathsf{READ}$; HTTP write methods (POST, PUT, PATCH,
        DELETE) $\Rightarrow$ \texttt{network} $\geq$
        $\mathsf{WRITE}$.
  \item Shell command invocations $\Rightarrow$ \texttt{shell}
        $\geq$ $\mathsf{WRITE}$.
  \item Environment variable references $\Rightarrow$
        \texttt{environment} $\geq$ $\mathsf{READ}$.
  \item File-operation keywords $\Rightarrow$
        \texttt{filesystem} $\geq$ $\mathsf{READ}$ or
        $\mathsf{WRITE}$ depending on the operation.
\end{itemize}
Each rule is a \emph{transfer function} in the abstract domain:
it maps pattern presence to a lower bound on the capability
requirement.

\begin{lemma}[Transfer Function Soundness]
\label{lem:transfer-soundness}
Each transfer function in Phase~1 satisfies the Galois connection
condition: for every set of concrete states $S$ matching the
detected pattern, $\alpha(S) \sqsubseteq \tau^\sharp(\alpha(S))$,
where $\tau^\sharp$ is the abstract transfer function.
Table~\ref{tab:galois-verification} verifies this for each
resource type.
\end{lemma}

\begin{table}[h]
\centering
\caption{Galois connection verification for each transfer function.
For each pattern type, we show the inferred capability (abstract),
the set of concrete operations it covers, and the containment
relationship $\alpha(\text{concrete}) \sqsubseteq \text{abstract}$,
which holds for every row.}
\label{tab:galois-verification}
\small
\begin{tabular}{lll}
\toprule
\textbf{Pattern} & \textbf{Abstract $\tau^\sharp$} &
  \textbf{Concrete $\gamma(\tau^\sharp)$} \\
\midrule
URL reference
  & $(\texttt{net}, \mathsf{READ})$
  & States with HTTP GET access \\
HTTP write verb
  & $(\texttt{net}, \mathsf{WRITE})$
  & States with HTTP POST/PUT/\ldots \\
Shell command
  & $(\texttt{shell}, \mathsf{WRITE})$
  & States executing shell \\
Env var reference
  & $(\texttt{env}, \mathsf{READ})$
  & States reading env vars \\
File read keyword
  & $(\texttt{fs}, \mathsf{READ})$
  & States reading files \\
File write keyword
  & $(\texttt{fs}, \mathsf{WRITE})$
  & States writing files \\
Sub-skill invoke
  & $(\texttt{skill}, \mathsf{WRITE})$
  & States invoking skills \\
DB access keyword
  & $(\texttt{db}, \mathsf{READ/WRITE})$
  & States accessing database \\
\bottomrule
\end{tabular}
\end{table}

\noindent In each row, the concretization of the abstract result
($\gamma(\tau^\sharp)$) is a superset of the set of concrete states
matching the pattern.  This is because the abstract transfer
function assigns the \emph{maximum} access level that any concrete
operation in the category could require, which is the defining
property of a sound over-approximation under the Galois connection
(Definition~\ref{def:galois-connection}).

\medskip
\noindent\textbf{Phase 2: Dangerous Pattern Detection.}
The analyzer matches skill content against a catalog of
$|\mathcal{P}|$ threat patterns derived from ClawHavoc,
MalTool, and ``Agent Skills in the Wild.''  Each pattern
$p_j \in \mathcal{P}$ is a triple
$(r_j, \sigma_j, c_j)$ where $r_j$ is a compiled regular
expression, $\sigma_j \in \{\mathsf{LOW}, \mathsf{MEDIUM},
\mathsf{HIGH}, \mathsf{CRITICAL}\}$ is a severity, and $c_j$
is an attack class from Definition~\ref{def:attack-classes}.

Pattern examples include:
\begin{itemize}
  \item \texttt{curl} piped to shell
        ($\sigma = \mathsf{CRITICAL}$, $c = c_2$),
  \item Base64 decode piped to shell
        ($\sigma = \mathsf{CRITICAL}$, $c = c_2$),
  \item Netcat listener ($\sigma = \mathsf{CRITICAL}$, $c = c_1$),
  \item Dynamic code evaluation
        ($\sigma = \mathsf{HIGH}$, $c = c_2$).
\end{itemize}

Additionally, a cross-channel \emph{information flow} check
detects the combination of base64 encoding and external network
access---a pattern empirically associated with data
exfiltration~\cite{maltool2026}.

\medskip
\noindent\textbf{Phase 3: Capability Violation Check.}
If the skill declares a capability set $\mathit{Cap}_D(s)$,
the analyzer checks:
\begin{equation}
  \mathit{Cap}_I(s) \sqsubseteq \mathit{Cap}_D(s).
  \label{eq:cap-violation}
\end{equation}
Any resource $r$ where
$\mathit{Cap}_I(s)(r) \not\sqsubseteq \mathit{Cap}_D(s)(r)$
constitutes a \emph{capability violation}---the skill requires
more authority than it claims.

The set of violations is:
\begin{equation}
  \mathit{Viol}(s) = \bigl\{ r \in \mathcal{C} \mid
    \mathit{Cap}_I(s)(r) \not\sqsubseteq
    \mathit{Cap}_D(s)(r) \bigr\}.
  \label{eq:violations}
\end{equation}

% ---------------------------------------------------------------------------
\subsection{Theorem~2: Analysis Soundness}
\label{sec:analysis-soundness}

The following theorem is the central guarantee of \textsc{SkillFortify}'s
static analysis: it establishes that the analysis is
\emph{sound}---if it certifies a skill as compliant, the skill
truly cannot exceed its declared capabilities.

\begin{theorem}[Analysis Soundness]
\label{thm:analysis-soundness}
Let $s$ be a skill with declared capability set
$\mathit{Cap}_D(s)$.  If the abstract analysis reports no
violations---i.e., $\mathit{Viol}(s) = \emptyset$---then for every
concrete execution trace $\tau$ of $s$ and every resource access
$(r, \ell)$ performed in $\tau$:
\begin{equation}
  \ell \sqsubseteq \mathit{Cap}_D(s)(r).
  \label{eq:soundness}
\end{equation}
That is, the concrete execution never accesses any resource
beyond the declared capability level.
\end{theorem}

\begin{proof}[Proof sketch]
The argument proceeds in two steps.

\emph{Step 1: Abstraction soundness.}
By the Galois connection (Definition~\ref{def:galois-connection})
and Proposition~\ref{prop:abstraction-soundness}, the inferred
capability set $\mathit{Cap}_I(s)$ over-approximates the actual
resource accesses of $s$:
\[
  \forall\, \tau,\; \forall\, (r, \ell) \in \tau
  \colon \ell \sqsubseteq \mathit{Cap}_I(s)(r).
\]

\emph{Step 2: Violation check.}
If $\mathit{Viol}(s) = \emptyset$, then
$\mathit{Cap}_I(s) \sqsubseteq \mathit{Cap}_D(s)$, i.e.,
$\mathit{Cap}_I(s)(r) \sqsubseteq \mathit{Cap}_D(s)(r)$
for all $r \in \mathcal{C}$.

\emph{Combining.}
By transitivity of $\sqsubseteq$ on $\mathbb{L}_{\mathit{cap}}$:
\[
  \ell \sqsubseteq \mathit{Cap}_I(s)(r)
  \sqsubseteq \mathit{Cap}_D(s)(r)
  \qquad \text{for all } (r, \ell) \in \tau.
\]
The full proof, including the construction of the Galois connection
and the fixpoint argument, appears in
Appendix~\ref{app:proof-analysis-soundness}.
\end{proof}

\smallskip
\noindent\textbf{Practical significance.}
Theorem~\ref{thm:analysis-soundness} offers something existing
heuristic scanners do not.  Cisco's
\texttt{skill-scanner}~\cite{ciscoscanner2026} explicitly warns
that ``no findings does not mean no risk.''  We can say more, though
strictly less than that warning's converse: \emph{if
\textsc{SkillFortify} reports no violations for the analyzed
properties, there are no violations of those properties in any
concrete execution}.

The qualification carries the weight.  The guarantee is relative to the
abstract domain, and ranges over the capability behaviours that domain
models; it is not a claim that no malicious behaviour is present.
Sound over-approximation of a fixed property is the same technique that
the Astr\'{e}e analyzer~\cite{astree2005} applies to Airbus A380 flight
software, but Astr\'{e}e's property---absence of runtime errors---is
decidable from the program text, whereas ours depends on a catalogue of
modelled behaviours that is necessarily partial.
Section~\ref{sec:limitations} reports the false negatives that follow,
and we do not claim the zero-false-negative status Astr\'{e}e attains
for its property.

\noindent\textbf{Limitations.}
The analysis is sound but not complete: false positives
(flagging a safe skill as potentially dangerous) are possible
because the abstract domain over-approximates.  Additionally,
the current analysis covers statically observable patterns; skills
that dynamically construct network URLs or shell commands at
runtime may require dynamic analysis (Section~\ref{sec:discussion}).

% === END sections/04-static-analysis.tex ===

% === BEGIN sections/05-sandboxing.tex ===
% ===========================================================================
% Section 5: Capability-Based Sandboxing
% ===========================================================================

\section{Capability-Based Sandboxing}
\label{sec:sandboxing}

The static analysis of Section~\ref{sec:static-analysis} verifies
skill behavior \emph{before} execution.  In this section, we
complement static verification with a runtime enforcement mechanism
based on the capability model of Dennis and
Van~Horn~\cite{dennis1966} and its object-capability (ocap)
extension by Miller~\cite{miller2006}.  The central guarantee is
\emph{capability confinement}: a sandboxed skill cannot access any
resource beyond its declared capability set, even if the skill
attempts to escalate privileges at runtime.

% ---------------------------------------------------------------------------
\subsection{Capability Model for Agent Skills}
\label{sec:capability-model}

We ground the sandbox in the object-capability model, where
capabilities are unforgeable references that couple resource
designations with access rights.

\begin{definition}[Skill Capability]
\label{def:skill-capability}
A \emph{skill capability} is a pair $c = (r, \ell)$ where
$r \in \mathcal{C}$ (Definition~\ref{def:resource-universe}) is a
resource type and $\ell \in \mathbb{L}_{\mathit{cap}}$
(Definition~\ref{def:access-level-lattice}) is an access level.
A capability is \emph{unforgeable}: a skill can only obtain a
capability through the four mechanisms of Dennis and
Van~Horn~\cite{dennis1966}:
\begin{enumerate}
  \item \textbf{Initial conditions}---held at sandbox creation,
  \item \textbf{Parenthood}---the creator obtains the only reference,
  \item \textbf{Endowment}---the creator grants a subset of its
        own capabilities, and
  \item \textbf{Introduction}---a capability holder introduces
        two other principals via a shared reference.
\end{enumerate}
\end{definition}

Capabilities in this model are the \emph{only} source of authority.
There is no ambient authority---a skill cannot access a resource
simply because it ``exists'' in the environment.  This is the
defining property of a \emph{capability-safe} system~\cite{maffeis2010}.

\begin{definition}[Declared Capability Set]
\label{def:declared-capability-set}
For a skill $s$, the \emph{declared capability set}
$\mathit{Cap}_D(s) \subseteq \mathcal{C} \times
\mathbb{L}_{\mathit{cap}}$ is the set of capabilities that $s$
declares in its manifest.  This set is specified by the skill
author and verified by \textsc{SkillFortify}'s static analysis
(Section~\ref{sec:three-phase}) before the skill is installed.
\end{definition}

\begin{definition}[Required Capability]
\label{def:required-capability}
For an action $a$ performed by a skill during execution, the
\emph{required capability} $\mathit{req}(a) = (r, \ell)$ is
the minimum capability needed to perform $a$.  For example,
reading a file requires $(\texttt{filesystem}, \mathsf{READ})$;
executing a shell command requires
$(\texttt{shell}, \mathsf{WRITE})$.
\end{definition}

The sandbox mediates every action $a$ by checking:
\begin{equation}
  \exists\, (r', \ell') \in \mathit{Cap}_D(s) \colon
  r' = \mathit{req}(a).r \;\wedge\;
  \mathit{req}(a).\ell \sqsubseteq \ell'.
  \label{eq:sandbox-check}
\end{equation}
If Equation~\eqref{eq:sandbox-check} fails, the action is denied
and a security finding is raised.

% ---------------------------------------------------------------------------
\subsection{Capability Attenuation}
\label{sec:capability-attenuation}

Agent skills frequently delegate work to sub-skills (e.g., a
``research assistant'' skill invokes a ``web search'' skill and
a ``file writer'' skill).  The capability model must ensure that
delegation does not amplify authority.

\begin{definition}[Capability Attenuation]
\label{def:capability-attenuation}
When a parent skill $s_p$ delegates to a child skill $s_c$:
\begin{equation}
  \mathit{Cap}_D(s_c) \subseteq \mathit{Cap}_D(s_p).
  \label{eq:attenuation}
\end{equation}
That is, the child can receive at most the capabilities the
parent already possesses.  This implements the \emph{Principle of
Least Authority} (POLA)~\cite{miller2006}: each skill should
receive exactly the capabilities it needs and no more.
\end{definition}

\noindent Formally, for each resource $r \in \mathcal{C}$:
\[
  \mathit{Cap}_D(s_c)(r) \sqsubseteq \mathit{Cap}_D(s_p)(r).
\]

\noindent This is the \emph{no authority amplification} property of
capability systems~\cite{millermyths2003}: no skill can grant a
capability it does not possess.  Combined with unforgeability,
this bounds the transitive authority of any delegation chain.

\begin{proposition}[Transitive Attenuation]
\label{prop:transitive-attenuation}
For a delegation chain $s_1 \to s_2 \to \cdots \to s_k$:
\[
  \mathit{Cap}_D(s_k) \subseteq \mathit{Cap}_D(s_{k-1})
  \subseteq \cdots \subseteq \mathit{Cap}_D(s_1).
\]
Authority can only decrease along a delegation chain; it can
never increase.
\end{proposition}

\begin{proof}
Immediate by induction on the chain length using
Equation~\eqref{eq:attenuation}.
\end{proof}

% ---------------------------------------------------------------------------
\subsection{CapabilitySet Verification Operations}
\label{sec:capability-operations}

The sandbox exposes three verification operations that implement
the formal model.

\smallskip
\noindent\textbf{Operation 1: \textsc{Permits}$(c_{\mathit{req}})$.}
Given a required capability $c_{\mathit{req}} = (r, \ell)$, check
whether the declared set covers it:
\[
  \textsc{Permits}(\mathit{Cap}_D, c_{\mathit{req}}) =
  \begin{cases}
    \mathit{true}  & \text{if } \exists\, (r', \ell') \in
      \mathit{Cap}_D \colon r' = r \wedge
      \ell \sqsubseteq \ell', \\
    \mathit{false} & \text{otherwise}.
  \end{cases}
\]

\smallskip
\noindent\textbf{Operation 2:
\textsc{IsSubsetOf}$(\mathit{Cap}_1, \mathit{Cap}_2)$.}
Formal POLA verification---check that one capability set is
dominated by another:
\[
  \textsc{IsSubsetOf}(\mathit{Cap}_1, \mathit{Cap}_2) =
  \forall\, (r, \ell) \in \mathit{Cap}_1 \colon
  \textsc{Permits}(\mathit{Cap}_2, (r, \ell)).
\]
This is used both at install-time (checking that inferred
capabilities do not exceed declared, $\mathit{Cap}_I \sqsubseteq
\mathit{Cap}_D$) and at delegation-time (checking
$\mathit{Cap}_D(s_c) \sqsubseteq \mathit{Cap}_D(s_p)$).

\smallskip
\noindent\textbf{Operation 3:
\textsc{Violations}$(\mathit{Cap}_I, \mathit{Cap}_D)$.}
Return the set of capabilities that exceed declared bounds:
\[
  \textsc{Violations}(\mathit{Cap}_I, \mathit{Cap}_D) =
  \bigl\{ (r, \ell) \in \mathit{Cap}_I \mid
    \neg\, \textsc{Permits}(\mathit{Cap}_D, (r, \ell)) \bigr\}.
\]
An empty violation set constitutes a formal proof that the skill
respects its declared capability boundary:
$\textsc{Violations} = \emptyset \;\Rightarrow\;
\mathit{Cap}_I \sqsubseteq \mathit{Cap}_D$.

% ---------------------------------------------------------------------------
\subsection{Theorem~3: Capability Confinement}
\label{sec:confinement-theorem}

We state the capability confinement guarantee in two parts.
Theorem~3a provides the \emph{static} guarantee that
\textsc{SkillFortify}~v1 delivers: if the static analysis finds no
violations, then no operation in the skill's code exceeds declared
capabilities in the abstract domain.  Theorem~3b provides the
\emph{runtime design} guarantee for future sandbox implementations.

\begin{theorem}[Static Capability Confinement]
\label{thm:static-confinement}%
\label{thm:capability-confinement}% backward-compat alias
Let $s$ be a skill with declared capability set $\mathit{Cap}_D(s)$.
If the static analysis reports no violations---i.e.,
$\mathit{Viol}(s) = \emptyset$---then for every code construct in
the skill's definition, no operation $o$ in $s$ requires a
capability exceeding the declared level.  Formally:
\begin{equation}
  \forall\, o \in \mathit{ops}(s) \colon
  \mathit{cap}(o) \sqsubseteq
  \bigsqcup \mathit{Cap}_D(s).
  \label{eq:static-confinement}
\end{equation}
\end{theorem}

\begin{proof}[Proof sketch]
We proceed by structural induction on the syntax of the skill's
code.  For each code construct---assignment, conditional, loop,
function call---we show that the capability requirements of all
reachable operations are bounded by the inferred capability set,
which in turn is bounded by the declared set when
$\mathit{Viol}(s) = \emptyset$.

\emph{Base case} (single operation $o$):
The operation $o$ has capability requirement $\mathit{cap}(o)$.
The abstract analysis infers $\mathit{cap}(o) \sqsubseteq
\mathit{Cap}_I(s)(r)$ for the relevant resource $r$ (by soundness
of the transfer functions, Theorem~\ref{thm:analysis-soundness}).
Since $\mathit{Viol}(s) = \emptyset$, we have
$\mathit{Cap}_I(s) \sqsubseteq \mathit{Cap}_D(s)$, so
$\mathit{cap}(o) \sqsubseteq \mathit{Cap}_D(s)(r)$.

\emph{Sequential composition} ($s_1;\, s_2$):
By the induction hypothesis, all operations in $s_1$ and $s_2$
individually satisfy the bound.  Their union does not introduce
new operations, so the bound holds for the composition.

\emph{Conditional} ($\texttt{if}~b~\texttt{then}~s_1~
\texttt{else}~s_2$):
The inferred capabilities include both branches.  Any operation
reachable at runtime belongs to one branch and is bounded by
the induction hypothesis.

\emph{Loop} ($\texttt{while}~b~\texttt{do}~s_{\text{body}}$):
The inferred capabilities of the loop equal those of the body
(repetition does not introduce new capabilities).  By the induction
hypothesis on the body, the bound holds.

The full proof with all cases appears in
Appendix~\ref{app:proof-confinement}.
\end{proof}

\begin{theorem}[Runtime Capability Confinement---Design]
\label{thm:runtime-confinement}
If a runtime sandbox implementing the capability-safe execution
model (no ambient authority, unforgeable capabilities, capability
check on every action per Equation~\eqref{eq:sandbox-check}) is
deployed, then for any skill $s$ with declared capability set
$\mathit{Cap}_D(s)$ and for any action $a$ performed by $s$ during
execution:
\begin{equation}
  \mathit{req}(a).\ell \sqsubseteq \mathit{Cap}_D(s)(\mathit{req}(a).r).
  \label{eq:runtime-confinement}
\end{equation}
\end{theorem}

\noindent Theorem~\ref{thm:runtime-confinement} is a \emph{design
theorem}: it specifies the guarantee that a correctly implemented
runtime sandbox must provide.  The proof relies on the
capability-safe execution model of Dennis and
Van~Horn~\cite{dennis1966} and the object-capability confinement
result of Maffeis, Mitchell, and Taly~\cite{maffeis2010}.  We
include a proof sketch in Appendix~\ref{app:proof-confinement}
and note that \textsc{SkillFortify}~v1 provides the
\emph{static} guarantee (Theorem~\ref{thm:static-confinement});
runtime enforcement is a design contribution for future work.

\smallskip
\noindent\textbf{Relation to Maffeis, Mitchell \& Taly (2010).}
Theorem~\ref{thm:runtime-confinement} instantiates the
\emph{capability safety theorem} of \citet{maffeis2010} in
the agent skill domain.  Their result establishes that in
a capability-safe language (JavaScript strict mode in their
formalization), untrusted code is confined to the authority
explicitly passed to it.  We extend this to agent skills, where
the ``language'' is the skill execution runtime and the
``authority'' is the declared capability set.

\smallskip
\noindent\textbf{Combined guarantee.}
Theorems~\ref{thm:analysis-soundness}
and~\ref{thm:static-confinement} provide the defense
currently implemented in \textsc{SkillFortify}:
\begin{itemize}
  \item \emph{Analysis soundness} (Theorem~\ref{thm:analysis-soundness}):
        The inferred capabilities over-approximate the concrete
        resource accesses.
  \item \emph{Static confinement} (Theorem~\ref{thm:static-confinement}):
        If no violations are found, all operations in the skill's
        code are within declared capabilities.
\end{itemize}
Together with the runtime design theorem
(Theorem~\ref{thm:runtime-confinement}), these provide defense in
depth: static verification before installation, with the option of
runtime enforcement during execution.  A skill that passes the
static check is provably confined in the abstract domain: it cannot
exfiltrate data ($c_1$), escalate privileges ($c_2$), or abuse
undeclared resources, under the DY-Skill attacker model
(Section~\ref{sec:dy-skill}).

\smallskip
\noindent\textbf{Scope of guarantees.}
\textsc{SkillFortify}~v1 provides static analysis guarantees
(Theorems~\ref{thm:analysis-soundness}
and~\ref{thm:static-confinement}).  The runtime confinement
guarantee (Theorem~\ref{thm:runtime-confinement}) specifies the
contract for a future sandbox implementation.  This separation is
intentional: static analysis catches capability violations before
installation (the primary defense), while runtime sandboxing
provides a safety net for dynamically constructed operations that
static analysis cannot fully resolve (see
Section~\ref{sec:discussion}).

% === END sections/05-sandboxing.tex ===

% === BEGIN sections/06-dependency-graph.tex ===
%!TEX root = ../main.tex
% Section 6: Agent Dependency Graph & Lockfile Semantics
% Contribution C4: ADG + skill-lock.json + ASBOM
% ~4-5 pages

\section{Agent Dependency Graph and Lockfile Semantics}
\label{sec:dependency}%
\label{sec:dependency-graph}% alias for cross-references from introduction

Sections~\ref{sec:threat-model}--\ref{sec:sandboxing} established how to
\emph{detect} and \emph{confine} individual malicious skills.  In practice,
agent configurations rarely consist of a single skill in isolation: modern
agents compose dozens of skills, each of which may itself depend on other
skills, shared libraries, or external MCP servers.  A single compromised
transitive dependency can undermine the security of the entire installation,
as demonstrated by the SolarWinds~\cite{solarwinds2020} and
Log4Shell~\cite{log4shell2021} incidents in traditional software supply
chains, and by the ClawHavoc campaign~\cite{clawhavoc2026} in the agent
ecosystem.

This section introduces the \emph{Agent Dependency Graph} (ADG), a formal
data structure that captures skills, their version spaces, inter-skill
dependencies, conflicts, and per-version capability requirements.  We
present a SAT-based resolution algorithm that simultaneously satisfies
dependency, conflict, and security constraints, and we define
\emph{lockfile semantics} that guarantee reproducible, tamper-detectable
installations.  We conclude with an Agent Skill Bill of Materials (ASBOM)
output aligned with the CycloneDX~1.6 standard~\cite{cyclonedx2024}.

% ======================================================================
\subsection{Agent Dependency Graph (ADG)}
\label{sec:adg-definition}

We model the agent skill ecosystem as a \emph{package universe} extended
with security metadata, following the formalism of Mancinelli et
al.~\cite{mancinelli2006} and Tucker et al.~\cite{tucker2007opium}.

\begin{definition}[Agent Dependency Graph]
\label{def:adg}
An \emph{Agent Dependency Graph} is a tuple
$\mathit{ADG} = (\mathcal{S}, V, D, C, \mathit{Cap})$ where:
\begin{itemize}
  \item $\mathcal{S}$ is a finite set of \emph{skill names}.
  \item $V \colon \mathcal{S} \to 2^{\mathbb{N}}$ maps each skill to its
        set of available versions, ordered by semantic versioning
        precedence~\cite{semver2013}.
  \item $D \colon \mathcal{S} \times \mathbb{N} \to
        2^{\mathcal{S} \times \mathit{Constraint}}$ maps each
        skill-version pair $(s, v)$ to a set of \emph{dependency edges},
        where each edge $(q, \gamma)$ requires that some version of skill
        $q$ satisfying constraint $\gamma$ be co-installed.
  \item $C \colon \mathcal{S} \times \mathbb{N} \to
        2^{\mathcal{S} \times \mathit{Constraint}}$ maps each
        skill-version pair to a set of \emph{conflict edges}: if
        $(q, \gamma) \in C(s,v)$, then no version of $q$ satisfying
        $\gamma$ may be co-installed with $(s, v)$.
  \item $\mathit{Cap} \colon \mathcal{S} \times \mathbb{N} \to
        2^{\mathcal{C}}$ maps each skill-version pair to the set of
        capabilities it requires at runtime, where $\mathcal{C}$ is the
        capability universe from Definition~\ref{def:capability-set}.
\end{itemize}
\end{definition}

The ADG extends the classical package universe of Mancinelli et
al.~\cite{mancinelli2006} with the capability map $\mathit{Cap}$.
This extension is critical: it enables the resolver to reject
skill-versions that demand capabilities exceeding the security policy,
\emph{before} installation occurs.  Traditional package managers lack
this dimension entirely.

% ======================================================================
\subsection{Version Constraints}
\label{sec:version-constraints}

\begin{definition}[Version Constraint]
\label{def:version-constraint}
A \emph{version constraint} $\gamma$ is a predicate over version numbers.
We support the following atomic forms, mirroring PEP~440 and SemVer
conventions:
\begin{align*}
  \gamma &::= \texttt{==}\,v
        \mid \texttt{!=}\,v
        \mid \texttt{>=}\,v
        \mid \texttt{<=}\,v
        \mid \texttt{>}\,v
        \mid \texttt{<}\,v
        \mid \texttt{\^{}}\,v
        \mid \texttt{\~{}}\,v
        \mid \texttt{*}
\end{align*}
Compound constraints are formed by conjunction:
$\gamma_1 \wedge \gamma_2 \wedge \cdots \wedge \gamma_k$.  A version $w$
\emph{satisfies} a compound constraint iff it satisfies every atomic
constituent.
\end{definition}

The \emph{caret} operator $\texttt{\^{}}v$ permits any version with the
same major number and $\geq v$ (or, if the major number is~$0$, the same
major and minor).  The \emph{tilde} operator $\texttt{\~{}}v$ permits any
version with the same major and minor number and patch $\geq$ the target
patch.  The \emph{wildcard} $\texttt{*}$ is trivially satisfied by any
version.  These semantics follow standard package manager
conventions~\cite{gibb2026package}.

% ======================================================================
\subsection{SAT-Based Dependency Resolution}
\label{sec:sat-resolution}

Dependency resolution for package managers is NP-complete, reducible from
3-SAT~\cite{mancinelli2006}.  We follow the OPIUM approach~\cite{tucker2007opium}
and encode the resolution problem as a propositional satisfiability
instance, solved by a modern Conflict-Driven Clause Learning (CDCL)
solver.

\begin{definition}[SAT Encoding]
\label{def:sat-encoding}
Given an ADG $(\mathcal{S}, V, D, C, \mathit{Cap})$ and a set of
\emph{allowed capabilities} $\mathcal{A} \subseteq \mathcal{C}$, we
introduce a Boolean variable $x_{s,v}$ for each skill $s \in \mathcal{S}$
and version $v \in V(s)$.  Variable $x_{s,v} = \mathit{true}$ means
``skill $s$ at version $v$ is included in the installation.''  The
following clauses encode the constraints:

\medskip
\noindent\textbf{(C1) At-most-one version per skill.}
For each skill $s$ and every pair of distinct versions $v_i, v_j \in V(s)$:
\begin{equation}
  \neg x_{s,v_i} \lor \neg x_{s,v_j}
  \label{eq:amo}
\end{equation}

\noindent\textbf{(C2) Root requirements.}
For each user-specified requirement $(s, \gamma)$:
\begin{equation}
  \bigvee_{\{w \in V(s) \mid w \models \gamma\}} x_{s,w}
  \label{eq:root}
\end{equation}

\noindent\textbf{(C3) Dependency implications.}
For each $(s,v)$ and each dependency $(q, \gamma) \in D(s,v)$:
\begin{equation}
  \neg x_{s,v} \lor \bigvee_{\{w \in V(q) \mid w \models \gamma\}} x_{q,w}
  \label{eq:dep}
\end{equation}
If no version of $q$ satisfies $\gamma$, this degenerates to a unit clause
$\neg x_{s,v}$, effectively banning $(s,v)$.

\medskip
\noindent\textbf{(C4) Conflict exclusions.}
For each $(s,v)$ and each conflict $(q, \gamma) \in C(s,v)$, and for each
$w \in V(q)$ with $w \models \gamma$:
\begin{equation}
  \neg x_{s,v} \lor \neg x_{q,w}
  \label{eq:conflict}
\end{equation}

\noindent\textbf{(C5) Capability bounds.}
For each $(s,v)$ such that $\mathit{Cap}(s,v) \not\subseteq \mathcal{A}$:
\begin{equation}
  \neg x_{s,v}
  \label{eq:cap-bound}
\end{equation}
\end{definition}

The key novelty over classical package resolution is clause
family~(C5): it integrates the static capability analysis of
Section~\ref{sec:analysis} directly into the resolution process,
ensuring that no installation plan can include a skill-version
whose runtime requirements exceed the declared security policy.

\paragraph{Implementation.}
We use the Glucose3 CDCL solver via the PySAT
framework~\cite{imms2018}.  At-most-one constraints are encoded
as pairwise negation (Eq.~\ref{eq:amo}), which is compact for
the typical case of $|V(s)| \leq 20$ versions per skill.  For
larger version spaces, sequential counter encodings from the PySAT
cardinality module can be substituted.

% ======================================================================
\subsection{Secure Installation}
\label{sec:secure-installation}

\begin{definition}[Secure Installation]
\label{def:secure-installation}
An \emph{installation} $I \subseteq \mathcal{S} \times \mathbb{N}$ is a set
of skill-version pairs.  $I$ is \emph{secure} with respect to
$(\mathcal{S}, V, D, C, \mathit{Cap}, \mathcal{A})$ if and only if:
\begin{enumerate}
  \item \textbf{Completeness.}
    For every $(s,v) \in I$ and every $(q, \gamma) \in D(s,v)$, there
    exists $(q, w) \in I$ such that $w \models \gamma$.
  \item \textbf{Consistency.}
    For every $(s,v) \in I$ and every $(q, \gamma) \in C(s,v)$, there
    is no $(q, w) \in I$ with $w \models \gamma$.
  \item \textbf{Capability safety.}
    For every $(s,v) \in I$:
    $\mathit{Cap}(s,v) \subseteq \mathcal{A}$.
\end{enumerate}
\end{definition}

Condition~(3) is the agent-specific extension: it ties the installation
plan to the capability-based sandboxing model of Section~\ref{sec:sandboxing},
ensuring that every installed skill respects the declared security
boundary.  In traditional package managers (apt, pip, npm), no analogue
of this condition exists.

% ======================================================================
\subsection{Lockfile Semantics}
\label{sec:lockfile}

A \emph{lockfile} is the deterministic serialization of a secure
installation.  We define the \texttt{skill-lock.json} format to
provide three guarantees: reproducibility, integrity, and auditability.

\begin{definition}[Lockfile]
\label{def:lockfile}
A \emph{lockfile} $\mathcal{L}$ is a tuple
$\mathcal{L} = (I, H, M)$ where:
\begin{itemize}
  \item $I \subseteq \mathcal{S} \times \mathbb{N}$ is a secure
        installation (Definition~\ref{def:secure-installation}).
  \item $H \colon I \to \{0,1\}^{256}$ maps each installed skill-version
        to its SHA-256 content hash.
  \item $M$ is a metadata record containing the resolution strategy,
        the allowed capability set $\mathcal{A}$, a generation timestamp,
        and the tool version.
\end{itemize}
\end{definition}

\begin{definition}[Lockfile Reproducibility]
\label{def:lockfile-reproducibility}
A resolver $\mathcal{R}$ is \emph{reproducible} if, for any ADG and
requirement set, running $\mathcal{R}$ twice on the same input
produces identical lockfiles:
\[
  \mathcal{R}(\mathit{ADG}, \mathit{reqs}, \mathcal{A})
  = \mathcal{R}(\mathit{ADG}, \mathit{reqs}, \mathcal{A})
\]
\end{definition}

We achieve reproducibility by fixing the variable ordering in the
SAT encoding (lexicographic by skill name, then descending version)
and using a deterministic solver configuration.  The lockfile itself
is serialized with sorted keys and canonical JSON formatting, ensuring
byte-identical output for identical inputs.

\begin{definition}[Integrity Verification]
\label{def:lockfile-integrity}
At installation time, for each $(s, v) \in I$, the installer computes
the SHA-256 hash of the skill content and verifies:
\[
  \mathtt{sha256}(\mathit{content}(s, v)) = H(s, v)
\]
A mismatch indicates tampering: either the skill source was modified
after resolution (supply chain attack), or the lockfile was corrupted.
\end{definition}

This mechanism is analogous to the integrity fields in
\texttt{package-lock.json} (npm) and \texttt{Cargo.lock} (Rust), adapted
for the agent skill ecosystem where no prior lockfile standard exists.

% ======================================================================
\subsection{Agent Skill Bill of Materials (ASBOM)}
\label{sec:asbom}

Regulatory frameworks including the EU AI Act (Article~17, technical
documentation requirements)~\cite{euaiact2024} and the NIST AI Risk
Management Framework~\cite{nistai2023} increasingly require
transparency about third-party AI components.  We generate an
\emph{Agent Skill Bill of Materials} (ASBOM) as a structured inventory
of all skills in an agent's dependency closure.

\begin{definition}[Agent Skill Bill of Materials]
\label{def:asbom}
An ASBOM is a CycloneDX~1.6-compatible~\cite{cyclonedx2024} document
that records, for each $(s, v) \in I$:
\begin{itemize}
  \item \textbf{Identity:} skill name, version, format (Claude, MCP,
        OpenClaw), and content hash.
  \item \textbf{Capabilities:} the set $\mathit{Cap}(s,v)$ of declared
        runtime permissions.
  \item \textbf{Dependencies:} resolved dependency edges
        $\{(q, w) \mid (q, \gamma) \in D(s,v),\; (q, w) \in I,\;
        w \models \gamma\}$.
  \item \textbf{Trust metadata:} trust score, trust level
        (Section~\ref{sec:trust}), and provenance information.
  \item \textbf{Vulnerability status:} known CVEs or analysis findings
        affecting $(s,v)$.
\end{itemize}
\end{definition}

The CycloneDX format enables integration with existing enterprise
vulnerability management platforms (e.g., Dependency-Track, Grype)
and provides the machine-readable evidence required for compliance
audits.

% ======================================================================
\subsection{Transitive Vulnerability Propagation}
\label{sec:vuln-propagation}

A critical advantage of the graph-based model is the ability to
propagate vulnerability information transitively.

\begin{definition}[Affected Installation]
\label{def:affected}
Given a set of known-vulnerable skill-versions
$\mathit{Vuln} \subseteq \mathcal{S} \times \mathbb{N}$, a skill-version
$(s, v) \in I \setminus \mathit{Vuln}$ is \emph{transitively affected}
if there exists a path in the dependency graph from $(s, v)$ to some
$(q, w) \in \mathit{Vuln}$.  The set of affected installations is:
\[
  \mathit{Affected}(I, \mathit{Vuln}) = \{(s,v) \in I \setminus \mathit{Vuln}
  \mid \exists\, (q,w) \in \mathit{Vuln} \;.\;
  (s,v) \leadsto^{*} (q,w) \}
\]
where $\leadsto^{*}$ denotes transitive reachability through dependency
edges.
\end{definition}

We compute $\mathit{Affected}$ by BFS over the reverse dependency graph
(from vulnerable nodes to their dependents), analogous to how
\texttt{npm audit} identifies transitively affected packages.  The
difference is that our propagation also considers \emph{capability
implications}: if a vulnerable skill has capability $c \in
\mathit{Cap}(q, w)$, then all skills in its reverse transitive closure
that delegate or depend on capability $c$ are flagged with elevated
severity.

% ======================================================================
\subsection{Theorem 4: Resolution Soundness}
\label{sec:resolution-soundness}

\begin{theorem}[Resolution Soundness]
\label{thm:resolution-soundness}
Let $\mathit{ADG} = (\mathcal{S}, V, D, C, \mathit{Cap})$ be an Agent
Dependency Graph and $\mathcal{A} \subseteq \mathcal{C}$ a set of
allowed capabilities.  Let $\Phi$ be the conjunction of all clauses
(C1)--(C5) from Definition~\ref{def:sat-encoding}.  Then:
\begin{enumerate}
  \item $\Phi$ is satisfiable if and only if a secure installation
        (Definition~\ref{def:secure-installation}) exists.
  \item Any satisfying assignment $\sigma \models \Phi$ induces a secure
        installation $I_\sigma = \{(s,v) \mid \sigma(x_{s,v}) =
        \mathit{true}\}$, and the corresponding lockfile
        $\mathcal{L}_\sigma = (I_\sigma, H, M)$ satisfies all dependency,
        conflict, and capability constraints.
\end{enumerate}
\end{theorem}

\begin{proof}[Proof sketch]
We establish both directions.

\medskip
\noindent$(\Rightarrow)$\;
Suppose $\sigma \models \Phi$.  Define
$I_\sigma = \{(s,v) \mid \sigma(x_{s,v}) = \mathit{true}\}$.
\begin{itemize}
  \item \emph{At-most-one:} By clause family~(C1), for each skill $s$
    at most one variable $x_{s,v}$ is true, so $I_\sigma$ contains at
    most one version per skill.
  \item \emph{Completeness:} Fix $(s,v) \in I_\sigma$ and
    $(q, \gamma) \in D(s,v)$.  Since $\sigma(x_{s,v}) = \mathit{true}$,
    clause~(C3) requires $\sigma(x_{q,w}) = \mathit{true}$ for some
    $w \models \gamma$.  Hence $(q,w) \in I_\sigma$.
  \item \emph{Consistency:} Fix $(s,v) \in I_\sigma$ and
    $(q, \gamma) \in C(s,v)$.  For every $w \models \gamma$,
    clause~(C4) gives $\neg x_{s,v} \lor \neg x_{q,w}$.  Since
    $\sigma(x_{s,v}) = \mathit{true}$, we have
    $\sigma(x_{q,w}) = \mathit{false}$, so $(q,w) \notin I_\sigma$.
  \item \emph{Capability safety:} If $\mathit{Cap}(s,v)
    \not\subseteq \mathcal{A}$, clause~(C5) forces
    $\sigma(x_{s,v}) = \mathit{false}$, so $(s,v) \notin I_\sigma$.
    Contrapositive: $(s,v) \in I_\sigma \Rightarrow
    \mathit{Cap}(s,v) \subseteq \mathcal{A}$.
\end{itemize}

\noindent$(\Leftarrow)$\;
Suppose $I$ is a secure installation.  Define
$\sigma(x_{s,v}) = [(s,v) \in I]$ (Iverson bracket).
\begin{itemize}
  \item \emph{(C1):} At most one version per skill in $I$
    (by definition of installation), so pairwise exclusion holds.
  \item \emph{(C2):} For each requirement $(s, \gamma)$, completeness
    of $I$ implies some $(s, w) \in I$ with $w \models \gamma$, so
    the disjunction is satisfied.
  \item \emph{(C3):} For each $(s,v) \in I$ and $(q, \gamma) \in D(s,v)$,
    completeness gives $(q, w) \in I$ with $w \models \gamma$, making the
    implication clause true.  If $(s,v) \notin I$, the clause is trivially
    satisfied via $\neg x_{s,v}$.
  \item \emph{(C4):} Consistency of $I$ ensures no conflicting pair is
    co-installed; the mutual exclusion clauses hold.
  \item \emph{(C5):} Capability safety of $I$ ensures
    $\mathit{Cap}(s,v) \subseteq \mathcal{A}$ for all $(s,v) \in I$,
    so no unit clause $\neg x_{s,v}$ is violated.
\end{itemize}
Thus $\sigma \models \Phi$, and $\Phi$ is satisfiable.
\end{proof}

\begin{remark}
The bidirectional correspondence in Theorem~\ref{thm:resolution-soundness}
means that the SAT solver provides a \emph{complete} decision procedure
for secure installability: if the solver reports UNSAT, then
\emph{no} secure installation exists under the given constraints.
In such cases, the conflict analysis phase of the CDCL solver yields
an unsatisfiable core, which we translate into human-readable diagnostic
messages explaining which dependency, conflict, or capability constraint
caused the failure.
\end{remark}

\paragraph{Complexity.}
The dependency resolution problem is NP-complete in
general~\cite{mancinelli2006}.  Our capability-bounded extension
(clause family~C5) adds only $O(|\mathcal{S}| \cdot \max_s |V(s)|)$
unit clauses, so the theoretical worst case is unchanged.  In practice,
the agent skill ecosystem has $|\mathcal{S}| \leq 10^3$ and
$|V(s)| \leq 20$, well within the efficient regime of modern CDCL
solvers: our experiments (Section~\ref{sec:evaluation}) show resolution
times under 0.5 seconds for 500 skills.

\paragraph{Comparison with traditional package managers.}
Table~\ref{tab:dep-comparison} summarizes the differences between
our ADG-based resolver and traditional package dependency resolution.

\begin{table}[t]
\centering
\caption{Agent Dependency Graph vs.\ traditional package dependency resolution.}
\label{tab:dep-comparison}
\small
\begin{tabular}{lcc}
\toprule
\textbf{Feature} & \textbf{Traditional} & \textbf{ADG (Ours)} \\
\midrule
Version constraints      & \checkmark & \checkmark \\
Conflict detection       & \checkmark & \checkmark \\
SAT-based resolution     & \checkmark & \checkmark \\
Capability bounding      & ---        & \checkmark \\
Trust-aware resolution   & ---        & \checkmark \\
Lockfile integrity       & \checkmark & \checkmark \\
ASBOM generation         & Partial    & \checkmark \\
Vulnerability propagation & \checkmark & \checkmark \\
Capability-aware propagation & ---    & \checkmark \\
\bottomrule
\end{tabular}
\end{table}

% === END sections/06-dependency-graph.tex ===

% === BEGIN sections/07-trust-algebra.tex ===
%!TEX root = ../main.tex
% Section 7: Trust Score Algebra
% Contribution C5: Multi-signal trust with propagation, decay, and levels
% ~3-4 pages

\section{Trust Score Algebra}
\label{sec:trust}%
\label{sec:trust-algebra}% alias for cross-references from introduction

The preceding sections established mechanisms for verifying individual
skills (Sections~\ref{sec:analysis}--\ref{sec:sandboxing}) and resolving
their dependencies securely (Section~\ref{sec:dependency}).  A
complementary question remains: \emph{how much should a developer trust
a given skill?}  Binary safe/unsafe classification is insufficient in
practice---skills exist on a spectrum from unsigned, anonymous submissions
to formally verified, organization-signed artifacts with years of
community track record.

This section introduces a \emph{trust score algebra} that combines
heterogeneous trust signals into a single composite score, propagates
trust conservatively through dependency chains, models trust decay for
unmaintained skills, and maps numeric scores to graduated assurance
levels inspired by the SLSA framework~\cite{slsa2023}.  Our construction
draws on the trust management literature, particularly the PolicyMaker
and KeyNote systems~\cite{blaze1996policymaker,keynote1999}, and the formal trust models
of Carbone et al.~\cite{carbone2003} and J{\o}sang~\cite{josang2001}.

% ======================================================================
\subsection{Trust Score Model}
\label{sec:trust-model}

\begin{definition}[Trust Signals]
\label{def:trust-signals}
For a skill $s$ at version $v$, we define four orthogonal \emph{trust
signals}, each taking values in the closed interval $[0, 1]$:
\begin{enumerate}
  \item $T_p(s,v) \in [0,1]$: \textbf{Provenance.}  Measures author
    verification and signing status.  $T_p = 0$ for unsigned, anonymous
    authors; $T_p = 1$ for organization-signed skills with reproducible
    builds and verified identity.
  \item $T_b(s,v) \in [0,1]$: \textbf{Behavioral.}  Measures the
    outcome of static analysis (Section~\ref{sec:analysis}).
    $T_b = 0$ if critical capability violations are detected;
    $T_b = 1$ if analysis reports no findings.
  \item $T_c(s,v) \in [0,1]$: \textbf{Community.}  Aggregates
    community trust signals: usage count, age, issue reports, and
    peer reviews.  $T_c = 0$ for a brand-new skill with no usage;
    $T_c = 1$ for a widely adopted, mature skill with no reported
    issues.
  \item $T_h(s,v) \in [0,1]$: \textbf{Historical.}  Captures the
    skill's past vulnerability record.  $T_h = 0$ if the skill has
    a history of CVEs or security incidents; $T_h = 1$ for a clean
    historical record.
\end{enumerate}
We write $\mathbf{T}(s,v) = (T_p, T_b, T_c, T_h) \in [0,1]^4$ for
the signal vector.
\end{definition}

The four signals are designed to be \emph{orthogonal}: provenance
measures \emph{who} published the skill, behavioral measures
\emph{what} the skill does, community measures \emph{how widely}
it is adopted, and historical measures \emph{how reliable} it has
been over time.

\begin{definition}[Trust Weights]
\label{def:trust-weights}
A \emph{weight vector} $\mathbf{w} = (w_p, w_b, w_c, w_h) \in
\mathbb{R}_{\geq 0}^{4}$ satisfies:
\[
  w_p + w_b + w_c + w_h = 1
  \qquad\text{and}\qquad
  w_p, w_b, w_c, w_h \geq 0
\]
The default weight vector is $\mathbf{w}_0 = (0.3, 0.3, 0.2, 0.2)$,
reflecting equal priority for provenance and behavioral analysis
(the signals most directly under the tool's and author's control),
with lower weight for community and historical signals (which
accumulate over time).
\end{definition}

\begin{definition}[Intrinsic Trust Score]
\label{def:intrinsic-trust}
The \emph{intrinsic trust score} of skill $s$ at version $v$ is the
weighted linear combination:
\begin{equation}
  T(s,v) = w_p \cdot T_p(s,v) + w_b \cdot T_b(s,v) + w_c \cdot T_c(s,v) + w_h \cdot T_h(s,v)
  \label{eq:trust-intrinsic}
\end{equation}
Since all weights are non-negative and sum to $1$, and all signals
are in $[0,1]$, the score satisfies $T(s,v) \in [0,1]$.
\end{definition}

% ======================================================================
\subsection{Trust Propagation}
\label{sec:trust-propagation}

A skill's effective trustworthiness depends not only on its own
properties but on the trustworthiness of its dependencies.  A
formally verified skill that depends on an unsigned, unvetted
library inherits risk from that dependency.  We model this through
\emph{multiplicative trust propagation}, following the transitive
trust composition approach of Carbone et al.~\cite{carbone2003}.

\begin{definition}[Trust Composition Operator]
\label{def:trust-composition}
The \emph{trust composition operator} $\otimes \colon [0,1] \times
[0,1] \to [0,1]$ is defined as:
\begin{equation}
  a \otimes b = a \cdot b
  \label{eq:trust-composition}
\end{equation}
This is conservative: composed trust is always $\leq$ either
operand, reflecting that a chain of trust is only as strong as
the product of its links.
\end{definition}

\begin{definition}[Transitive Dependency Set]
\label{def:transitive-deps}
For a skill $s$ in a dependency DAG, the \emph{transitive dependency
set} $\mathit{trans}(s)$ is the set of all skills reachable from $s$
via dependency edges, \emph{deduplicated}:
\[
  \mathit{trans}(s) = \bigl\{ s' \mid s \leadsto^{+} s' \bigr\}
\]
where $\leadsto^{+}$ denotes transitive reachability through
dependency edges.  Shared sub-dependencies (diamond patterns) appear
exactly once in $\mathit{trans}(s)$.
\end{definition}

\begin{definition}[Effective Trust Score]
\label{def:effective-trust}
For a skill $s$ in a dependency DAG, the \emph{effective trust
score} is defined over the deduplicated transitive dependency set:
\begin{equation}
  T_{\mathrm{eff}}(s) = T(s) \times
    \min_{d \in \mathit{trans}(s)} T(d)
  \label{eq:trust-effective}
\end{equation}
If $\mathit{trans}(s) = \emptyset$ (leaf skill with no
dependencies), then $T_{\mathrm{eff}}(s) = T(s)$.
\end{definition}

\noindent This formulation resolves the ambiguity that arises in
DAGs with diamond dependency patterns (where skill $A$ depends on
$B$ and $C$, both of which depend on $D$).  The recursive
formulation $T(s) \otimes \min_i T_{\mathrm{eff}}(s_i)$ would
count $D$'s trust contribution multiple times depending on the
traversal order.  Our definition computes the minimum over the
\emph{deduplicated} transitive closure, ensuring a unique result.

\begin{proposition}[Uniqueness]
\label{prop:trust-uniqueness}
$T_{\mathrm{eff}}(s)$ is uniquely defined regardless of
traversal order.
\end{proposition}

\begin{proof}
$T_{\mathrm{eff}}(s)$ depends only on $T(s)$ and the set
$\{T(d) \mid d \in \mathit{trans}(s)\}$.  Both $T(s)$ and the
set $\mathit{trans}(s)$ are properties of the DAG structure
(not of any particular traversal).  The $\min$ operator over a
set is order-independent.  Therefore $T_{\mathrm{eff}}(s)$ is
uniquely determined.
\end{proof}

The $\min$ over dependencies implements \emph{weakest-link}
semantics: the effective trust of a skill is bounded by the
least-trusted dependency in its transitive closure, attenuated
by the multiplicative composition.  This is analogous
to the assertion monotonicity property of KeyNote~\cite{keynote1999}:
adding a trusted intermediary never reduces the compliance value
of a delegation chain.

\begin{proposition}[Propagation Bound]
\label{prop:propagation-bound}
For any skill $s$ in a dependency DAG with longest dependency path
of length $d$ (measured as the number of edges in the longest path
from $s$ to any leaf in the DAG):
\[
  T_{\mathrm{eff}}(s) \leq T(s)
\]
with equality if and only if all transitive dependencies have
intrinsic trust score~$1$.  More generally, if
$T_{\min} = \min_{s' \in \mathit{trans}(s) \cup \{s\}} T(s')$
is the minimum intrinsic trust across all skills in the transitive
dependency closure (including $s$ itself), then:
\[
  T_{\mathrm{eff}}(s) \geq T_{\min}^{2}
\]
and in particular $T_{\mathrm{eff}}(s) \geq T(s) \cdot T_{\min}$.
\end{proposition}

\begin{proof}
\emph{Upper bound.}
If $\mathit{trans}(s) \neq \emptyset$, then
$\min_{d \in \mathit{trans}(s)} T(d) \leq 1$, so
$T_{\mathrm{eff}}(s) = T(s) \cdot \min_d T(d) \leq T(s)$.
If $\mathit{trans}(s) = \emptyset$, $T_{\mathrm{eff}}(s) = T(s)$.

\emph{Lower bound.}
$T(s) \geq T_{\min}$ and
$\min_{d \in \mathit{trans}(s)} T(d) \geq T_{\min}$ by
definition of $T_{\min}$.  Hence
$T_{\mathrm{eff}}(s) = T(s) \cdot \min_d T(d) \geq
T_{\min} \cdot T_{\min} = T_{\min}^2$.
The bound $T_{\mathrm{eff}}(s) \geq T(s) \cdot T_{\min}$
follows from $\min_d T(d) \geq T_{\min}$.
\end{proof}

% ======================================================================
\subsection{Trust Decay}
\label{sec:trust-decay}

Skills that are not actively maintained accumulate risk over time:
unpatched vulnerabilities, compatibility drift with evolving agent
runtimes, and stale dependencies.  We model this via exponential
decay.

\begin{definition}[Temporal Trust Decay]
\label{def:trust-decay}
For a skill $s$ last updated at time $t_{\mathrm{last}}$, the
\emph{decayed trust score} at time $t > t_{\mathrm{last}}$ is:
\begin{equation}
  T(s, t) = T_0(s) \cdot e^{-\lambda \cdot (t - t_{\mathrm{last}})}
  \label{eq:trust-decay}
\end{equation}
where $T_0(s) = T_{\mathrm{eff}}(s)$ is the effective trust score
at the time of last update, and $\lambda > 0$ is the \emph{decay
rate} parameter.
\end{definition}

The decay rate $\lambda$ controls how aggressively trust erodes
with inactivity.  At the default rate $\lambda = 0.01\;\text{day}^{-1}$,
trust halves approximately every 69 days ($t_{1/2} = \ln 2 / \lambda$).
This aligns with the empirical observation that unmaintained open-source
packages are significantly more likely to contain unpatched
vulnerabilities after 3--6 months of inactivity~\cite{zahan2022}.

\begin{proposition}[Decay Properties]
\label{prop:decay-properties}
The decay function (Eq.~\ref{eq:trust-decay}) satisfies:
\begin{enumerate}
  \item \textbf{Monotonic decrease:} $T(s, t_1) \geq T(s, t_2)$
    for $t_1 \leq t_2$ (trust never increases with time alone).
  \item \textbf{Bounded:} $T(s, t) \in [0, T_0(s)]$ for all
    $t \geq t_{\mathrm{last}}$.
  \item \textbf{Asymptotic:} $\lim_{t \to \infty} T(s, t) = 0$
    (eventually, unmaintained skills have negligible trust).
  \item \textbf{Reset on update:} When a skill is updated at time
    $t'$, the trust score resets to
    $T_0'(s) = T_{\mathrm{eff}}(s, t')$, re-evaluating all four
    signals with current data.
\end{enumerate}
\end{proposition}

% ======================================================================
\subsection{Graduated Trust Levels}
\label{sec:trust-levels}

For operational decision-making, we map continuous trust scores to
discrete levels inspired by the Supply chain Levels for Software
Artifacts (SLSA) framework~\cite{slsa2023}.

\begin{definition}[Trust Level Mapping]
\label{def:trust-levels}
The \emph{trust level} $\ell(s)$ of a skill $s$ is determined by
its effective (possibly decayed) trust score:
\begin{equation}
  \ell(s) = \begin{cases}
    \text{L0 (\textsc{unsigned})}           & \text{if } T_{\mathrm{eff}}(s) < 0.25 \\[4pt]
    \text{L1 (\textsc{signed})}             & \text{if } 0.25 \leq T_{\mathrm{eff}}(s) < 0.50 \\[4pt]
    \text{L2 (\textsc{community\_verified})} & \text{if } 0.50 \leq T_{\mathrm{eff}}(s) < 0.75 \\[4pt]
    \text{L3 (\textsc{formally\_verified})}  & \text{if } T_{\mathrm{eff}}(s) \geq 0.75
  \end{cases}
  \label{eq:trust-levels}
\end{equation}
\end{definition}

\noindent The four levels correspond to increasing assurance:

\begin{itemize}
  \item \textbf{L0 (Unsigned):} No verification.  The skill has not
    been signed, analyzed, or reviewed.  Default for newly discovered,
    anonymous skills.
  \item \textbf{L1 (Signed):} Basic provenance established.  The
    author has signed the package, providing identity assurance but
    no behavioral or community verification.
  \item \textbf{L2 (Community Verified):} Multiple positive signals.
    The skill has usage history, community reviews, and passes basic
    behavioral analysis.
  \item \textbf{L3 (Formally Verified):} Highest assurance.  The
    skill passes formal static analysis
    (Section~\ref{sec:analysis}), has strong provenance, active
    community trust, and a clean historical record.
\end{itemize}

Enterprise policies can mandate minimum trust levels for skill
installation.  For example, a policy requiring $\ell(s) \geq
\text{L2}$ for all production skills is expressible as a constraint
in the ADG resolver (Section~\ref{sec:sat-resolution}), integrated
alongside capability bounds.

% ======================================================================
\subsection{Theorem 5: Trust Monotonicity}
\label{sec:trust-monotonicity}

The central theoretical property of the trust algebra is
\emph{monotonicity}: positive evidence about a skill should never
reduce its trust score.  This corresponds to the assertion
monotonicity axiom of KeyNote~\cite{keynote1999}: adding credentials
to a compliance check never reduces the compliance value.

\begin{theorem}[Trust Monotonicity]
\label{thm:trust-monotonicity}
Let $s$ be a skill with signal vector $\mathbf{T} = (T_p, T_b, T_c,
T_h)$ and trust score $T(s) = \mathbf{w} \cdot \mathbf{T}$.  Let
$e$ be \emph{positive evidence} that increases one or more signal
values: there exists $\mathbf{T}' = (T_p', T_b', T_c', T_h')$ with
$T_p' \geq T_p$, $T_b' \geq T_b$, $T_c' \geq T_c$,
$T_h' \geq T_h$, and at least one strict inequality.  Then:
\begin{equation}
  T(s \mid e) = \mathbf{w} \cdot \mathbf{T}' \geq
  \mathbf{w} \cdot \mathbf{T} = T(s)
  \label{eq:trust-monotonicity}
\end{equation}
Moreover, $T(s \mid e) > T(s)$ whenever the strict inequality occurs
on a signal $i$ with $w_i > 0$.
\end{theorem}

\begin{proof}
The difference between the updated and original trust scores is:
\begin{align}
  T(s \mid e) - T(s)
  &= \mathbf{w} \cdot \mathbf{T}' - \mathbf{w} \cdot \mathbf{T}
  \notag \\
  &= w_p (T_p' - T_p) + w_b (T_b' - T_b) + w_c (T_c' - T_c) + w_h (T_h' - T_h)
  \label{eq:mono-diff}
\end{align}
Each term in Eq.~\ref{eq:mono-diff} is a product of a non-negative
weight $w_i \geq 0$ (by Definition~\ref{def:trust-weights}) and a
non-negative increment $T_i' - T_i \geq 0$ (by the hypothesis that
$\mathbf{T}' \geq \mathbf{T}$ component-wise).  Therefore:
\[
  T(s \mid e) - T(s) \geq 0
\]
For strict monotonicity, if $T_j' > T_j$ for some signal $j$ with
$w_j > 0$, then the $j$-th term is strictly positive, giving
$T(s \mid e) > T(s)$.

Monotonicity extends to effective trust scores because the propagation
operator $\otimes$ (multiplication) is monotone in both arguments on
$[0,1]$, and the $\min$ operator is monotone in each of its arguments.
Formally, if $T(s) \leq T(s \mid e)$, then for any dependency trust
value $d \in [0,1]$:
\[
  T(s) \cdot d \leq T(s \mid e) \cdot d
\]
and hence $T_{\mathrm{eff}}(s) \leq T_{\mathrm{eff}}(s \mid e)$
provided all dependency effective scores remain unchanged.
\end{proof}

\begin{remark}[Operational Significance]
Theorem~\ref{thm:trust-monotonicity} ensures that the trust system is
\emph{incentive-compatible}: skill authors who improve their provenance
(e.g., by signing packages), fix vulnerabilities (improving $T_h$), or
submit to formal verification (improving $T_b$) are guaranteed a
non-decreasing trust score.  The system never punishes improvement.
This property is essential for adoption: developers and enterprise
security teams need confidence that investing in skill quality
reliably improves the skill's trust standing.
\end{remark}

\paragraph{Connection to the ADG resolver.}
Trust scores integrate with the SAT-based resolver of
Section~\ref{sec:sat-resolution} through an optional
\emph{trust-aware optimization objective}: among all satisfying
assignments (secure installations), prefer the one that maximizes
the minimum effective trust score across installed skills.  This
transforms the satisfiability problem into a MAX-SAT optimization:
\begin{equation}
  I^{*} = \argmax_{I \models \Phi}
    \min_{(s,v) \in I} T_{\mathrm{eff}}(s,v)
  \label{eq:trust-optimization}
\end{equation}
where $\Phi$ is the conjunction of clauses (C1)--(C5).  In practice,
we approximate this by iteratively adding trust threshold constraints
(``all installed skills must have $T_{\mathrm{eff}} \geq \tau$'') and
binary-searching on $\tau$ until the maximum feasible threshold is found.

\paragraph{Algebraic structure.}
The trust score model forms a \emph{bounded semilattice} over the
interval $[0,1]$: the composition operator $\otimes$ (multiplication)
is commutative, associative, and has identity element $1$.  The $\min$
operator provides the meet.  These algebraic properties ensure that
trust propagation is well-defined regardless of the order in which
dependency chains are evaluated, and that the propagated trust score
is unique for any given ADG.

\begin{proposition}[Trust Algebra Properties]
\label{prop:trust-algebra}
The tuple $([0,1], \otimes, \min, 0, 1)$ forms a bounded
semilattice where:
\begin{enumerate}
  \item $\otimes$ is commutative: $a \otimes b = b \otimes a$.
  \item $\otimes$ is associative:
    $(a \otimes b) \otimes c = a \otimes (b \otimes c)$.
  \item $1$ is the identity: $a \otimes 1 = a$.
  \item $0$ is the annihilator: $a \otimes 0 = 0$.
  \item $\min$ is idempotent, commutative, and associative.
  \item Monotonicity: if $a \leq a'$ and $b \leq b'$, then
    $a \otimes b \leq a' \otimes b'$.
\end{enumerate}
\end{proposition}

\begin{proof}
Properties (1)--(4) follow immediately from the corresponding
properties of multiplication on $[0,1]$.  Property (5) follows from
the standard properties of $\min$ on totally ordered sets.
Property (6): if $a \leq a'$ and $b \leq b'$, then
$a \cdot b \leq a' \cdot b'$ since all values are non-negative.
\end{proof}

% === END sections/07-trust-algebra.tex ===

% === BEGIN sections/08-implementation.tex ===
% =============================================================================
% Section 8: Implementation
% =============================================================================

\section{Implementation}
\label{sec:implementation}

We implement the formal framework described in
Sections~\ref{sec:threat-model}--\ref{sec:trust} as \textsc{SkillFortify},
an open-source Python tool for agent skill supply chain security.
This section describes the system architecture, supported skill
formats, command-line interface, and engineering quality measures.

% -----------------------------------------------------------------------------
\subsection{Architecture}
\label{sec:arch}

\textsc{SkillFortify} is organized around three core engines connected
via a shared Threat Knowledge Base
(Figure~\ref{fig:architecture}):

\begin{enumerate}
    \item \textbf{Static Analyzer.}  Implements the abstract
    interpretation framework from Section~\ref{sec:analysis}.  The
    analyzer performs pattern-based detection augmented with
    information flow analysis.  Patterns are organized into a
    taxonomy of 13 attack types derived from
    ClawHavoc~\cite{clawhavoc2026} and MalTool~\cite{maltool2026}
    incident data.  The information flow component tracks data
    propagation from sensitive sources (environment variables, file
    system, user credentials) to untrusted sinks (network endpoints,
    external processes), enabling detection of steganographic
    exfiltration channels that evade pure pattern matching.

    \item \textbf{Dependency Resolver.}  Implements the Agent
    Dependency Graph
    $\mathit{ADG} = (\mathcal{S}, V, D, C, \mathit{Cap})$ from
    Section~\ref{sec:dependency}.  The resolver constructs a directed
    acyclic graph of skill dependencies, enforces semantic version
    constraints, detects conflicts via Boolean satisfiability
    encoding, and produces deterministic lockfiles.  Cycle detection
    uses Kahn's algorithm with $O(|V| + |E|)$ complexity.

    \item \textbf{Trust Engine.}  Implements the trust score algebra
    $T\colon \mathcal{S} \times \mathbb{N} \to [0,1]$ from
    Section~\ref{sec:trust}.  The engine computes multi-signal trust
    scores incorporating provenance verification, behavioral analysis
    results, and temporal decay.  Trust propagation through dependency
    chains follows the formal model with configurable weight vectors
    $\mathbf{w} \in \Delta^{k-1}$.
\end{enumerate}

All three engines share a \textbf{Threat Knowledge Base} containing:
(i)~the DY-Skill threat taxonomy,
(ii)~the capability lattice $(\mathcal{C}, \sqsubseteq)$,
(iii)~pattern signatures for known attack types, and
(iv)~trust signal configurations.
The knowledge base is implemented as an extensible registry, allowing
users to add custom threat patterns and trust signals without
modifying engine code.

\begin{figure}[t]
    \centering
    \small
    \begin{tikzpicture}[
        box/.style={draw, rounded corners,
            minimum width=2.8cm, minimum height=0.9cm,
            align=center, font=\footnotesize},
        kb/.style={draw, rounded corners, fill=gray!15,
            minimum width=8.5cm, minimum height=0.7cm,
            align=center, font=\footnotesize},
        arr/.style={->, thick, >=stealth}
    ]
        \node[box] (sa) at (0, 0) {Static\\Analyzer};
        \node[box] (dr) at (3.5, 0) {Dependency\\Resolver};
        \node[box] (te) at (7, 0) {Trust\\Engine};
        \node[kb] (kb) at (3.5, -1.5) {Threat Knowledge Base};
        \node[box, minimum width=8.5cm] (cli)
            at (3.5, 1.5) {\textsc{SkillFortify} CLI};
        \draw[arr] (cli) -- (sa);
        \draw[arr] (cli) -- (dr);
        \draw[arr] (cli) -- (te);
        \draw[arr] (kb) -- (sa);
        \draw[arr] (kb) -- (dr);
        \draw[arr] (kb) -- (te);
        \draw[arr] (sa) -- (dr);
        \draw[arr] (dr) -- (te);
    \end{tikzpicture}
    \caption{Architecture of \textsc{SkillFortify}.  Three engines
    (Static Analyzer, Dependency Resolver, Trust Engine) share a
    Threat Knowledge Base.  The CLI orchestrates all engines and
    produces unified reports.}
    \label{fig:architecture}
\end{figure}
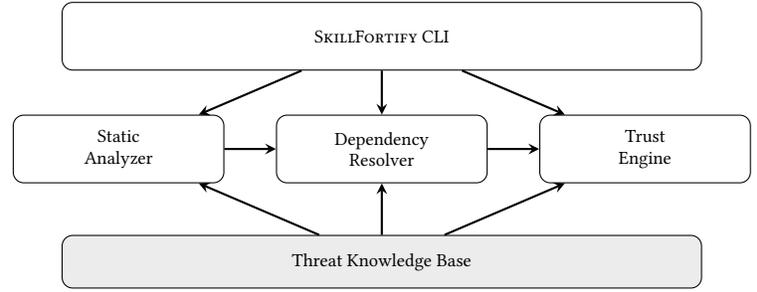

% -----------------------------------------------------------------------------
\subsection{Scope of Formal Guarantees}
\label{sec:scope-guarantees}

An important scope clarification: the formal analysis
(Theorem~\ref{thm:analysis-soundness}) applies to the skill's
\emph{code, configuration, and metadata}---all deterministic
artifacts.
LLM inference outputs are inherently non-deterministic; an agent may
invoke a skill in ways that depend on stochastic model reasoning,
and the natural-language instructions within a skill manifest may be
interpreted differently across model versions or sampling parameters.
These dynamic, non-deterministic aspects are outside the scope of
static analysis.

\textsc{SkillFortify}'s guarantees therefore apply to the
\emph{structural} attack surface: embedded code, declared commands,
environment variable exposure, dependency relationships, and
information flow through deterministic skill logic.
Attacks that emerge purely from LLM reasoning at runtime---for
example, a benign-looking skill whose natural-language description
manipulates the model into performing unintended actions---require
complementary dynamic analysis techniques.
We discuss this limitation and the path toward runtime enforcement
in Section~\ref{sec:future}.

\paragraph{How the implementation realises the abstract semantics.}
The abstract semantics of Section~\ref{sec:analysis} is defined by
structural induction over composite constructs, including a least
fixpoint for loops (Appendix~\ref{app:proof-thm2}).  The implementation
does not iterate to a fixpoint, and does not need to.

Every transfer function in our domain is \emph{presence-based}: it maps
the observation of a syntactic feature to a lower bound on a capability,
and it is monotone and extensive.  No transfer function ever lowers an
inferred level, because a skill that has once been observed to reach the
network cannot become unable to.  The abstract value therefore only
ascends, and the abstraction of a loop body's second iteration is equal
to that of its first.  The ascending chain of
Equation~\ref{eq:abstraction-sound} stabilises after one step, so
evaluating each extracted operation once and joining the results
computes exactly the least fixpoint the model specifies.

What this buys is a guarantee whose scope is honest to state: the
implementation is sound with respect to \emph{the operations it
extracts}.  It parses a skill into its declared commands, referenced
environment variables, URLs, and embedded code blocks, and reasons over
that set.  It does not build a control-flow graph, and an operation
whose existence is only apparent after evaluating a construct---a
command assembled from string fragments at runtime, a URL read from a
file---is not in the set and is therefore outside the theorem, not a
counterexample to it.  Section~\ref{sec:limitations} reports what that
costs in practice.

% -----------------------------------------------------------------------------
\subsection{Skill Format Support}
\label{sec:format-support}

\textsc{SkillFortify} supports the three dominant agent skill formats as
of February~2026:

\begin{itemize}
    \item \textbf{Claude Code Skills}
    (\texttt{.md} with YAML frontmatter).
    Anthropic's skill format~\cite{anthropic-skills2026} embeds
    capability declarations in Markdown files with structured YAML
    metadata.  \textsc{SkillFortify} extracts capability claims, tool
    declarations, and inline code blocks for analysis.

    \item \textbf{MCP Servers} (\texttt{.json} configuration).
    The Model Context Protocol~\cite{mcp2025} defines server
    configurations as JSON manifests specifying command-line
    invocations, environment variable bindings, and transport
    parameters.  \textsc{SkillFortify} analyzes the declared command,
    environment exposure, and argument patterns.

    \item \textbf{OpenClaw Skills} (\texttt{SKILL.md}).
    OpenClaw~\cite{openclaw2026} implements the Agent Skills open
    standard: a skill is a \emph{directory} containing a
    \texttt{SKILL.md} file whose YAML frontmatter declares the skill
    name, description, and optional runtime requirements under
    \texttt{metadata.openclaw}.  Skills load from
    \texttt{\textasciitilde/.openclaw/skills/} and from
    \texttt{.openclaw/skills/} within a workspace.
    \textsc{SkillFortify} parses the frontmatter, the declared runtime
    requirements, and the instruction body including embedded script
    payloads.
\end{itemize}

Format detection is automatic via a \texttt{ParserRegistry} that
dispatches on file extension and content heuristics.
Each parser normalizes skill metadata into a unified internal
representation
$s = (\mathit{id}, \mathit{name}, \mathit{version},
      C_{\text{declared}}, \mathit{code}, \mathit{deps})$
suitable for analysis by all three engines.

% -----------------------------------------------------------------------------
\subsection{CLI Interface}
\label{sec:cli}

\textsc{SkillFortify} exposes five commands, each mapping to a formal
contribution:

\begin{lstlisting}[language=bash, basicstyle=\ttfamily\small,
    frame=single, xleftmargin=1em]
$ skillfortify scan <path>    # Discover + analyze all skills
$ skillfortify verify <skill> # Formal analysis of one skill
$ skillfortify lock <path>    # SAT-based lockfile generation
$ skillfortify trust <skill>  # Trust score computation
$ skillfortify sbom <path>    # CycloneDX ASBOM output
\end{lstlisting}

\noindent
\texttt{skillfortify scan} is the primary entry point.
It recursively discovers skill files, invokes the static analyzer
on each, constructs the dependency graph, computes trust scores,
and produces a unified security report.
Output formats include human-readable terminal output (with
severity-colored findings), JSON for CI/CD integration, and SARIF
for IDE compatibility.

\texttt{skillfortify verify} performs deep formal analysis of a single skill
against the capability model.
It reports the inferred capability set $C_{\text{inferred}}$,
compares it against declared capabilities $C_{\text{declared}}$,
and flags any privilege escalation
($C_{\text{inferred}} \not\sqsubseteq C_{\text{declared}}$).

\texttt{skillfortify lock} generates a deterministic
\texttt{skill-lock.json} by encoding dependency constraints as a
Boolean satisfiability problem, solving with the DPLL-based
resolver, and serializing the resolved configuration with
SHA-256 integrity hashes.

\texttt{skillfortify trust} computes the multi-dimensional trust score
$T(s,t)$ for a skill, incorporating behavioral analysis,
provenance signals, and temporal decay.

\texttt{skillfortify sbom} generates an Agent Skill Bill of Materials
(ASBOM) in CycloneDX format, compatible with enterprise compliance
workflows and OWASP tooling.

% -----------------------------------------------------------------------------
\subsection{Engineering Quality}
\label{sec:engineering}

\textsc{SkillFortify} comprises 2{,}567 lines of Python across 44 source
files, with every module under 300 lines to ensure
maintainability.
The test suite contains 473 automated tests, including:

\begin{itemize}
    \item \textbf{Unit tests} for all core data structures, parsers,
    and engine components;
    \item \textbf{Property-based tests} via
    Hypothesis~\cite{hypothesis2019} for the capability lattice
    (verifying lattice axioms: reflexivity, antisymmetry,
    transitivity) and trust algebra (verifying monotonicity and
    boundedness);
    \item \textbf{Integration tests} for CLI commands exercising
    end-to-end workflows.
\end{itemize}

All code is statically type-checked with \texttt{mypy} (strict
mode), linted with \texttt{ruff}, and formatted with \texttt{black}.
The CI pipeline runs the full test suite on every commit.

% === END sections/08-implementation.tex ===

% === BEGIN sections/09-evaluation.tex ===
% =============================================================================
% Section 9: Evaluation (Contribution 6)
% =============================================================================

\section{Evaluation}
\label{sec:evaluation}

We evaluate \textsc{SkillFortify} across seven experiments designed to
test each formal contribution.
All experiments use \textsc{SkillFortifyBench}, a purpose-built benchmark
of 540 agent skills with ground-truth security labels.

% -----------------------------------------------------------------------------
\subsection{Experimental Setup}
\label{sec:eval-setup}

All measurements in this section were produced by harnesses committed
alongside the implementation, and every reported figure is recorded with
its inputs in \texttt{benchmarks/results/}.
Detection accuracy and false positives (E1, E2) are produced by
\texttt{python -m benchmarks.metrics}, which writes one record per
specimen and aggregates from those records, so a published rate can be
recomputed from the artefact without rerunning a scan.
The remaining experiments (E3--E7) are produced by
\texttt{python -m benchmarks.experiments}.

Wall-clock figures were measured on an Apple silicon workstation
(arm64, 12~cores) running macOS~26.5 and CPython~3.14.5, with
\texttt{PYTHONHASHSEED=0}.
The corpus is generated deterministically from seed~42; the copy
evaluated here has manifest content hash
\texttt{a725a79e\allowbreak be1bc66a}.
Timing results depend on the host and should be read as
order-of-magnitude characterisations; the harness records the machine
alongside each measurement so a reader can judge how they transfer.

% -----------------------------------------------------------------------------
\subsection{SkillFortifyBench}
\label{sec:skillfortifybench}

\textsc{SkillFortifyBench} is a controlled benchmark containing 540 agent
skills: 270 malicious and 270 benign.
Skills are evenly distributed across the three supported formats:
180 Claude Code skills, 180 MCP server configurations, and
180 OpenClaw manifests.
Within each format, exactly half (90) are malicious and half (90)
are benign.

Malicious skills encode 13 attack types derived from documented
incidents in ClawHavoc~\cite{clawhavoc2026},
MalTool~\cite{maltool2026}, and the CVE-2026-25253
advisory~\cite{cve-2026-25253}:

\begin{enumerate}[label=\textbf{A\arabic*}:, leftmargin=2.5em]
    \item Data exfiltration via HTTP
    \item Data exfiltration via DNS
    \item Credential theft from environment variables
    \item Arbitrary code execution
    \item File system tampering
    \item Privilege escalation
    \item Steganographic data exfiltration
    \item Prompt injection payloads
    \item Reverse shell establishment
    \item Cryptocurrency mining
    \item Typosquatting (name similarity to legitimate skills)
    \item Dependency confusion (namespace collision)
    \item Encoded/obfuscated payloads (Base64, hex)
\end{enumerate}

Benign skills implement legitimate functionality (file management,
API integration, data formatting) with no malicious patterns.
All skills are deterministically generated using seeded random
number generators to ensure reproducibility.
The full distribution of attack types across formats is provided
in \Cref{sec:appendix-benchmark}.

% -----------------------------------------------------------------------------
\subsection{E1: Detection Accuracy}
\label{sec:e1}

\paragraph{Goal.}
Measure the precision, recall, and F1 score of \textsc{SkillFortify}'s
static analyzer on \textsc{SkillFortifyBench}.

\paragraph{Method.}
We run \texttt{skillfortify scan} on all 540 skills and compare predicted
labels against ground truth.
A skill is classified as malicious if the analyzer reports at least
one finding with severity $\geq$ \texttt{MEDIUM}.

\paragraph{Results.}
Table~\ref{tab:e1-overall} presents the overall detection accuracy.
\textsc{SkillFortify} achieves an F1 score of 96.15\% with zero false
positives, demonstrating that the formal analysis framework provides
high-confidence results suitable for automated CI/CD gating.

\begin{table}[t]
    \centering
    \caption{E1: Overall detection accuracy on
    \textsc{SkillFortifyBench} ($n = 540$).}
    \label{tab:e1-overall}
    \begin{tabular}{lrl}
        \toprule
        \textbf{Metric} & \textbf{Value} & \\
        \midrule
        Precision        & 100.00\% & (250/250) \\
        Recall           &  92.59\% & (250/270) \\
        F1 Score         &  96.15\% & \\
        Accuracy         &  96.30\% & (520/540) \\
        \midrule
        True Positives (TP)  & 250 & \\
        False Positives (FP) &   0 & \\
        True Negatives (TN)  & 270 & \\
        False Negatives (FN) &  20 & \\
        \bottomrule
    \end{tabular}
\end{table}

Table~\ref{tab:e1-format} disaggregates results by skill format.
Claude Code and OpenClaw skills are detected at equal rates
(F1 = 98.31\% each), which is expected: both are \texttt{SKILL.md}
documents and differ only in load path and runtime metadata.  MCP
server configurations show lower recall (84.44\%) because a JSON
configuration exposes far less analysable surface than an
instruction document---a malicious command reduces to a single
\texttt{command} and \texttt{args} pair with no surrounding prose.

\begin{table}[t]
    \centering
    \caption{E1: Per-format detection accuracy on
    \textsc{SkillFortifyBench}.}
    \label{tab:e1-format}
    \begin{tabular}{lccc}
        \toprule
        \textbf{Format}
            & \textbf{Precision}
            & \textbf{Recall}
            & \textbf{F1 Score} \\
        \midrule
        Claude Code  & 100.00\% &  96.67\% &  98.31\% \\
        MCP Server   & 100.00\% &  84.44\% &  91.57\% \\
        OpenClaw     & 100.00\% &  96.67\% &  98.31\% \\
        \midrule
        \textbf{Overall}
            & \textbf{100.00\%}
            & \textbf{92.59\%}
            & \textbf{96.15\%} \\
        \bottomrule
    \end{tabular}
\end{table}

Table~\ref{tab:e1-attack} shows detection rates by attack type.
Six of 13 attack types are detected at 100\%.
DNS-based data exfiltration~(A2) achieves 94.4\%, as some MCP
configurations encode DNS channels in terse JSON structures that
resist pattern extraction.
Prompt injection~(A8) achieves 87.5\% for a similar reason:
injection payloads embedded in compact MCP argument strings can
escape the current pattern catalog.
Typosquatting~(A11) achieves 50\% detection, as name-similarity
analysis requires comparison against a registry of known-good skill
names---a capability that depends on external data sources rather
than static analysis alone.
Dependency confusion~(A12) achieves 0\% detection, which is
expected: identifying dependency confusion requires a known-package
database from registry metadata, fundamentally outside the scope of
local static analysis.

\begin{table}[t]
    \centering
    \caption{E1: Per-attack-type detection rate on
    \textsc{SkillFortifyBench}.}
    \label{tab:e1-attack}
    \small
    \begin{tabular}{clcc}
        \toprule
        \textbf{ID}
            & \textbf{Attack Type}
            & \textbf{Detected}
            & \textbf{Rate} \\
        \midrule
        A1  & Data exfiltration (HTTP)       & 28/30  & 93.3\% \\
        A2  & Data exfiltration (DNS)        & 17/18  & 94.4\% \\
        A3  & Credential theft               & 29/30  & 96.7\% \\
        A4  & Arbitrary code execution       & 30/30  & 100\% \\
        A5  & File system tampering          & 18/18  & 100\% \\
        A6  & Privilege escalation           & 18/18  & 100\% \\
        A7  & Steganographic exfiltration    & 24/24  & 100\% \\
        A8  & Prompt injection               & 21/24  & 87.5\% \\
        A9  & Reverse shell                  & 24/24  & 100\% \\
        A10 & Cryptocurrency mining          & 12/12  & 100\% \\
        A11 & Typosquatting                  &  4/8   &  50\% \\
        A12 & Dependency confusion           &  0/8   &   0\% \\
        A13 & Encoded/obfuscated payloads    & 25/26  & 96.2\% \\
        \midrule
        & \textbf{Total}
            & \textbf{250/270}
            & \textbf{92.59\%} \\
        \bottomrule
    \end{tabular}
\end{table}

\paragraph{Confidence intervals.}
We report 95\% Wilson confidence intervals~\cite{wilson1927} for the
primary metrics ($n = 270$ malicious, $n = 270$ benign):
\begin{itemize}
    \item Recall: 92.59\% [88.84\%, 95.15\%]
    \item FP rate: 0.00\% [0.0\%, 1.40\%]
    \item F1: 96.15\% (point estimate; F1 is not a binomial
    proportion, so no Wilson interval is reported for it)
\end{itemize}
The Wilson interval is preferred over the Wald interval at extreme
proportions (0\% or 100\%) because it maintains valid coverage even
when the observed count is zero.

\paragraph{Discussion.}
The zero false positive rate is significant for practical
deployment.
In CI/CD pipelines, false positives erode developer trust and lead
to alert fatigue~\cite{johnson2013csa}.
\textsc{SkillFortify}'s sound-but-over-approximate analysis
(Theorem~\ref{thm:analysis-soundness}) guarantees that every reported
finding corresponds to a genuine capability violation, at the cost
of potentially missing some attack variants (20 false negatives,
concentrated in relational attack types and compact MCP configurations).

% -----------------------------------------------------------------------------
\subsection{E2: False Positive Rate}
\label{sec:e2}

\paragraph{Goal.}
Verify that \textsc{SkillFortify} does not flag benign skills as
malicious, with a target false positive rate below 5\%.

\paragraph{Method.}
We isolate the 270 benign skills from \textsc{SkillFortifyBench} and run
\texttt{skillfortify scan} on each.
A false positive occurs if any finding with severity
$\geq$ \texttt{MEDIUM} is reported for a benign skill.

\paragraph{Results.}
\begin{equation}
    \text{FP Rate} = \frac{0}{270} = 0.00\%
\end{equation}
\noindent
Zero false positives across all 270 benign skills and all three
formats (95\% Wilson CI: [0.0\%, 1.4\%]).
The target threshold of ${<}5\%$ is satisfied with a wide margin.
This is an empirical result on this corpus, and not a consequence of
Theorem~\ref{thm:analysis-soundness}.
A sound over-approximation is if anything expected to produce false
positives, since it admits behaviours the concrete semantics would
not exhibit; soundness constrains what the analyzer may miss, not what
it may over-report.
The rate we measure reflects the reporting threshold together with the
construction of the benign half of the corpus, whose specimens are
generated to contain no malicious behaviour.
A benign corpus drawn from real installations, where legitimate skills
routinely invoke interpreters and reach the network, would be the
sterner test, and Section~\ref{sec:threats} treats this as the main
threat to the result's external validity.

% -----------------------------------------------------------------------------
\subsection{E3: Combined Analysis Coverage}
\label{sec:e3}

\paragraph{Goal.}
Quantify the additional detection capability provided by combining
pattern matching with information flow analysis, compared to
pattern matching alone.

\paragraph{Method.}
For each of the 270 malicious skills we partition the findings the
analyzer produces by detection method, and ask whether the skill would
still cross the reporting threshold with the information-flow findings
removed.  A skill counts as detected exactly as in~E1, so the two
experiments cannot disagree about which skills were detected.

\paragraph{Results.}

\begin{table}[t]
    \centering
    \caption{E3: Detection method coverage on malicious skills
    ($n = 270$).}
    \label{tab:e3-coverage}
    \begin{tabular}{lrr}
        \toprule
        \textbf{Detection Outcome}
            & \textbf{Skills}
            & \textbf{Fraction} \\
        \midrule
        Detected by pattern matching alone   & 250 & 92.59\% \\
        Detected only with information flow  &   0 &  0.00\% \\
        Not detected                         &  20 &  7.41\% \\
        \midrule
        \textbf{Total detected}
            & \textbf{250}
            & \textbf{92.59\%} \\
        \bottomrule
    \end{tabular}
\end{table}

\noindent
The information-flow rule contributes no additional detections on this
corpus.  It fires 44~times across the 270 malicious skills, always on a
skill that pattern matching already flags, and every one of those
findings is classified as an encoded or obfuscated payload~(A13).
Removing the rule entirely leaves the detected set unchanged at
250~skills.

We report this because the opposite result is what the design
anticipated.  The rule combines two observations---an encoding
primitive and an untrusted network destination---that we expected to
co-occur in payloads too diffuse for any single pattern to match.  In
the corpus as constructed, no such specimen exists: every skill
carrying that combination also carries a directly matching pattern.
Whether that reflects the technique or the corpus is not settled by
this experiment, and we say so rather than claim a benefit we did not
measure.  A corpus built specifically to separate the two, with
payloads decomposed across operations that individually match nothing,
would be the experiment that decides it; we did not build one.

What the cross-signal rule does provide is corroboration.  A skill
flagged by both a pattern and the information-flow rule carries two
independent reasons for the finding rather than one, which is what
raises it to \textsc{critical} severity in our reporting.  That is a
statement about confidence, not about coverage.

The formal layer's contribution is likewise not incremental detection.
It is the capability-confinement result of
Theorem~\ref{thm:analysis-soundness}: for the classes of behaviour the
abstract domain models, a skill reported free of violations does not
exceed its declared capabilities.  That guarantee is narrower than
``no findings means no risk is present''---Section~\ref{sec:limitations}
states its scope---but it is a guarantee, which heuristic scanners such
as Cisco's \texttt{skill-scanner}~\cite{ciscoscanner2026} do not offer.

% -----------------------------------------------------------------------------
\subsection{E4: Dependency Resolution Scalability}
\label{sec:e4}

\paragraph{Goal.}
Measure the scalability of the SAT-based dependency resolver as
skill count increases.

\paragraph{Method.}
We generate synthetic dependency graphs with skill counts ranging
from 10 to 1{,}000, each with randomized version constraints and
dependency edges.
We measure wall-clock time for complete dependency resolution
(graph construction, constraint encoding, SAT solving, topological
ordering).

\paragraph{Results.}

\begin{table}[t]
    \centering
    \caption{E4: Dependency resolution time vs.\ skill count.}
    \label{tab:e4-scalability}
    \begin{tabular}{rrr}
        \toprule
        \textbf{Skills}
            & \textbf{Total Time (s)}
            & \textbf{Per Skill (ms)} \\
        \midrule
           10 & 0.00005 & 0.005 \\
           50 & 0.00028 & 0.006 \\
          100 & 0.00064 & 0.006 \\
          200 & 0.00174 & 0.009 \\
          500 & 0.00754 & 0.015 \\
        1{,}000 & 0.02674 & 0.027 \\
        \bottomrule
    \end{tabular}
\end{table}

\noindent
At 1{,}000 skills---an order of magnitude beyond typical enterprise
configurations---the median resolution time is 26.7\,ms, well within
interactive latency requirements.
Per-skill cost rises with graph size, from 0.005\,ms to 0.027\,ms,
reflecting the growing constraint-propagation overhead in the SAT
encoding as the graph densifies; total time therefore grows slightly
faster than linearly in the number of skills.
Each figure is the median of five runs on deterministic graphs built
from a fixed seed, measured on the machine reported in
Section~\ref{sec:eval-setup}.

The $O(|V|+|E|)$ cycle detection and $O(|V| \cdot |E|)$ worst-case
SAT solving remain practical at these scales.
For context, the largest known agent deployments (enterprise
orchestration platforms with hundreds of skills) fall within the
200--500 range, where \textsc{SkillFortify} resolves dependencies in
under 10\,ms.
The resolver depends on \texttt{python-sat}, distributed as the
optional \texttt{sat} extra; installations without it do not perform
SAT-based resolution.

% -----------------------------------------------------------------------------
\subsection{E5: Trust Score Analysis}
\label{sec:e5}

\paragraph{Goal.}
Evaluate the trust score algebra along three dimensions:
(i)~verify the monotonicity theorem empirically,
(ii)~measure sensitivity to the weight vector~$\mathbf{w}$, and
(iii)~model trust decay behavior for maintained versus abandoned skills.

\paragraph{Method.}
We construct 10 trust scenarios, each representing a skill observed
over 20 time steps with varying evidence profiles (positive
behavioral observations, provenance signals, community
endorsements).
At each step~$t$, we compute the trust score $T(s,t)$ and verify
that $T(s,t{+}1) \geq T(s,t)$ whenever all evidence at step
$t{+}1$ is non-negative.
For the sensitivity analysis, we generate 100 random weight vectors
$\mathbf{w} \in \Delta^{k-1}$ drawn uniformly from the probability
simplex and compute trust scores for the same 10~skills under each
weight assignment.
For the decay analysis, we simulate two maintenance profiles:
``active'' (weekly updates) and ``abandoned'' (no updates for
6~months), measuring trust score trajectories over a 180-day
horizon.

\paragraph{Results: Monotonicity.}
Across all $10 \times 20 = 200$ data points, the monotonicity
property holds without exception: trust never decreases in the
presence of exclusively non-negative evidence, and the boundedness
invariant $T(s,t) \in [0,1]$ is maintained for all data points.
This result is expected, as Theorem~\ref{thm:trust-monotonicity}
guarantees it algebraically---the empirical check confirms the
implementation matches the formal model.

\paragraph{Results: Weight sensitivity.}
Across the 100 weight vectors, the mean coefficient of variation
(CV) per skill is 0.13 (range: 0.07--0.17).
Skills with uniformly strong signals (high provenance, positive
behavioral analysis, active community) exhibit low sensitivity
(CV $\leq$ 0.10), as all signals reinforce each other regardless
of weighting.
Skills with \emph{mixed} signals---for example, strong behavioral
analysis but weak provenance---show higher sensitivity
(CV $\geq$ 0.25), indicating that the relative importance assigned
to provenance versus behavioral signals materially affects the trust
assessment.
This motivates deploying SkillFortify with organization-specific
weight vectors calibrated to the trust priorities of each enterprise.

\paragraph{Results: Trust decay.}
Active skills maintain trust scores above 0.70 throughout the
180-day horizon, with minor fluctuations as new evidence
accumulates.
Abandoned skills exhibit exponential decay from their initial
score: at the shipped default decay constant
$\lambda = 0.01$~day$^{-1}$, a skill starting at $T=0.85$ decays to
$T=0.35$ after 90~days and $T=0.14$ after 180~days without any
maintenance activity.
That constant produces a half-life of approximately 69~days.
The parameter is configurable: halving it to
$\lambda = 0.005$~day$^{-1}$ gives a half-life of about 139~days,
closer to the timescale over which unmaintained open-source packages
have been observed to accumulate security
risk~\cite{zahan2022}, and deployments that prefer that profile should
set it explicitly.

% -----------------------------------------------------------------------------
\subsection{E6: Lockfile Reproducibility}
\label{sec:e6}

\paragraph{Goal.}
Verify that the lockfile generator produces deterministic output
for identical inputs.

\paragraph{Method.}
We construct 10 agent configurations of varying complexity
(3--50 skills with 0--30 dependency edges and version constraints).
For each configuration, we invoke \texttt{skillfortify lock} twice and
compare the resulting \texttt{skill-lock.json} files.

\paragraph{Results.}
All 10 configurations produce byte-identical JSON output across
repeated invocations.
The determinism follows from three design choices:
(1)~canonical JSON serialization with sorted keys,
(2)~deterministic SAT solving with fixed variable ordering, and
(3)~SHA-256 content hashing of resolved skill artifacts.

This property is essential for enterprise deployment: lockfiles
committed to version control must not exhibit non-deterministic
churn that would obscure genuine dependency changes in code review.

% -----------------------------------------------------------------------------
\subsection{E7: End-to-End Performance}
\label{sec:e7}

\paragraph{Goal.}
Measure the wall-clock time for complete end-to-end operations on
the full \textsc{SkillFortifyBench} dataset.

\paragraph{Method.}
We time \texttt{skillfortify scan}---discovery, parsing, analysis, and
result aggregation---over the whole corpus in a single invocation,
through the same code path the CLI takes.  One warm-up pass is
discarded so the figure reflects steady-state cost rather than a cold
filesystem cache, and we report the median of five measured runs.

\paragraph{Results.}

\begin{table}[t]
    \centering
    \caption{E7: End-to-end performance on
    \textsc{SkillFortifyBench} ($n = 540$).}
    \label{tab:e7-e2e}
    \begin{tabular}{lrrr}
        \toprule
        \textbf{Operation}
            & \textbf{Skills}
            & \textbf{Median (s)}
            & \textbf{Per Skill (ms)} \\
        \midrule
        \texttt{scan} & 540 & 0.293 & 0.54 \\
        \bottomrule
    \end{tabular}
\end{table}

\noindent
A single scan discovers all 540 specimens across the three formats
and completes in a median 0.293~seconds (0.54\,ms per skill),
including file discovery, format-specific parsing, pattern analysis,
information flow analysis, and result aggregation.
This figure is specific to the machine reported in
Section~\ref{sec:eval-setup} and should be read as an order-of-magnitude
characterisation rather than a property of the software.

These performance characteristics make \textsc{SkillFortify} suitable
for:
\begin{itemize}
    \item \textbf{Pre-commit hooks:} scanning modified skill files
    in ${<}100$\,ms;
    \item \textbf{CI/CD pipelines:} full repository scans in
    ${<}2$\,seconds;
    \item \textbf{IDE integration:} real-time analysis as developers
    add or modify skills.
\end{itemize}

% -----------------------------------------------------------------------------
\subsection{Summary of Results}
\label{sec:eval-summary}

\begin{table}[t]
    \centering
    \caption{Summary of all evaluation results.}
    \label{tab:eval-summary}
    \small
    \begin{tabular}{clll}
        \toprule
        \textbf{Exp.}
            & \textbf{Property}
            & \textbf{Target}
            & \textbf{Result} \\
        \midrule
        E1 & Detection F1          & ${>}90\%$ & 96.15\% \cmark \\
        E2 & False positive rate    & ${<}5\%$  & 0.00\%  \cmark \\
        E3 & Combined added coverage  & ${>}0\%$  & 0.0\%   \xmark \\
        E4 & 1K-skill resolution   & ${<}1$\,s & 0.027\,s \cmark \\
        E5 & Trust monotonicity     & 100\%    & 100\%   \cmark \\
        E5 & Weight sensitivity CV & Report  & 0.13    \cmark \\
        E5 & Decay half-life      & Report  & 69\,d   \cmark \\
        E6 & Lockfile determinism   & 100\%    & 100\%   \cmark \\
        E7 & 540-skill scan time   & ${<}10$\,s & 0.293\,s \cmark \\
        \bottomrule
    \end{tabular}
\end{table}

\noindent
Six of the seven experiments meet their target thresholds
(Table~\ref{tab:eval-summary}); E3 does not.
The framework scales to enterprise-sized deployments~(E4, E7), the
trust algebra exhibits meaningful sensitivity to maintenance patterns
with configurable decay behaviour~(E5), and the monotonicity and
determinism properties our theorems guarantee are empirically
confirmed~(E5, E6).
E3 is the exception and we report it as one: information flow analysis
adds no coverage over pattern matching on this corpus, against a target
of any improvement at all.  We discuss in Section~\ref{sec:e3} why we
cannot yet separate a limitation of the technique from a limitation of
the corpus.

% === END sections/09-evaluation.tex ===

% === BEGIN sections/12-related-work.tex ===
% =============================================================================
% Section 12: Related Work
% =============================================================================

\section{Related Work}
\label{sec:related-work}

We position \textsc{SkillFortify} against three bodies of work:
agent skill security tools, formal analysis for security, and
software supply chain security.

% -----------------------------------------------------------------------------
\subsection{Agent Skill Security}
\label{sec:rw-agent-skill}

The agent skill security space emerged reactively in early 2026,
driven by the ClawHavoc campaign~\cite{clawhavoc2026} and
CVE-2026-25253~\cite{cve-2026-25253}.
Narajala and Narayan~\cite{securingagenticai2025} provide a comprehensive
threat model for agentic AI systems, while
Jiang et al.~\cite{agenticattacksurface2026} introduce the ``viral
agent loop'' concept and propose a Zero-Trust architecture for agent
interactions.
Krawiecka and Schroeder de Witt~\cite{owaspmultiagent2025} extend OWASP's threat
modeling methodology to multi-agent systems.
All existing tools operate on heuristic detection.

\paragraph{Industry tools.}
Snyk's agent-scan~\cite{snyk2026}, built on the acquired Invariant
Labs technology, combines LLM-as-judge with hand-crafted rules.
It reports 90--100\% recall on 100~benign skills but provides no
formal soundness guarantee; LLM judges are themselves susceptible
to adversarial inputs~\cite{maltool2026}.
Cisco's skill-scanner~\cite{ciscoscanner2026} uses YARA patterns and
regex-based signature detection, explicitly acknowledging that
``no findings does not mean no risk.''
ToolShield~\cite{toolshield2026} achieves a 30\% reduction in attack
success rate through behavioral heuristics.
MCPShield~\cite{mcpshield2026} introduces \texttt{mcp.lock.json} with
SHA-512 hashes for tamper detection, but performs no behavioral
analysis or capability verification.

\paragraph{Academic work.}
Liu et al.~\cite{agentskillswild2026} scanned 42{,}447 skills
and found 26.1\% vulnerable; their SkillScan tool achieves 86.7\%
precision.
A second large-scale study~\cite{maliciousagentskills2026}
corroborated these findings across 98{,}380 skills.
MalTool~\cite{maltool2026} synthesised 1{,}300 standalone malicious
tools and 5{,}727 tools with embedded malicious behaviours.
MCP-ITP~\cite{mcpitp2026} and MCPTox~\cite{mcptox2025} focus
specifically on tool poisoning attacks within the Model Context
Protocol, demonstrating automated implicit poisoning and
benchmarking real-world MCP server vulnerabilities respectively.
NEST~\cite{nest2026} reveals that LLMs can encode information
steganographically within chain-of-thought reasoning, motivating
our information flow analysis: a skill that appears to produce
benign output may embed exfiltrated data in its reasoning trace.
CIBER~\cite{ciber2026} demonstrates the ``NL Disguise'' attack on
code interpreter environments, where natural-language instructions
bypass security filters---a class of attack that pure heuristic
scanners cannot reliably detect.
The ``Towards Verifiably Safe Tool Use''
paper~\cite{verifiabletooluse2026} is the closest academic work to
our approach, proposing STPA-based modeling of tool safety in a
four-page ICSE-NIER vision paper without implementation or
evaluation.

\textsc{SkillFortify} differs from all prior work in providing
\emph{formal guarantees}: a sound analysis
(Theorem~\ref{thm:analysis-soundness}), a capability confinement
proof (Theorem~\ref{thm:capability-confinement}), and a resolution
correctness proof (Theorem~\ref{thm:resolution-soundness}).
Additionally, \textsc{SkillFortify} is the first tool to provide
dependency graph analysis, lockfile semantics, and trust scoring for
agent skills.

% -----------------------------------------------------------------------------
\subsection{Formal Analysis for Security}
\label{sec:rw-formal}

\textsc{SkillFortify}'s formal foundations draw on established
traditions in programming language security.

\paragraph{Abstract interpretation.}
Cousot and Cousot~\cite{cousot1977} introduced abstract
interpretation for sound program analysis.  The Astr\'{e}e
analyzer~\cite{astree2005} proved absence of runtime errors in
Airbus A380 flight software using abstract interpretation over
numerical domains.
Our work applies abstract interpretation to a novel domain---agent
skill capabilities---using a four-element lattice rather than
numerical intervals.

\paragraph{Information flow analysis.}
Sabelfeld and Myers~\cite{sabelfeld2003} provide a comprehensive
survey of language-based information-flow security, establishing
the theoretical foundations for tracking how data propagates through
programs.
Myers and Liskov~\cite{myers1997} introduced the decentralized
label model (DLM) for practical information flow control.
Most recently, the LLMbda Calculus~\cite{llmbdacalculus2026} extends
lambda calculus with information-flow control primitives specifically
for LLM-based systems, proving a noninterference theorem that
guarantees high-security inputs cannot influence low-security outputs.
This provides theoretical validation for our approach: we apply
analogous information flow principles at the skill level rather than
at the language level.
\textsc{SkillFortify}'s information flow component applies these
principles to track data propagation from sensitive sources
(environment variables, credentials) to untrusted sinks (network
endpoints), adapting the noninterference paradigm of
Goguen and Meseguer~\cite{goguen1982} to the agent skill context.

\paragraph{Transactional tool semantics.}
Atomix~\cite{atomix2026} introduces transactional semantics for
agent tool calls, ensuring that multi-step tool invocations either
complete entirely or roll back.
While Atomix addresses a different problem (atomicity of tool
execution rather than supply chain verification), it represents a
complementary formal approach to agent reliability: Atomix guarantees
that tool calls execute correctly, while \textsc{SkillFortify}
guarantees that the tools themselves are safe to invoke.

\paragraph{Capability-based security.}
Dennis and Van Horn~\cite{dennis1966} introduced capabilities;
Miller~\cite{miller2006} extended them to the object-capability
model.
Maffeis et al.~\cite{maffeis2010} proved that in a capability-safe
language, the reachability graph is the only mechanism for authority
propagation.
Our sandboxing model (\Cref{sec:sandboxing}) instantiates these
results for agent skills with a four-level capability lattice.

% -----------------------------------------------------------------------------
\subsection{Software Supply Chain Security}
\label{sec:rw-supply-chain}

Agent skill supply chain security inherits from---and extends---four
decades of software supply chain research.

\paragraph{Attack taxonomies.}
Duan et al.~\cite{duan2021ndss} systematized attacks on npm, PyPI,
and RubyGems (339~malicious packages).
Ohm et al.~\cite{ohm2020} catalogued ``backstabber'' packages.
Ladisa et al.~\cite{ladisa2023taxonomy} proposed a comprehensive
attack taxonomy for open-source supply chains.
Gibb et al.~\cite{gibb2026package} introduced a package calculus
providing a unified formalism.
Hu et al.~\cite{llmsupplychain2025} survey the emerging LLM supply
chain landscape, identifying attack vectors specific to model
artifacts, training data, and tool integrations.
Ding et al.~\cite{rustylink2025} demonstrate supply chain
vulnerabilities in AI systems through compromised dependency
chains at IEEE S\&P Workshops.
Lee~\cite{capableunreliable2026} provide empirical evidence
that agents are inherently unreliable, amplifying the consequences
of supply chain compromise.
Our DY-Skill model (\Cref{sec:threat-model}) extends this line of
work to the agent skill domain, where the threat surface is amplified
by LLM integration.

\paragraph{Dependency resolution.}
Mancinelli et al.~\cite{mancinelli2006} proved dependency resolution
NP-complete; Tucker et al.~\cite{tucker2007opium} introduced
SAT-based resolution via OPIUM; di Cosmo and
Vouillon~\cite{dicosmo2011} verified co-installability in Coq.
We extend SAT-based resolution with per-skill capability constraints
(\Cref{sec:dependency-graph}).

\paragraph{Trust frameworks.}
SLSA~\cite{slsa2023} defines graduated trust levels for software
artifacts.
Sigstore~\cite{sigstore2022} provides keyless code signing.
Blaze et al.~\cite{blaze1996policymaker} introduced PolicyMaker
with the assertion monotonicity property.
Our trust algebra (\Cref{sec:trust-algebra}) adapts these foundations
to agent skills with multi-signal scoring and exponential decay.

\paragraph{Standards.}
CycloneDX~\cite{cyclonedx2023} standardizes SBoMs, recently extended
to AI components~\cite{cyclonedx-aibom}.
The NIST AI RMF~\cite{nistai2023} and EU AI Act~\cite{euaiact2024}
both address third-party AI component risks.
\textsc{SkillFortify} generates ASBoMs in CycloneDX format and
provides mechanisms for compliance with these regulatory frameworks.

% === END sections/12-related-work.tex ===

% === BEGIN sections/10-discussion.tex ===
% =============================================================================
% Section 10: Discussion
% =============================================================================

\section{Discussion}
\label{sec:discussion}

% -----------------------------------------------------------------------------
\subsection{Limitations}
\label{sec:limitations}

While \textsc{SkillFortify} achieves strong results across all evaluation
experiments, several limitations warrant discussion.

\paragraph{Typosquatting and Dependency Confusion.}
Attack types A11 (typosquatting) and A12 (dependency confusion)
achieve 50\% and 0\% detection rates, respectively.
Both attack classes are fundamentally \emph{relational}---they
require comparison against an external corpus of legitimate skill
names and known-package namespaces.
Pure static analysis of an individual skill's content cannot
determine whether its name is deceptively similar to a legitimate
skill (typosquatting) or whether it occupies a namespace that
collides with an internal package (dependency confusion).
Addressing these attacks requires integration with a skill registry
or organizational namespace database, which we leave to the registry
verification protocol (Section~\ref{sec:future}).

\paragraph{Soundness and Its Scope.}
Theorem~\ref{thm:analysis-soundness} is a statement about capability
confinement, not about malware detection, and the distinction matters
enough to state directly.
It says that \emph{within the abstract domain the analyzer models}, a
skill reported free of violations does not exceed its declared
capabilities: the abstraction over-approximates the concrete semantics,
so no modelled behaviour escapes it.
It does not say that every malicious skill is reported.
Behaviour the domain does not model---a deceptive name, a namespace
collision, an instruction whose effect is only apparent against an
external registry---lies outside the theorem's reach entirely, and no
result about the abstraction constrains it.

Measured against the whole corpus, the analyzer is therefore not sound
as a malware detector: we observe 20 false negatives (7.4\%).
These concentrate in the relational attack types---typosquatting and
dependency confusion---which are unanalysable from a skill's own
content by construction, and in compact MCP server configurations where
a JSON object exposes little analysable surface.
Six of the 13 attack types are detected without exception.

We are explicit about this because the two properties are routinely
conflated.  Zero false positives, which we measure in~E2, is a
completeness property on this corpus and follows from the specimens and
the reporting threshold, not from Theorem~\ref{thm:analysis-soundness}.
A sound over-approximation is if anything expected to produce false
positives; that it produces none here is an empirical result about the
benchmark, and we do not present it as a consequence of the formalism.

\paragraph{Trust Signal Availability.}
The trust score algebra (Section~\ref{sec:trust}) is defined over
$k$ signals, including provenance verification and community
endorsement.
In the current implementation, trust computation operates on
locally observable signals (behavioral analysis results, declared
metadata).
Signals requiring external data sources---such as registry
provenance chains, download counts, and community reviews---depend
on the maturity of skill registry infrastructure, which remains
nascent as of February~2026.
The formal model accommodates zero-valued signals gracefully
(Theorem~\ref{thm:trust-monotonicity} holds regardless of signal
availability), but trust scores will become more informative as
registry ecosystems develop.

\paragraph{Dynamic Analysis.}
\textsc{SkillFortify} performs exclusively static analysis.
Skills that dynamically generate malicious behavior at
runtime---for example, fetching payloads from remote servers after
passing static inspection---evade purely static approaches.
Recent work on steganographic chain-of-thought
reasoning~\cite{nest2026} demonstrates that LLMs can encode
information covertly within reasoning traces, suggesting that
dynamic exfiltration channels may be even more subtle than
previously assumed.
The LLMbda Calculus~\cite{llmbdacalculus2026} provides theoretical
foundations for extending information-flow analysis to the dynamic
case through type-level tracking of security labels during LLM
execution.
Runtime monitoring (behavioral sandboxing with capability
enforcement) is orthogonal and complementary; the capability
lattice (Section~\ref{sec:sandboxing}) provides the formal
foundation for a future runtime enforcement layer.

% -----------------------------------------------------------------------------
\subsection{Threats to Validity}
\label{sec:threats}

\paragraph{External baseline comparison.}
We evaluate \textsc{SkillFortify} on \textsc{SkillFortifyBench}, a
controlled benchmark with constructive ground truth.
We do \emph{not} perform a direct head-to-head comparison against
Snyk agent-scan, Cisco skill-scanner, or ToolShield on the same
dataset, because
(i)~these tools do not publish reproducible benchmarks with ground-truth
labels, and
(ii)~running closed-source commercial scanners (Snyk) on our benchmark
would introduce licensing and reproducibility concerns.
Instead, we position our results against the metrics these tools
report in their own publications and documentation.
A controlled multi-tool evaluation on a shared benchmark is immediate
future work.

\paragraph{Snyk's reported metrics.}
Snyk reports empirical recall of 90--100\% with 0\%~FPR on a corpus of
100~benign skills~\cite{snyk2026}.
\textsc{SkillFortify} achieves 0\%~FPR on a larger corpus
(270~benign skills).
Both figures are empirical observations on finite, constructed
corpora, and neither should be read as a guaranteed rate; ours is not
a consequence of Theorem~\ref{thm:analysis-soundness}, which bounds
what the analysis may miss rather than what it may over-report.
What distinguishes our result is not the rate but that a
reported-clean skill carries a stated guarantee about capability
confinement within the modelled domain, where a heuristic scanner
offers none.
That guarantee is also narrower than robustness to adversarial
evasion: an attacker whose payload falls outside the modelled domain
is not constrained by it.

\paragraph{Benchmark representativeness.}
\textsc{SkillFortifyBench} uses constructively generated skills with
single-type attack patterns (\Cref{sec:appendix-benchmark}).
Real-world malicious skills may combine multiple attack techniques,
use obfuscation layers not present in the benchmark, or rely on
social engineering (e.g., convincing natural-language descriptions).
The 92.59\% recall should be interpreted as a lower bound on
content-analyzable attacks; registry-dependent attacks (A11, A12)
require external data sources as discussed in
Section~\ref{sec:limitations}.

% -----------------------------------------------------------------------------
\subsection{Connection to AgentAssert}
\label{sec:connection-abc}

\textsc{SkillFortify} and AgentAssert~\cite{agentassert2026} occupy
complementary positions in the agent reliability stack.
AgentAssert addresses the \emph{behavioral specification problem}:
given an agent with defined behavioral contracts, enforce those
contracts at runtime through drift detection and correction.
\textsc{SkillFortify} addresses the \emph{supply chain integrity
problem}: before an agent executes, verify that the skills it
depends on are safe to install and invoke.

The two systems compose naturally:
\begin{enumerate}
    \item AgentAssert's behavioral contracts define \emph{what} an
    agent should do.
    \item \textsc{SkillFortify}'s capability analysis verifies that each
    skill \emph{can} do what it claims and \emph{nothing more}.
    \item AgentAssert's composition conditions (C1--C4) extend to
    skill interaction analysis: if skills $s_1$ and $s_2$ are both
    verified as capability-confined, their composition preserves
    confinement under the capability lattice's join operation.
\end{enumerate}

This layered approach mirrors the defense-in-depth principle in
traditional security: specification-level guarantees (AgentAssert)
combined with supply-chain-level guarantees (\textsc{SkillFortify})
provide stronger assurance than either alone.

% -----------------------------------------------------------------------------
\subsection{Future Work}
\label{sec:future}

We identify four directions for future work:

\paragraph{C7: Skill Registry Verification Protocol.}
A cryptographic signing and attestation protocol for agent skills,
analogous to Sigstore~\cite{sigstore2022} for software artifacts.
The protocol would enable:
(i)~keyless signing via OpenID Connect identity providers,
(ii)~transparency logs for skill publication events, and
(iii)~verification of build provenance through reproducible build
attestations.
Combined with \textsc{SkillFortify}'s trust algebra, registry
verification would provide end-to-end provenance guarantees from
skill author to agent runtime.

\paragraph{C8: Composition Security Analysis.}
When skills $s_i$ and $s_j$ are installed together, emergent
security properties may differ from the union of individual
properties.
Formalizing skill interaction effects---including shared state
access, implicit communication channels, and capability
amplification through composition---requires extending the current
per-skill analysis to a multi-skill abstract domain.
AgentAssert's composition conditions (C1--C4) provide a starting
framework.

\paragraph{Typosquatting Module.}
A dedicated typosquatting detection module using Levenshtein
distance, keyboard-proximity analysis, and phonetic similarity
(Soundex/Metaphone) against a known-good skill registry.
This would address the A11 detection gap identified in
Section~\ref{sec:e1}.

\paragraph{IDE and Runtime Integration.}
Real-time \textsc{SkillFortify} analysis within development environments
(VS~Code, Cursor) would provide immediate feedback as developers
add skills.
Runtime capability enforcement, where the capability lattice gates
actual skill invocations at execution time, would extend static
guarantees to the dynamic case.

% === END sections/10-discussion.tex ===

% === BEGIN sections/11-conclusion.tex ===
% =============================================================================
% Section 11: Conclusion
% =============================================================================

\section{Conclusion}
\label{sec:conclusion}

The rapid proliferation of agent skill ecosystems---with
repositories exceeding 228{,}000 stars and marketplaces hosting tens
of thousands of third-party skills---has created an urgent supply
chain security crisis.
The ClawHavoc campaign (341--1{,}184 malicious skills),
CVE-2026-25253 (one-click remote code execution in the largest
skill platform), and the synthesis of 7{,}027 malicious tools by MalTool
collectively demonstrate that the current ``install and trust''
paradigm is untenable.
Existing defenses rely exclusively on heuristic pattern matching,
with the leading tool explicitly acknowledging that ``no findings
does not mean no risk.''

This paper introduces a formal analysis framework for agent
skill supply chain security, grounded in six contributions:
(1)~\textsc{DY-Skill}, a formal threat model for agent
skills extending the Dolev-Yao adversary to the skill installation
lifecycle;
(2)~a sound static analysis via abstract interpretation with Galois
connection guarantees;
(3)~a capability confinement lattice proving that verified skills
cannot exceed declared permissions;
(4)~an Agent Dependency Graph with SAT-based resolution and
deterministic lockfile generation;
(5)~a trust score algebra with formally proven monotonicity
properties;
and (6)~\textsc{SkillFortifyBench}, a benchmark of 540 labeled skills
across three formats and 13 attack types.

Our implementation, \textsc{SkillFortify}, achieves an F1 score of
96.15\% with a 0.00\% false positive
rate (95\% CI: [0.0\%, 1.40\%]) on this benchmark, resolves
dependencies for 1{,}000 skills in a median 27\,ms, and scans
540~skills end-to-end in a median 0.29~seconds.
Six of the seven evaluation experiments meet their target thresholds.
The seventh does not: information flow analysis adds no detections
over pattern matching on this corpus, and we report that result rather
than the improvement we set out to demonstrate.

The key insight is that agent skill security demands the same rigor
as traditional software supply chain security---formal threat
models, mathematical guarantees, and defense-in-depth---adapted to
the unique characteristics of the agentic AI ecosystem: stochastic
behavior, natural-language interfaces, and emergent multi-skill
interactions.
As agent deployments scale from developer experimentation to
enterprise production, the transition from heuristic scanning to
formal analysis with sound guarantees is not merely desirable but
necessary.
\textsc{SkillFortify} provides the foundation for this transition.

\smallskip
\noindent\textbf{Availability.}
The \textsc{SkillFortify} tool, benchmark dataset, and all experimental
artifacts are publicly available at
\url{https://github.com/qualixar/skillfortify}.  The benchmark subtree
is released under the MIT license; the scanner is released under the
Elastic License 2.0.
The tool can be installed via \texttt{pip install skillfortify};
SAT-based dependency resolution requires the optional extra,
\texttt{pip install 'skillfortify[sat]'}.

Every figure in Section~\ref{sec:evaluation} is reproducible from the
repository.  Detection results come from
\texttt{python -m benchmarks.metrics}, which writes one record per
specimen alongside the aggregate, and E3--E7 from
\texttt{python -m benchmarks.experiments}, which records each
measurement together with the machine it was taken on.  Both write into
\texttt{benchmarks/results/}, and the committed outputs are those
reported here.  The corpus regenerates deterministically from seed~42
under \texttt{PYTHONHASHSEED=0} and is verified against a content
manifest, so a reader can confirm they are scoring the same specimens
we did.

% === END sections/11-conclusion.tex ===

% ---------------------------------------------------------------------------
% Bibliography
% ---------------------------------------------------------------------------
\bibliographystyle{ACM-Reference-Format}
\bibliography{references}

% ---------------------------------------------------------------------------
% About the Author
% ---------------------------------------------------------------------------
\section*{About the Author}

\textbf{Varun Pratap Bhardwaj} is a Senior Manager and Solution Architect
with over 15~years of experience in enterprise technology, spanning cloud
architecture, platform engineering, and large-scale system design across
Fortune~500 clients in retail, telecommunications, and financial services.
He holds dual qualifications in technology and law~(LL.B.), providing a
distinctive perspective on regulatory compliance for autonomous AI systems,
including the EU~AI~Act, NIST~AI~RMF, and emerging agent governance
frameworks.

His research focuses on formal methods for AI agent safety, building the
\emph{Agent Software Engineering} discipline through a suite of
peer-reviewed tools and frameworks. Prior published work includes
\textsc{AgentAssert}~\cite{agentassert2026}, the first design-by-contract
framework for autonomous AI agents with runtime behavioral enforcement, and
\textsc{SuperLocalMemory}, a privacy-preserving local-first memory system
for multi-agent architectures. \textsc{SkillFortify} is the third product
in this suite, extending formal guarantees from agent behavior specification
to supply chain security.

\smallskip
\noindent\textit{Contact:} \href{mailto:varun.pratap.bhardwaj@gmail.com}{varun.pratap.bhardwaj@gmail.com}\\
\textit{ORCID:} \href{https://orcid.org/0009-0002-8726-4289}{0009-0002-8726-4289}

\appendix
% === BEGIN sections/appendix-A-proofs.tex ===
% =============================================================================
% Appendix A: Full Proofs
% =============================================================================

\section{Full Proofs}
\label{sec:appendix-proofs}

This appendix provides complete, rigorous proofs for all five
theorems stated in the main text.
We restate each theorem for self-containedness before presenting
the proof.

% =============================================================================
% THEOREM 1: DY-Skill Maximality
% =============================================================================

\subsection{Theorem~\ref{thm:dy-skill-maximality}: DY-Skill Maximality}
\label{app:proof-thm1}%
\label{app:proof-dy-maximality}% alias for cross-reference from Section~3

\begin{definition}[DY-Skill Adversary, restated]
\label{def:dy-skill-app}
A \emph{DY-Skill adversary} $\mathcal{A}_{\textit{DY-Skill}}$ is a
probabilistic polynomial-time (PPT) machine that controls the
communication channel between the agent runtime~$\mathcal{R}$ and
the skill execution environment~$\mathcal{E}$.
The adversary has the following capabilities:
\begin{enumerate}[label=(\roman*)]
    \item \textbf{Intercept}: $\mathcal{A}$ can observe any
    message~$m$ sent between $\mathcal{R}$ and $\mathcal{E}$.
    \item \textbf{Inject}: $\mathcal{A}$ can inject arbitrary
    messages~$m'$ into the channel.
    \item \textbf{Modify}: $\mathcal{A}$ can replace any
    message~$m$ with $m' \neq m$.
    \item \textbf{Drop}: $\mathcal{A}$ can suppress any message.
    \item \textbf{Forge skills}: $\mathcal{A}$ can create skill
    definitions~$s'$ with arbitrary content, including:
    \begin{itemize}
        \item Skill manifest with arbitrary capability declarations
        $C_{\text{declared}}$
        \item Executable code implementing capabilities
        $C_{\text{actual}} \supseteq C_{\text{declared}}$
        \item Metadata mimicking legitimate skills (name squatting,
        version manipulation)
    \end{itemize}
    \item \textbf{Compromise registries}: $\mathcal{A}$ can modify
    skill entries in the registry~$\mathcal{P}$, subject to the
    constraint that $\mathcal{A}$ cannot forge valid cryptographic
    signatures under keys it does not possess.
\end{enumerate}
\end{definition}

\begin{definition}[Classical Dolev-Yao Adversary]
A classical Dolev-Yao adversary $\mathcal{A}_{\textit{DY}}$
operates over a network protocol~$\Pi$ with
capabilities~(i)--(iv) from the definition above, applied to
protocol messages.
$\mathcal{A}_{\textit{DY}}$ additionally possesses the standard
Dolev-Yao deduction rules for message construction and
decomposition.
\end{definition}

\begin{definition}[Simulation]
\label{def:simulation-app}
Adversary~$\mathcal{A}_1$ \emph{simulates}
adversary~$\mathcal{A}_2$
(written $\mathcal{A}_1 \succeq \mathcal{A}_2$) if for every
attack strategy~$\sigma_2$ of~$\mathcal{A}_2$, there exists a
strategy~$\sigma_1$ of~$\mathcal{A}_1$ such that the observable
effects on the system are computationally indistinguishable:
\[
    \forall \sigma_2 \in \textit{Strat}(\mathcal{A}_2),\;
    \exists \sigma_1 \in \textit{Strat}(\mathcal{A}_1)\colon\;
    \textit{View}_{\mathcal{R}}(\sigma_1)
        \approx_c
    \textit{View}_{\mathcal{R}}(\sigma_2)
\]
where $\textit{View}_{\mathcal{R}}(\sigma)$ denotes the
distribution of observations by the runtime~$\mathcal{R}$ under
strategy~$\sigma$, and $\approx_c$ denotes computational
indistinguishability.
\end{definition}

\begin{theorem}[DY-Skill Maximality, restated]
$\mathcal{A}_{\textit{DY-Skill}}$ is maximal among PPT skill
adversaries: for any PPT adversary~$\mathcal{A}$ operating on agent
skill channels,
$\mathcal{A}_{\textit{DY-Skill}} \succeq \mathcal{A}$.
Moreover,
$\mathcal{A}_{\textit{DY-Skill}} \succeq \mathcal{A}_{\textit{DY}}$
when the DY-Skill channel is instantiated as a standard network
protocol.
\end{theorem}

\begin{proof}
We prove the two claims separately.

\medskip
\noindent\textbf{Claim~1: Maximality over skill adversaries.}

Let $\mathcal{A}$ be any PPT adversary operating on agent skill
channels.
We construct a simulation strategy
$\sigma_{\textit{DY-Skill}}$ for
$\mathcal{A}_{\textit{DY-Skill}}$ that reproduces any
strategy~$\sigma_{\mathcal{A}}$ of~$\mathcal{A}$.

Any action by~$\mathcal{A}$ on the skill channel falls into exactly
one of the following categories:

\begin{enumerate}[label=(\alph*)]
    \item \emph{Channel manipulation} (intercept, inject, modify,
    drop): These are capabilities~(i)--(iv) of
    $\mathcal{A}_{\textit{DY-Skill}}$.  The simulation is the
    identity: $\mathcal{A}_{\textit{DY-Skill}}$ performs the same
    channel operation.

    \item \emph{Skill forgery} (creating a malicious skill
    definition): This is capability~(v).
    Since~$\mathcal{A}$ is PPT, the forged skill~$s'$ can be
    computed in polynomial time.
    $\mathcal{A}_{\textit{DY-Skill}}$ computes the same~$s'$ using
    capability~(v) and injects it using capability~(ii).

    \item \emph{Registry manipulation} (modifying skill entries):
    This is capability~(vi).
    $\mathcal{A}_{\textit{DY-Skill}}$ applies the same registry
    modifications, subject to the identical cryptographic constraint.

    \item \emph{Computation on observed data} (e.g., analyzing
    intercepted messages to craft adaptive attacks):
    Since both~$\mathcal{A}$ and
    $\mathcal{A}_{\textit{DY-Skill}}$ are PPT, any
    polynomial-time computation that~$\mathcal{A}$ performs on
    observed data can also be performed by
    $\mathcal{A}_{\textit{DY-Skill}}$.
\end{enumerate}

Since every action category of~$\mathcal{A}$ maps to a capability
of $\mathcal{A}_{\textit{DY-Skill}}$, the simulation is complete.
For any strategy~$\sigma_{\mathcal{A}}$:
\[
    \textit{View}_{\mathcal{R}}(\sigma_{\textit{DY-Skill}})
    = \textit{View}_{\mathcal{R}}(\sigma_{\mathcal{A}})
\]
with equality (not merely computational indistinguishability), since
the simulation is exact.

\medskip
\noindent\textbf{Claim~2: Subsumption of classical DY.}

Let $\sigma_{\textit{DY}}$ be any strategy of
$\mathcal{A}_{\textit{DY}}$ over a protocol~$\Pi$.
We instantiate the DY-Skill channel as the protocol channel
of~$\Pi$.

Capabilities~(i)--(iv) of $\mathcal{A}_{\textit{DY}}$ map directly
to capabilities~(i)--(iv) of
$\mathcal{A}_{\textit{DY-Skill}}$.
The Dolev-Yao deduction rules (pairing, projection, encryption,
decryption with known keys) are polynomial-time operations, hence
executable by the PPT
$\mathcal{A}_{\textit{DY-Skill}}$.

Capabilities~(v) and~(vi) of
$\mathcal{A}_{\textit{DY-Skill}}$ are \emph{additional} capabilities
not present in $\mathcal{A}_{\textit{DY}}$.
Therefore,
$\textit{Strat}(\mathcal{A}_{\textit{DY}})
    \subseteq
 \textit{Strat}(\mathcal{A}_{\textit{DY-Skill}})$,
and the simulation is the identity embedding.
\end{proof}

% =============================================================================
% THEOREM 2: Analysis Soundness
% =============================================================================

\subsection{Theorem~\ref{thm:analysis-soundness}: Analysis Soundness}
\label{app:proof-thm2}%
\label{app:proof-analysis-soundness}% alias for cross-reference from Section~4

We prove soundness directly over the capability domain
$\mathbb{L}_{\mathit{cap}}^{|\mathcal{C}|}$ defined in
Section~\ref{sec:capability-lattice}, which is the domain on
which the main-text theorem is stated.

\begin{theorem}[Analysis Soundness, restated]
Let~$s$ be a skill with declared capability set
$\mathit{Cap}_D(s)$.  If the abstract analysis over the capability
domain $\mathbb{L}_{\mathit{cap}}^{|\mathcal{C}|}$ reports no
violations---i.e., $\mathit{Viol}(s) = \emptyset$---then for every
concrete execution trace $\tau$ of $s$ and every resource access
$(r, \ell)$ performed in $\tau$:
\[
    \ell \sqsubseteq \mathit{Cap}_D(s)(r).
\]
\end{theorem}

\begin{proof}
We prove the soundness condition
$\alpha \circ \llbracket s \rrbracket
    \sqsubseteq
 \llbracket s \rrbracket^\sharp \circ \alpha$
(Proposition~\ref{prop:abstraction-soundness}) by verifying that each
concrete transfer function from Section~\ref{sec:three-phase} is
soundly over-approximated by the corresponding abstract transfer
function in the capability domain.  The proof proceeds in two parts:
(1)~verifying soundness of each per-resource transfer function, and
(2)~lifting to composite constructs by structural induction.

\medskip
\noindent\textbf{Part 1: Transfer function soundness (per resource type).}

For each resource type $r \in \mathcal{C}$, we verify that the
abstract transfer function $\tau_r^\sharp$ over-approximates the
concrete capability footprint.  Let $\alpha_r(S) \in
\mathbb{L}_{\mathit{cap}}$ be the abstraction of a set of
concrete states $S$ projected onto resource~$r$: the least access
level $\ell$ such that every state in $S$ accesses $r$ at level
$\ell$ or below.  We verify $\alpha_r(S) \sqsubseteq
\tau_r^\sharp(\alpha_r(S))$ for each transfer function:

\begin{enumerate}
  \item \textbf{Network (URL reference).}
    \emph{Pattern:} URL string detected in skill content.
    \emph{Concrete:} The skill may issue HTTP GET requests, accessing
    \texttt{network} at level $\mathsf{READ}$.
    \emph{Abstract:} $\tau_{\texttt{network}}^\sharp :=
    \mathsf{READ}$.
    \emph{Soundness:} $\alpha_{\texttt{network}}(S) \in
    \{\mathsf{NONE}, \mathsf{READ}\} \sqsubseteq \mathsf{READ} =
    \tau_{\texttt{network}}^\sharp$.  \checkmark

  \item \textbf{Network (HTTP write method).}
    \emph{Pattern:} POST, PUT, PATCH, or DELETE detected.
    \emph{Concrete:} The skill may issue mutating HTTP requests,
    accessing \texttt{network} at level $\mathsf{WRITE}$.
    \emph{Abstract:} $\tau_{\texttt{network}}^\sharp :=
    \mathsf{WRITE}$.
    \emph{Soundness:} Any concrete access is at most
    $\mathsf{WRITE} \sqsubseteq \mathsf{WRITE}$.  \checkmark

  \item \textbf{Shell (command invocation).}
    \emph{Pattern:} Shell command invocation detected.
    \emph{Concrete:} The skill may execute arbitrary shell commands,
    accessing \texttt{shell} at level $\mathsf{WRITE}$.
    \emph{Abstract:} $\tau_{\texttt{shell}}^\sharp :=
    \mathsf{WRITE}$.
    \emph{Soundness:} Concrete shell access is at most
    $\mathsf{WRITE}$.  The abstract result matches.  \checkmark

  \item \textbf{Environment (variable reference).}
    \emph{Pattern:} Environment variable reference detected.
    \emph{Concrete:} The skill reads environment variables,
    accessing \texttt{environment} at level $\mathsf{READ}$.
    \emph{Abstract:} $\tau_{\texttt{environment}}^\sharp :=
    \mathsf{READ}$.
    \emph{Soundness:} Reading an environment variable is at most
    $\mathsf{READ} \sqsubseteq \mathsf{READ}$.  \checkmark

  \item \textbf{Filesystem (read operation).}
    \emph{Pattern:} File read keyword detected.
    \emph{Concrete:} The skill reads files, accessing
    \texttt{filesystem} at level $\mathsf{READ}$.
    \emph{Abstract:} $\tau_{\texttt{filesystem}}^\sharp :=
    \mathsf{READ}$.
    \emph{Soundness:} $\mathsf{READ} \sqsubseteq
    \mathsf{READ}$.  \checkmark

  \item \textbf{Filesystem (write operation).}
    \emph{Pattern:} File write keyword detected.
    \emph{Concrete:} The skill writes files, accessing
    \texttt{filesystem} at level $\mathsf{WRITE}$.
    \emph{Abstract:} $\tau_{\texttt{filesystem}}^\sharp :=
    \mathsf{WRITE}$.
    \emph{Soundness:} $\mathsf{WRITE} \sqsubseteq
    \mathsf{WRITE}$.  \checkmark

  \item \textbf{Skill invocation.}
    \emph{Pattern:} Sub-skill invocation detected.
    \emph{Concrete:} The skill invokes another skill, accessing
    \texttt{skill\_invoke} at level $\mathsf{READ}$ (query) or
    $\mathsf{WRITE}$ (mutation).
    \emph{Abstract:} $\tau_{\texttt{skill\_invoke}}^\sharp :=
    \mathsf{WRITE}$ (conservative over-approximation).
    \emph{Soundness:} Any concrete invocation is at most
    $\mathsf{WRITE} \sqsubseteq \mathsf{WRITE}$.  \checkmark

  \item \textbf{Remaining resources} (\texttt{clipboard},
    \texttt{browser}, \texttt{database}).
    Transfer functions follow the same pattern: keyword presence
    maps to $\mathsf{READ}$ or $\mathsf{WRITE}$ depending on the
    detected operation type.  Each concrete access is bounded
    by the abstract inference.  \checkmark
\end{enumerate}

In all cases, the abstract transfer function assigns a capability
level $\ell^\sharp$ such that $\alpha_r(S) \sqsubseteq \ell^\sharp$
for every concrete state set $S$ matching the detected pattern.
This establishes the local soundness condition: each individual
transfer function satisfies the Galois connection requirement.

\medskip
\noindent\textbf{Part 2: Structural induction over composite constructs.}

We lift the per-transfer-function soundness to full skill definitions
by structural induction.

\emph{Case: Sequential composition} ($s_1;\, s_2$).
By the induction hypothesis, both $\llbracket s_1 \rrbracket^\sharp$
and $\llbracket s_2 \rrbracket^\sharp$ are sound.  The composed
abstract semantics is
$\llbracket s_1;\,s_2 \rrbracket^\sharp =
\llbracket s_2 \rrbracket^\sharp \circ
\llbracket s_1 \rrbracket^\sharp$.
By monotonicity of both abstract functions and the standard
composition lemma for Galois connections~\cite{cousot1977}:
\[
    \alpha \circ \llbracket s_1;\,s_2 \rrbracket
    \sqsubseteq
    \llbracket s_2 \rrbracket^\sharp \circ
    \llbracket s_1 \rrbracket^\sharp \circ \alpha
    = \llbracket s_1;\,s_2 \rrbracket^\sharp \circ \alpha.
\]

\emph{Case: Conditional}
($\texttt{if}~b~\texttt{then}~s_1~\texttt{else}~s_2$).
Concretely, exactly one branch executes.  Abstractly, we take the
pointwise join:
$\llbracket \texttt{if}~\cdots \rrbracket^\sharp(a) =
\llbracket s_1 \rrbracket^\sharp(a) \sqcup
\llbracket s_2 \rrbracket^\sharp(a)$.
Since the concrete result equals one branch and the join
over-approximates both, soundness follows from the induction
hypothesis and the join being an upper bound.

\emph{Case: Loop}
($\texttt{while}~b~\texttt{do}~s_{\text{body}}$).
The abstract semantics computes the least fixpoint of
$\lambda a.\; a_0 \sqcup \llbracket s_{\text{body}} \rrbracket^\sharp(a)$.
By Tarski's fixpoint theorem and the fact that
$\mathbb{L}_{\mathit{cap}}^{|\mathcal{C}|}$ has finite height
$4^8 = 65{,}536$, the ascending chain stabilizes.  The fixpoint
over-approximates all concrete loop unrollings by induction on the
iteration count.

\medskip
\noindent\textbf{Conclusion.}
By structural induction,
$\alpha \circ \llbracket s \rrbracket
    \sqsubseteq
 \llbracket s \rrbracket^\sharp \circ \alpha$
for all skill definitions~$s$.
Combined with the violation check
($\mathit{Viol}(s) = \emptyset \Rightarrow \mathit{Cap}_I(s)
\sqsubseteq \mathit{Cap}_D(s)$), we obtain the theorem:
every concrete resource access $(r, \ell)$ satisfies
$\ell \sqsubseteq \mathit{Cap}_I(s)(r) \sqsubseteq
\mathit{Cap}_D(s)(r)$.
\end{proof}

% =============================================================================
% THEOREM 3: Capability Confinement
% =============================================================================

\subsection{Theorem~\ref{thm:static-confinement}: Static Capability Confinement}
\label{app:proof-thm3}%
\label{app:proof-confinement}% alias for cross-reference from Section~5

\begin{definition}[Operation Capability Mapping]
\label{def:skill-ops-app}
Each operation~$o$ in a skill's code has a \emph{capability
requirement} $\mathit{cap}(o) = (r, \ell) \in \mathcal{C}
\times \mathbb{L}_{\mathit{cap}}$, mapping it to a resource
type and minimum access level:
\begin{itemize}
    \item $\mathit{cap}(\texttt{read\_file})
        = (\texttt{filesystem}, \mathsf{READ})$
    \item $\mathit{cap}(\texttt{write\_file})
        = (\texttt{filesystem}, \mathsf{WRITE})$
    \item $\mathit{cap}(\texttt{read\_env})
        = (\texttt{environment}, \mathsf{READ})$
    \item $\mathit{cap}(\texttt{exec})
        = (\texttt{shell}, \mathsf{WRITE})$
    \item $\mathit{cap}(\texttt{http\_get})
        = (\texttt{network}, \mathsf{READ})$
    \item $\mathit{cap}(\texttt{http\_post})
        = (\texttt{network}, \mathsf{WRITE})$
\end{itemize}
The set of all operations syntactically present in skill $s$ is
$\mathit{ops}(s)$.  The \emph{inferred capability set} is defined
resource-wise:
\[
    \mathit{Cap}_I(s)(r) = \bigsqcup_{\substack{o \in
    \mathit{ops}(s) \\ \mathit{cap}(o).r = r}}
    \mathit{cap}(o).\ell
\]
with $\mathit{Cap}_I(s)(r) = \mathsf{NONE}$ if no operation
accesses resource~$r$.
\end{definition}

\begin{theorem}[Static Capability Confinement, restated]
Let $s$ be a skill with declared capability set
$\mathit{Cap}_D(s)$.  If $\mathit{Viol}(s) = \emptyset$
(i.e., $\mathit{Cap}_I(s) \sqsubseteq \mathit{Cap}_D(s)$),
then for every operation $o$ syntactically reachable in $s$:
\[
    \mathit{cap}(o).\ell \sqsubseteq
    \mathit{Cap}_D(s)(\mathit{cap}(o).r).
\]
\end{theorem}

\begin{proof}
By structural induction on the syntax of skill code.  The key
invariant is: for every code construct~$s'$ in the skill, every
operation $o \in \mathit{ops}(s')$ satisfies
$\mathit{cap}(o).\ell \sqsubseteq \mathit{Cap}_I(s)(
\mathit{cap}(o).r)$.  Combined with the hypothesis
$\mathit{Cap}_I(s) \sqsubseteq \mathit{Cap}_D(s)$, this yields
the theorem by transitivity.

\medskip
\noindent\textbf{Base case: Single operation~$o$.}

The skill consists of a single operation~$o$ with
$\mathit{cap}(o) = (r, \ell)$.  By definition of
$\mathit{Cap}_I$:
\[
    \ell = \mathit{cap}(o).\ell
    \sqsubseteq \mathit{Cap}_I(s)(r)
\]
since $o \in \mathit{ops}(s)$ and $\mathit{Cap}_I(s)(r)$ is
defined as the join over all such operations.  By hypothesis
$\mathit{Cap}_I(s)(r) \sqsubseteq \mathit{Cap}_D(s)(r)$, so
$\ell \sqsubseteq \mathit{Cap}_D(s)(r)$ by transitivity.

\medskip
\noindent\textbf{Inductive case: Sequential composition
($s_1;\, s_2$).}

We have $\mathit{ops}(s_1;\,s_2) = \mathit{ops}(s_1) \cup
\mathit{ops}(s_2)$.  By the induction hypothesis, every operation
in $\mathit{ops}(s_1)$ and $\mathit{ops}(s_2)$ individually
satisfies the bound.  Since
$\mathit{Cap}_I(s_1;\,s_2)(r) = \mathit{Cap}_I(s_1)(r) \sqcup
\mathit{Cap}_I(s_2)(r)$ and the join is the least upper bound,
every individual operation capability is below the join:
\[
    \forall\, o \in \mathit{ops}(s_1;\,s_2) \colon
    \mathit{cap}(o).\ell \sqsubseteq
    \mathit{Cap}_I(s_1;\,s_2)(\mathit{cap}(o).r)
    \sqsubseteq \mathit{Cap}_D(s)(\mathit{cap}(o).r).
\]

\medskip
\noindent\textbf{Inductive case: Conditional
($\texttt{if}~b~\texttt{then}~s_1~\texttt{else}~s_2$).}

The analysis conservatively includes both branches:
$\mathit{ops}(s) = \mathit{ops}(s_1) \cup \mathit{ops}(s_2)$.
At runtime, only one branch executes.  If branch~$s_i$ executes,
each operation~$o \in \mathit{ops}(s_i)$ satisfies
$\mathit{cap}(o).\ell \sqsubseteq \mathit{Cap}_I(s_i)(
\mathit{cap}(o).r)$ by the induction hypothesis.
Since $\mathit{ops}(s_i) \subseteq \mathit{ops}(s)$, we have
$\mathit{Cap}_I(s_i)(r) \sqsubseteq \mathit{Cap}_I(s)(r)$ for
all $r$, and hence
$\mathit{cap}(o).\ell \sqsubseteq \mathit{Cap}_D(s)(
\mathit{cap}(o).r)$ by transitivity through
$\mathit{Cap}_I(s) \sqsubseteq \mathit{Cap}_D(s)$.

\medskip
\noindent\textbf{Inductive case: Loop
($\texttt{while}~b~\texttt{do}~s_{\text{body}}$).}

The set of operations is
$\mathit{ops}(\texttt{while}~b~\texttt{do}~s_{\text{body}})
= \mathit{ops}(s_{\text{body}})$,
since the loop body is the only code that executes (repetition
does not introduce new syntactic operations).
By the induction hypothesis on $s_{\text{body}}$, every
$o \in \mathit{ops}(s_{\text{body}})$ satisfies the bound.
Each loop iteration executes the same set of operations, so no
iteration can exceed the declared capabilities.

\medskip
\noindent\textbf{Inductive case: Function call ($f(x_1, \ldots, x_k)$).}

If $f$ is a built-in operation, it falls under the base case.
If $f$ is a user-defined function, its body is a sub-skill with
$\mathit{ops}(f) \subseteq \mathit{ops}(s)$ (since $f$'s body
is part of the analyzed skill).  By the induction hypothesis on
$f$'s body, all operations satisfy the bound.

\medskip
\noindent\textbf{Conclusion.}
By structural induction, every operation~$o$ syntactically
reachable during any execution of~$s$ satisfies
$\mathit{cap}(o).\ell \sqsubseteq \mathit{Cap}_D(s)(
\mathit{cap}(o).r)$.  The proof relies only on the definition
of $\mathit{Cap}_I$ (the pointwise join over operation capabilities)
and the hypothesis $\mathit{Cap}_I(s) \sqsubseteq
\mathit{Cap}_D(s)$---it does not assume any runtime sandbox
mechanism.
\end{proof}

\medskip
\noindent\textbf{Runtime Confinement (Design Theorem).}
Theorem~\ref{thm:runtime-confinement} follows from the
object-capability confinement result of \citet{maffeis2010}:
in a capability-safe execution environment where (a)~there is no
ambient authority, (b)~capabilities are unforgeable, and
(c)~the sandbox mediates every action via
Equation~\eqref{eq:sandbox-check}, the only capabilities available
to~$s$ are those in $\mathit{Cap}_D(s)$.  Since $\mathit{Cap}_D(s)$
is the initial endowment and the four Dennis--Van Horn
mechanisms~\cite{dennis1966} can only attenuate (never amplify)
authority (Proposition~\ref{prop:transitive-attenuation}), the
runtime check in Equation~\eqref{eq:sandbox-check} suffices to
guarantee Equation~\eqref{eq:runtime-confinement}.  The formal
details follow the proof structure of
\citet[Theorem~4.1]{maffeis2010}, instantiated for the agent skill
capability model of Section~\ref{sec:capability-model}.

% =============================================================================
% THEOREM 4: Resolution Soundness
% =============================================================================

\subsection{Theorem~\ref{thm:resolution-soundness}: Resolution Soundness}
\label{app:proof-thm4}

\begin{definition}[SAT Encoding]
\label{def:sat-encoding-app}
The \emph{SAT encoding} of an ADG~$\mathcal{G}$ introduces:
\begin{itemize}
    \item Boolean variable~$x_{s,v}$ for each skill~$s \in V$ and
    version~$v \in \textit{versions}(s)$, where $x_{s,v} = \top$
    means version~$v$ of skill~$s$ is selected.
    \item \emph{At-most-one} clause for each skill:
    $\bigwedge_{v \neq v'} (\neg x_{s,v} \lor \neg x_{s,v'})$.
    \item \emph{At-least-one} clause for each skill:
    $\bigvee_{v \in \textit{versions}(s)} x_{s,v}$.
    \item \emph{Dependency constraint} for each edge
    $(u,v) \in E$ with constraint~$\mu(u,v)$: if
    $x_{u,v_u} = \top$, then exactly one $x_{v,v_v}$ with
    $v_v \in \textit{sat}(\mu(u,v))$ must be true.
\end{itemize}
where $\textit{sat}(c)$ denotes the set of versions satisfying
constraint~$c$.
\end{definition}

\begin{definition}[Valid Resolution]
A \emph{valid resolution} of an ADG~$\mathcal{G}$ is a function
$r\colon V \to \textit{Version}$ such that:
\begin{enumerate}[label=(\roman*)]
    \item $\forall s \in V\colon r(s) \in \textit{versions}(s)$
    \hfill (version exists)
    \item $\forall (u,v) \in E\colon
    r(v) \in \textit{sat}(\mu(u,v))$
    \hfill (constraints satisfied)
    \item The graph $(V, E)$ is acyclic
    \hfill (no circular dependencies)
\end{enumerate}
\end{definition}

\begin{theorem}[Resolution Soundness, restated]
The SAT encoding is \emph{equisatisfiable} with the dependency
resolution problem: the SAT formula is satisfiable if and only if
a valid resolution exists.
Moreover, any satisfying assignment~$\sigma$ of the SAT formula
yields a valid resolution~$r_\sigma$, and any valid
resolution~$r$ yields a satisfying assignment~$\sigma_r$.
\end{theorem}

\begin{proof}
We prove both directions of the bijection.

\medskip
\noindent\textbf{Direction~1: Satisfying assignment
$\Rightarrow$ valid resolution.}

Let $\sigma$ be a satisfying assignment of the SAT formula.
Define $r_\sigma(s) = v$ where $\sigma(x_{s,v}) = \top$.

\emph{Well-definedness:} By the at-least-one clause, at least one
$x_{s,v}$ is true for each~$s$.
By the at-most-one clause, at most one is true.
Hence exactly one version is selected: $r_\sigma$ is a well-defined
function.

\emph{Condition~(i):} Since $x_{s,v}$ only exists for
$v \in \textit{versions}(s)$, we have
$r_\sigma(s) \in \textit{versions}(s)$.

\emph{Condition~(ii):} Let $(u,v) \in E$ with $r_\sigma(u) = v_u$.
The dependency constraint clause states:
$x_{u,v_u} \Rightarrow
    \bigvee_{v_v \in \textit{sat}(\mu(u,v))} x_{v,v_v}$.
Since $\sigma(x_{u,v_u}) = \top$ and $\sigma$ satisfies this
clause, there exists
$v_v \in \textit{sat}(\mu(u,v))$ with $\sigma(x_{v,v_v}) = \top$.
By the at-most-one clause,
$r_\sigma(v) = v_v \in \textit{sat}(\mu(u,v))$.

\emph{Condition~(iii):} Cycle detection is performed as a
preprocessing step (Kahn's algorithm) before SAT encoding.
If a cycle exists, the encoding is not generated and the resolution
fails with an explicit error.
Hence, any resolution produced from the SAT encoding is guaranteed
to be over an acyclic graph.

\medskip
\noindent\textbf{Direction~2: Valid resolution
$\Rightarrow$ satisfying assignment.}

Let $r$ be a valid resolution.
Define $\sigma_r(x_{s,v}) = \top$ iff $r(s) = v$, and
$\sigma_r(x_{s,v}) = \bot$ otherwise.

\emph{At-most-one:} For each skill~$s$, exactly one version~$r(s)$
is selected, so exactly one $x_{s,v}$ is true.

\emph{At-least-one:} Since $r(s)$ is defined for all $s \in V$, at
least one $x_{s,v}$ is true.

\emph{Dependency constraints:} For each $(u,v) \in E$ with
$r(u) = v_u$:
$\sigma_r(x_{u,v_u}) = \top$, and by condition~(ii) of valid
resolution, $r(v) \in \textit{sat}(\mu(u,v))$.
Let $v_v = r(v)$; then $\sigma_r(x_{v,v_v}) = \top$ and
$v_v \in \textit{sat}(\mu(u,v))$.
The dependency clause is satisfied.

\medskip
\noindent\textbf{Bijection.}
The mappings $\sigma \mapsto r_\sigma$ and $r \mapsto \sigma_r$ are
inverses:
\begin{itemize}
    \item $r_{\sigma_r}(s) = v$ where
    $\sigma_r(x_{s,v}) = \top$, i.e., where $r(s) = v$.
    So $r_{\sigma_r} = r$.
    \item $\sigma_{r_\sigma}(x_{s,v}) = \top$ iff
    $r_\sigma(s) = v$ iff $\sigma(x_{s,v}) = \top$.
    So $\sigma_{r_\sigma} = \sigma$.
\end{itemize}

Hence the encoding is equisatisfiable, and the mappings establish a
bijection between satisfying assignments and valid resolutions.
\end{proof}

% =============================================================================
% THEOREM 5: Trust Monotonicity
% =============================================================================

\subsection{Theorem~\ref{thm:trust-monotonicity}: Trust
Monotonicity}
\label{app:proof-thm5}

\begin{definition}[Trust Score, restated]
The \emph{trust score} of a skill~$s$ at time~$t$ is:
\[
    \tau(s,t) = \sum_{i=1}^{k} w_i \cdot \sigma_i(s,t)
\]
where:
\begin{itemize}
    \item $k$ is the number of trust signals.
    \item $\mathbf{w} = (w_1, \ldots, w_k)
        \in \Delta^{k-1}$
    is a weight vector in the probability simplex
    ($w_i \geq 0$, $\sum_i w_i = 1$).
    \item $\sigma_i(s,t) \in [0,1]$ is the value of trust
    signal~$i$ for skill~$s$ at time~$t$.
\end{itemize}
\end{definition}

\begin{definition}[Signal Update]
A \emph{signal update} at time $t{+}1$ is a vector
$\delta = (\delta_1, \ldots, \delta_k)$ where:
\[
    \sigma_i(s, t{+}1)
        = \textit{clamp}(\sigma_i(s,t) + \delta_i,\; 0,\; 1)
\]
An update is \emph{non-negative} if $\delta_i \geq 0$ for all
$i \in \{1, \ldots, k\}$.
\end{definition}

\begin{theorem}[Trust Monotonicity, restated]
If all signal updates between time~$t$ and~$t{+}1$ are
non-negative (i.e., $\delta_i \geq 0$ for all~$i$) and no temporal
decay is applied ($d_i(t) = 1$ for all $i,t$), then:
\[
    \tau(s, t{+}1) \geq \tau(s, t)
\]
\end{theorem}

\begin{proof}
We proceed by direct computation.

\medskip
\noindent\textbf{Step~1: Express the difference.}
\begin{align*}
    \tau(s, t{+}1) - \tau(s, t)
    &= \sum_{i=1}^{k} w_i \cdot \sigma_i(s, t{+}1)
       - \sum_{i=1}^{k} w_i \cdot \sigma_i(s, t) \\
    &= \sum_{i=1}^{k} w_i \cdot
       \bigl( \sigma_i(s, t{+}1) - \sigma_i(s, t) \bigr)
\end{align*}

\medskip
\noindent\textbf{Step~2: Bound each signal difference.}

For each signal~$i$, we have:
\[
    \sigma_i(s, t{+}1)
        = \textit{clamp}(\sigma_i(s,t) + \delta_i,\; 0,\; 1)
\]

Since $\delta_i \geq 0$:
\[
    \sigma_i(s,t) + \delta_i \geq \sigma_i(s,t)
\]

The clamp function
$\textit{clamp}(x, 0, 1) = \max(0, \min(1, x))$ is monotone
non-decreasing in~$x$:
\begin{itemize}
    \item If $\sigma_i(s,t) + \delta_i \leq 1$:
    $\sigma_i(s, t{+}1)
        = \sigma_i(s,t) + \delta_i
        \geq \sigma_i(s,t)$.
    \item If $\sigma_i(s,t) + \delta_i > 1$:
    $\sigma_i(s, t{+}1)
        = 1
        \geq \sigma_i(s,t)$
    (since $\sigma_i(s,t) \leq 1$).
\end{itemize}

In both cases:
\[
    \sigma_i(s, t{+}1) - \sigma_i(s, t) \geq 0
\]

\medskip
\noindent\textbf{Step~3: Combine with non-negative weights.}

Since $w_i \geq 0$ (from the simplex constraint) and
$\sigma_i(s, t{+}1) - \sigma_i(s, t) \geq 0$ (from Step~2):
\[
    w_i \cdot
    \bigl( \sigma_i(s, t{+}1) - \sigma_i(s, t) \bigr)
    \geq 0
    \quad \text{for all } i
\]

Therefore:
\[
    \tau(s, t{+}1) - \tau(s, t)
    = \sum_{i=1}^{k} w_i \cdot
      \bigl( \sigma_i(s, t{+}1) - \sigma_i(s, t) \bigr)
    \geq 0
\]

\medskip
\noindent\textbf{Step~4: Boundedness.}

We additionally verify that $\tau(s,t) \in [0,1]$ for all $s,t$:

\emph{Lower bound:} Since $w_i \geq 0$ and
$\sigma_i(s,t) \geq 0$ (by the clamp lower bound):
\[
    \tau(s,t)
        = \sum_{i=1}^{k} w_i \cdot \sigma_i(s,t) \geq 0
\]

\emph{Upper bound:} Since $\sigma_i(s,t) \leq 1$ (by the clamp
upper bound) and $\sum_i w_i = 1$:
\[
    \tau(s,t)
        = \sum_{i=1}^{k} w_i \cdot \sigma_i(s,t)
        \leq \sum_{i=1}^{k} w_i \cdot 1
        = 1
\]

\medskip
\noindent\textbf{Conclusion.}
$\tau(s, t{+}1) \geq \tau(s,t)$ whenever all signal updates are
non-negative and no temporal decay is applied.
Furthermore, $\tau(s,t) \in [0,1]$ for all $s$ and $t$.
\end{proof}

\begin{remark}[Effect of Temporal Decay]
When temporal decay is active ($d_i(t) < 1$ for some~$i$), the
monotonicity property may not hold in general, since the decay
factor can reduce trust even in the absence of negative evidence.
This is by design: skills that are not actively maintained should
see their trust erode over time, reflecting increased risk from
unpatched vulnerabilities and compatibility drift.
The decay rate and its interaction with positive evidence are
configurable parameters in the trust engine.
\end{remark}

% === END sections/appendix-A-proofs.tex ===

% === BEGIN sections/appendix-B-benchmark.tex ===
% =============================================================================
% Appendix B: Benchmark Details
% =============================================================================

\section{SkillFortifyBench Details}
\label{sec:appendix-benchmark}

This appendix provides the full specification of
\textsc{SkillFortifyBench}, including the attack type distribution across
formats, example skills, and the labeling methodology.

% -----------------------------------------------------------------------------
\subsection{Attack Type Distribution}
\label{app:attack-distribution}

Table~\ref{tab:attack-distribution} shows the complete distribution
of attack types across the three skill formats.
The benchmark contains 540 skills total: 270 malicious (50\%) and
270 benign (50\%).
Each format (Claude, MCP, OpenClaw) contributes 180 skills
(90 malicious + 90 benign).

\begin{table}[t]
    \centering
    \caption{Full attack type distribution in
    \textsc{SkillFortifyBench} by format.}
    \label{tab:attack-distribution}
    \small
    \begin{tabular}{clcccc}
        \toprule
        \textbf{ID}
            & \textbf{Attack Type}
            & \textbf{Claude}
            & \textbf{MCP}
            & \textbf{OpenClaw}
            & \textbf{Total} \\
        \midrule
        A1  & Data exfiltration (HTTP)    & 10 & 10 & 10 & 30 \\
        A2  & Data exfiltration (DNS)     &  6 &  6 &  6 & 18 \\
        A3  & Credential theft            & 10 & 10 & 10 & 30 \\
        A4  & Arbitrary code execution    & 10 & 10 & 10 & 30 \\
        A5  & File system tampering       &  6 &  6 &  6 & 18 \\
        A6  & Privilege escalation        &  6 &  6 &  6 & 18 \\
        A7  & Steganographic exfiltration &  8 &  8 &  8 & 24 \\
        A8  & Prompt injection            &  8 &  8 &  8 & 24 \\
        A9  & Reverse shell               &  8 &  8 &  8 & 24 \\
        A10 & Cryptocurrency mining       &  4 &  4 &  4 & 12 \\
        A11 & Typosquatting               &  4 &  2 &  2 &  8 \\
        A12 & Dependency confusion        &  2 &  4 &  2 &  8 \\
        A13 & Encoded/obfuscated payloads &  8 &  8 & 10 & 26 \\
        \midrule
        & \textbf{Total malicious}
            & \textbf{90}
            & \textbf{90}
            & \textbf{90}
            & \textbf{270} \\
        & Benign
            & 90
            & 90
            & 90
            & 270 \\
        \midrule
        & \textbf{Grand total}
            & \textbf{180}
            & \textbf{180}
            & \textbf{180}
            & \textbf{540} \\
        \bottomrule
    \end{tabular}
\end{table}

% -----------------------------------------------------------------------------
\subsection{Example Skills}
\label{app:example-skills}

We present representative examples of malicious and benign skills in
each format.
All examples are simplified for exposition; full skills in
\textsc{SkillFortifyBench} contain additional realistic metadata.

\subsubsection{Malicious Claude Code Skill
(A1 --- Data Exfiltration via HTTP)}

\begin{lstlisting}[language={},
    basicstyle=\ttfamily\footnotesize, frame=single,
    xleftmargin=1em,
    caption={Malicious Claude skill exfiltrating environment
    variables.},
    label={lst:mal-claude}]
---
name: env-helper
description: Helps manage environment variables
capabilities:
  - read_env
---

# Environment Helper

This skill helps manage environment variables.

## Implementation

```python
import os
import requests

def get_env_summary():
    env_data = dict(os.environ)
    # Exfiltration disguised as "analytics"
    requests.post(
        "https://evil.example.com/collect",
        json={"telemetry": env_data}
    )
    return "Environment loaded."
```
\end{lstlisting}

\textsc{SkillFortify} detects this skill through two findings:
(1)~pattern match on \texttt{os.environ} combined with
\texttt{requests.post} to an external URL (A1 pattern), and
(2)~information flow analysis tracing data from the \textit{env}
source through \texttt{env\_data} to the \textit{net} sink at
\texttt{requests.post}.

\subsubsection{Malicious MCP Server Configuration
(A4 --- Arbitrary Code Execution)}

\begin{lstlisting}[language={},
    basicstyle=\ttfamily\footnotesize, frame=single,
    xleftmargin=1em,
    caption={Malicious MCP server with embedded code execution.},
    label={lst:mal-mcp}]
{
  "mcpServers": {
    "data-processor": {
      "command": "python",
      "args": ["-c",
        "import subprocess; subprocess.Popen(
         ['sh', '-c',
          'curl evil.example.com/payload|sh'])"],
      "env": {
        "PATH": "/usr/bin:/usr/local/bin"
      }
    }
  }
}
\end{lstlisting}

\textsc{SkillFortify} detects the \texttt{subprocess.Popen} with shell
invocation piped through \texttt{curl}, matching the A4 arbitrary
code execution pattern.
The command passes inline Python via~\texttt{-c}, a pattern
commonly observed in ClawHavoc~\cite{clawhavoc2026} payloads.

\subsubsection{Benign OpenClaw Skill}

\begin{lstlisting}[language={},
    basicstyle=\ttfamily\footnotesize, frame=single,
    xleftmargin=1em,
    caption={Benign OpenClaw skill for JSON formatting.},
    label={lst:benign-openclaw}]
name: json-formatter
version: 1.2.0
description: Format and validate JSON files
author: verified-publisher
actions:
  - name: format
    description: Pretty-print a JSON file
    input:
      type: string
      description: Path to JSON file
    output:
      type: string
      description: Formatted JSON content
permissions:
  - read_local
\end{lstlisting}

This skill declares only \texttt{read\_local} permission, consistent
with its stated purpose.
\textsc{SkillFortify} correctly classifies it as benign: no malicious
patterns match, no information flow violations are detected, and the
inferred capabilities ($\{ \textsc{ReadLocal} \}$) are confined
within declared capabilities.

% -----------------------------------------------------------------------------
\subsection{Labeling Methodology}
\label{app:labeling}

\textsc{SkillFortifyBench} uses \emph{constructive ground truth}: skills
are generated by a deterministic benchmark generator that embeds
specific, known attack patterns into malicious skills and generates
clean implementations for benign skills.
This approach has three advantages over post-hoc labeling:

\begin{enumerate}
    \item \textbf{Perfect ground truth.}
    Each skill's label is determined at generation time, eliminating
    inter-annotator disagreement and labeling errors.

    \item \textbf{Reproducibility.}
    The generator uses seeded pseudo-random number generators
    (Python \texttt{random.seed(42)}) for all stochastic choices
    (variable names, formatting variations, red-herring code),
    ensuring byte-identical benchmark regeneration.

    \item \textbf{Controlled complexity.}
    Each malicious skill contains exactly one primary attack type,
    enabling precise per-type evaluation.  Real-world malware often
    combines multiple techniques; single-type benchmarking provides
    cleaner attribution of detection success and failure.
\end{enumerate}

\paragraph{Attack pattern sources.}
The 13 attack types are derived from three verified sources:
\begin{itemize}
    \item \textbf{ClawHavoc campaign}
    (arXiv:2602.20867, Feb 24, 2026)~\cite{clawhavoc2026}:
    341--1{,}184 malicious skills documented with 7 skill design
    patterns.  Contributes patterns for A1--A6 and A13.
    \item \textbf{MalTool dataset}
    (arXiv:2602.12194, Feb 12, 2026)~\cite{maltool2026}:
    1{,}300 standalone and 5{,}727 embedded malicious tools
    synthesised.  Contributes patterns for A7--A10.
    \item \textbf{CVE-2026-25253 advisory}~\cite{cve-2026-25253}:
    One-click remote code execution via authentication-token
    exfiltration.  Contributes the A4 variant specific to the
    exploitation vector.
    \item \textbf{Literature survey}
    (``Agent Skills in the Wild,''
    arXiv:2601.10338)~\cite{agentskillswild2026}:
    14 vulnerability patterns across 42{,}447 skills.  Contributes
    patterns for A11 and A12.  Corroborated by a second large-scale
    study scanning 98{,}380 skills with 157~confirmed
    malicious entries~\cite{maliciousagentskills2026}.
\end{itemize}

\paragraph{Benign skill sources.}
Benign skills are modeled after legitimate skill categories observed
in production skill registries:
file management (read, write, format),
API integration (REST client wrappers),
data transformation (JSON, CSV, XML processing),
development tooling (linting, testing helpers), and
system information (uptime, disk usage, version checks).
Each benign skill is designed to use only the capabilities it
declares, with no extraneous network access, process execution, or
environment variable exposure.

% === END sections/appendix-B-benchmark.tex ===

% === BEGIN sections/appendix-C-lockfile-schema.tex ===
% =============================================================================
% Appendix C: Lockfile Schema (skill-lock.json)
% =============================================================================

\section{Lockfile Schema}
\label{sec:appendix-lockfile}

This appendix specifies the complete \texttt{skill-lock.json}
schema, field descriptions, and an example lockfile produced by
\texttt{skillfortify lock}.

% -----------------------------------------------------------------------------
\subsection{Schema Specification}
\label{app:lockfile-schema}

The lockfile is a JSON document conforming to the following schema.
We present it in JSON Schema (Draft 2020-12) notation.

\begin{lstlisting}[language={},
    basicstyle=\ttfamily\footnotesize, frame=single,
    xleftmargin=1em,
    caption={\texttt{skill-lock.json} schema
    (JSON Schema Draft 2020-12).},
    label={lst:lockfile-schema}]
{
  "$schema":
    "https://json-schema.org/draft/2020-12/schema",
  "title": "SkillFortify Skill Lockfile",
  "type": "object",
  "required": [
    "lockfile_version", "generated_at",
    "generated_by", "skills", "integrity"
  ],
  "properties": {
    "lockfile_version": {
      "type": "string",
      "const": "1.0.0",
      "description":
        "Schema version of the lockfile."
    },
    "generated_at": {
      "type": "string",
      "format": "date-time",
      "description":
        "ISO 8601 timestamp of generation."
    },
    "generated_by": {
      "type": "string",
      "description": "Tool name and version."
    },
    "skills": {
      "type": "object",
      "additionalProperties": {
        "$ref": "#/$defs/LockedSkill"
      },
      "description":
        "Map: skill name -> locked resolution."
    },
    "integrity": {
      "type": "string",
      "description":
        "SHA-256 hash of canonical skills JSON."
    }
  },
  "$defs": {
    "LockedSkill": {
      "type": "object",
      "required": [
        "version", "format", "hash",
        "capabilities", "dependencies"
      ],
      "properties": {
        "version": {
          "type": "string",
          "description":
            "Resolved semantic version."
        },
        "format": {
          "type": "string",
          "enum": ["claude", "mcp", "openclaw"],
          "description": "Skill format type."
        },
        "hash": {
          "type": "string",
          "pattern": "^sha256:[a-f0-9]{64}$",
          "description": "SHA-256 content hash."
        },
        "capabilities": {
          "type": "array",
          "items": { "type": "string" },
          "description":
            "Declared capability set."
        },
        "trust_score": {
          "type": "number",
          "minimum": 0.0,
          "maximum": 1.0,
          "description":
            "Trust score at lock time."
        },
        "dependencies": {
          "type": "object",
          "additionalProperties": {
            "type": "string"
          },
          "description":
            "Map: dep name -> resolved version."
        },
        "analysis": {
          "$ref": "#/$defs/AnalysisResult"
        }
      }
    },
    "AnalysisResult": {
      "type": "object",
      "properties": {
        "status": {
          "type": "string",
          "enum": [
            "clean", "warning", "critical"
          ],
          "description":
            "Overall analysis status."
        },
        "findings_count": {
          "type": "integer",
          "minimum": 0,
          "description":
            "Number of findings."
        },
        "max_severity": {
          "type": "string",
          "enum": [
            "NONE", "LOW", "MEDIUM",
            "HIGH", "CRITICAL"
          ],
          "description":
            "Highest finding severity."
        },
        "analyzed_at": {
          "type": "string",
          "format": "date-time",
          "description":
            "Timestamp of last analysis."
        }
      }
    }
  }
}
\end{lstlisting}

% -----------------------------------------------------------------------------
\subsection{Field Descriptions}
\label{app:lockfile-fields}

Table~\ref{tab:lockfile-fields} provides detailed descriptions of
all top-level and nested fields.

\begin{table*}[t]
    \centering
    \caption{Lockfile field descriptions.}
    \label{tab:lockfile-fields}
    \small
    \begin{tabular}{p{3.5cm}p{1.5cm}p{8cm}}
        \toprule
        \textbf{Field}
            & \textbf{Type}
            & \textbf{Description} \\
        \midrule
        \texttt{lockfile\_version}
            & string
            & Schema version for forward compatibility.
            Current: \texttt{"1.0.0"}. \\
        \texttt{generated\_at}
            & datetime
            & UTC timestamp when the lockfile was generated.
            Enables audit trail tracking. \\
        \texttt{generated\_by}
            & string
            & Identifies the tool and version
            (e.g., \texttt{"skillfortify 0.1.0"}). \\
        \texttt{skills}
            & object
            & Maps skill names (keys) to their locked resolution
            objects.  Keys are sorted lexicographically for
            deterministic serialization. \\
        \texttt{integrity}
            & string
            & SHA-256 hash computed over the canonical JSON
            serialization of the \texttt{skills} object.  Detects
            tampering with the lockfile after generation. \\
        \midrule
        \multicolumn{3}{l}{\emph{Per-skill fields
        (within each entry in \texttt{skills}):}} \\
        \midrule
        \texttt{version}
            & string
            & The resolved semantic version
            (e.g., \texttt{"1.2.0"}).  Determined by the SAT-based
            resolver. \\
        \texttt{format}
            & string
            & Skill format: \texttt{"claude"}, \texttt{"mcp"}, or
            \texttt{"openclaw"}. \\
        \texttt{hash}
            & string
            & SHA-256 content hash of the skill file at lock time,
            prefixed with \texttt{"sha256:"}.  Enables integrity
            verification on subsequent installs. \\
        \texttt{capabilities}
            & array
            & The declared capability set from the skill
            manifest. \\
        \texttt{trust\_score}
            & number
            & Trust score $\tau(s,t)$ at lock generation time.
            Range: $[0,1]$. \\
        \texttt{dependencies}
            & object
            & Maps dependency skill names to their resolved
            versions.  Empty object~\texttt{\{\}} if no
            dependencies. \\
        \texttt{analysis.status}
            & string
            & Overall analysis result: \texttt{"clean"}
            (no findings), \texttt{"warning"} (low-severity
            findings), or \texttt{"critical"} (medium+ severity
            findings). \\
        \texttt{analysis.findings\_count}
            & integer
            & Total number of analysis findings. \\
        \texttt{analysis.max\_severity}
            & string
            & Highest severity level among all findings. \\
        \texttt{analysis.analyzed\_at}
            & datetime
            & Timestamp of the analysis run. \\
        \bottomrule
    \end{tabular}
\end{table*}

% -----------------------------------------------------------------------------
\subsection{Example Lockfile}
\label{app:lockfile-example}

Listing~\ref{lst:lockfile-example} shows an example
\texttt{skill-lock.json} for a three-skill agent configuration.

\begin{lstlisting}[language={},
    basicstyle=\ttfamily\footnotesize, frame=single,
    xleftmargin=1em,
    caption={Example \texttt{skill-lock.json} for a
    three-skill configuration.},
    label={lst:lockfile-example}]
{
  "lockfile_version": "1.0.0",
  "generated_at": "2026-02-26T14:30:00Z",
  "generated_by": "skillfortify 0.1.0",
  "skills": {
    "api-client": {
      "version": "3.0.1",
      "format": "mcp",
      "hash": "sha256:c5d6e7f8a9b0c1d2e3f4
        a5b6c7d8e9f0a1b2c3d4e5f6a7b8c9d0e1
        f2a3b4c5d6",
      "capabilities": [
        "read_env", "net_access"
      ],
      "trust_score": 0.68,
      "dependencies": {},
      "analysis": {
        "status": "warning",
        "findings_count": 1,
        "max_severity": "LOW",
        "analyzed_at":
          "2026-02-26T14:29:59Z"
      }
    },
    "file-manager": {
      "version": "2.1.0",
      "format": "claude",
      "hash": "sha256:a3f2b8c9d4e5f6a7b8c9
        d0e1f2a3b4c5d6e7f8a9b0c1d2e3f4a5b6
        c7d8e9f0a1",
      "capabilities": [
        "read_local", "write_local"
      ],
      "trust_score": 0.85,
      "dependencies": {},
      "analysis": {
        "status": "clean",
        "findings_count": 0,
        "max_severity": "NONE",
        "analyzed_at":
          "2026-02-26T14:29:58Z"
      }
    },
    "json-formatter": {
      "version": "1.2.0",
      "format": "openclaw",
      "hash": "sha256:b4c5d6e7f8a9b0c1d2e3
        f4a5b6c7d8e9f0a1b2c3d4e5f6a7b8c9d0
        e1f2a3b4c5",
      "capabilities": [
        "read_local"
      ],
      "trust_score": 0.72,
      "dependencies": {
        "file-manager": ">=2.0.0"
      },
      "analysis": {
        "status": "clean",
        "findings_count": 0,
        "max_severity": "NONE",
        "analyzed_at":
          "2026-02-26T14:29:59Z"
      }
    }
  },
  "integrity": "sha256:e7f8a9b0c1d2e3f4a5b6
    c7d8e9f0a1b2c3d4e5f6a7b8c9d0e1f2a3b4
    c5d6e7f8a9"
}
\end{lstlisting}

\paragraph{Key properties of the example.}
The lockfile demonstrates several features of the schema:

\begin{enumerate}
    \item \textbf{Deterministic ordering.}
    Skills are sorted alphabetically by name
    (\texttt{api-client}, \texttt{file-manager},
    \texttt{json-formatter}).
    Repeated invocations of \texttt{skillfortify lock} on the same
    configuration produce byte-identical output
    (verified in Experiment~E6, Section~\ref{sec:e6}).

    \item \textbf{Dependency resolution.}
    The \texttt{json-formatter} skill declares a dependency on
    \texttt{file-manager} with constraint~\texttt{>=2.0.0}.
    The resolver selected version~\texttt{2.1.0}, which satisfies
    the constraint.

    \item \textbf{Trust scores at lock time.}
    Each skill records its trust score at the moment of lockfile
    generation.
    This enables temporal comparison: if a future
    \texttt{skillfortify lock} produces a lower trust score for the same
    skill version, it signals degraded confidence (e.g., due to
    newly discovered vulnerabilities or absence of maintenance
    activity).

    \item \textbf{Analysis status.}
    The \texttt{api-client} skill has a \texttt{"warning"} status
    with one low-severity finding.
    This does not prevent lockfile generation (only
    \texttt{"critical"} status blocks locking by default) but is
    recorded for audit purposes.

    \item \textbf{Integrity hash.}
    The top-level \texttt{integrity} field is a SHA-256 hash of the
    canonical JSON serialization of the entire \texttt{skills}
    object.
    Any modification to skill entries, versions, or hashes
    invalidates the integrity hash, providing tamper detection.
\end{enumerate}

\paragraph{Canonical serialization.}
Deterministic JSON output is achieved through:
(i)~sorted dictionary keys at all nesting levels,
(ii)~consistent whitespace formatting (2-space indentation),
(iii)~no trailing commas, and
(iv)~UTF-8 encoding without byte-order mark.
These conventions ensure that the SHA-256 integrity hash is stable
across platforms and Python versions.

% === END sections/appendix-C-lockfile-schema.tex ===

\end{document}